\documentclass [11pt]{article} 
\usepackage{amsmath,amsthm,amsfonts,amscd,eucal,latexsym,amssymb,bm,amsbsy,mathrsfs,yhmath,manfnt} 

\usepackage{graphicx}

\def\sp{\hskip -5pt}
\def\spa{\hskip -3pt}

\def\emptyset{\varnothing} 

\def\bearray{\begin{eqnarray}}
\def\earray{\end{eqnarray}}
\def\beq{\begin{equation}}
\def\eeq{\end{equation}}

\def\b0{{\bf 0}}

\def\mpasto{\mapsto}

\def\det{\mbox{det}}

\def\gor{{\bf g}}

\def\cB{\mathscr{B}}
\def\cC{\mathscr{C}}
\def\cD{{\cal D}}

\def\cF{\mathscr{F}}

\def\cK{{\cal K}}

\def\cL{\mathscr{L}} 
\def\cH{{\cal H}} 
\def\cM{\mathscr{M}}
  
\def\cP{\mathscr{P}}

\def\cS{\mathscr{S}}

\newsymbol\bt 1202             

\def\bC{{\mathbb C}}           

\def\bM{{\mathbb M}}
\def\bN{{\mathbb N}}

\def\bR{{\mathbb R}}

\def\sA{{\mathsf A}}

\def\sK{{\mathsf K}}
\def\sM{{\mathsf M}}

\def\sH{{\mathsf H}}

\def\sP{{\mathsf P}}
\def\sQ{{\mathsf Q}}

\def\sV{\mathsf{V}}
\def\sT{{\mathsf T}}
\def\sE{{\mathsf E}}

\newsymbol\rest 1316         
\def\gB{{\mathfrak B}}

\def\gF{{\mathfrak F}}

\theoremstyle{TheoremStyle}
\newtheorem{theorem}{Theorem}
\newtheorem{corollary}[theorem]{Corollary}
\newtheorem{proposition}[theorem]{Proposition}
\newtheorem{lemma}[theorem]{Lemma}
\newtheorem{definition}[theorem]{Definition}
\newtheorem{remark}[theorem]{Remark}

\begin{document} 
\hfill{\sl  April   2024 (revised version)} 
\par 
\bigskip 
\par 
\rm


\par
\bigskip
\large
\noindent
{\bf  Quantum particle localization observables  on Cauchy surfaces of Minkowski spacetime and their causal properties}
\bigskip
\par
\rm
\normalsize 


\noindent  {\bf Carmine De Rosa$^a$   and Valter Moretti$^{b}$ }\\
\par

\noindent 
  Department of  Mathematics, University of Trento, and INFN-TIFPA \\
 via Sommarive 14, I-38123  Povo (Trento), Italy.\\
  $^a$carmine.derosa@unitn.it\\
 $^b$valter.moretti@unitn.it\\

 \normalsize

\par

\rm\normalsize

\rm\normalsize


\par

\begin{abstract}
We introduce and study  a general notion of spatial localization on spacelike smooth Cauchy surfaces  of quantum systems in Minkowski spacetime. The notion is constructed in terms of a coherent family of normalized POVMs, one for each said Cauchy surface. We prove that a family of POVMs of this type automatically satisfies a causality condition which generalizes Castrigiano's one and implies it when restricting to flat spacelike Cauchy surfaces. As a consequence no  conflict with Hegerfeldt's theorem arises.
We furthermore prove that such families of POVMs do exist for massive  Klein-Gordon particles,
since some of them  are extensions of  already known spatial localization observables. These are constructed  out of positive definite  kernels or are defined in terms of  the stress-energy tensor operator. Some further features of these structures are investigated, in particular,  the relation with the triple of Newton-Wigner selfadjoint operators and a modified form of Heisenberg inequality in the rest $3$-spaces of Minkowski reference frames. 
\end{abstract}
\tableofcontents

\section{Introduction}

\subsection{The subtle issue of spatial localization of quantum relativistic systems}
The study of  notions of {\em spatial localization} at given time for a quantum relativistic particle can be traced back to the seminal paper by  Newton and Wigner \cite{NW}.
There,  spatial localization was referred to the  rest $3$-space $\Sigma$, at given time,  of an inertial (Minkowskian) reference frame.
Later, guided by  Mackey's imprimitivity theory,  Wightman established \cite{Wightman} an uniqueness theorem. He proved that 
 the  joint (projector-valued) spectral measure $\sQ_\Sigma$ on $\Sigma$ of the triple of {\em Newton-Wigner selfadjoint  operators} $N_\Sigma^1,N_\Sigma^2,N_\Sigma^3$
is the unique projector valued measure (PVM)   which is  covariant  with respect to the Euclidean group of isometries of $\Sigma$   and complies  with  some further technical  hypotheses. If the notion of spatial localization is described in terms of  selfadjoint observables, they must be the Newton-Wigner ones necessarily.

Unfortunately, these  operators and their common spectral measure (their joint PVM) resulted  to be plagued by a number of fundamental issues related to causality. The {\em Hegerfeldt theorem(s)} \cite{Hegerfeldt,Hegerfeldt2} and the {\em Malament theorem} \cite{Malament}, with several modern reformulation mainly due to  Busch \cite{Buschloc} and Halvorson and Clifton \cite{HC}, proved that some  causality requirements  and the request of energy positivity are definitely incompatible for quantum particles described as one-particle states of  a quantum relativistic field (Wigner particles).

Malament's result and its modern extensions and  reformulations are directly or indirectly  related with the description of post measurement states, in terms of {\em L\"uders-von Neuman projection postulate} or referring to a {\em Kraus  decomposition}  of the  effects of the  POVM describing the localization observable. This is a deep  and outstanding issue \cite{Beck} we shall not discuss in this paper.

Conversely, the various versions of the Hegerfeldt theorem only focus on the spatial detection probability of a quantum relativistic particle at given time. 
The most elementary version  is like this. Let us consider the rest 3-space $\Sigma$ of a Minkowski observer and a quantum  relativistic free particle defined according to Wigner classification, with every mass $m\geq 0$ and every permitted  spin $s$.  
 Suppose that  the probability to detect the  particle,  whose pure quantum state is represented by the normalized vector $\psi$,   vanishes outside  a {\em bounded  spatial region} $\Delta$ at time $t=0$. Then,  the same state $\psi$   gives rise to a   {\em strictly positive}  probability to find the particle {\em arbitrarily far} from $\Delta$ and at {\em arbitrarily small} time $t>0$.  In other words,  {\em superluminal propagation of  probability} shows up here.
The hypotheses leading to Hegerfeldt's  result are quite mild: the crucial one is {\em positivity of energy} (more precisely below boundedness)  which is embodied in the Wigner's definition of particle\footnote{The proof of Hegerfeldt's  theorem relies on  properties of analytic continuations of one-parameter  evolution semigroups and turns out quite involved. A quantitative analysis of the effects  of the Hegerfeldt theorem in the form discussed in \cite{Hegerfeldt} appears in \cite{FP}.}. In that case,  pure states are represented by  unit  vectors in the {\em one-particle Hilbert space} $\cH$ of the  Fock space of the associated quantum field.  In the proof of the theorem, the detection probability is supposed to be computed through a PVM or, more generally, by means of a POVM (see below) labelled by the sets of a given rest $3$-space $\Sigma$ of an inertial observer in Minkowski spacetime.

The impact on NW operators is evident. Probability distributions which vanish outside a bounded region $\Delta$ are trivially constructed  when adopting a spatial localization notion  in terms of  a PVM  on $\Sigma$.
The position observable are in this case a triple of mutually compatible selfadjoint operators
$$X^k :=  \int_{\Sigma} x^k d\sP_\Sigma(x)\:, \quad k=1,2,3$$
in the Hilbert space $\cH$. $\sP_\Sigma=\sP_\Sigma(\Delta)$ for  $\Delta \in \cB( \Sigma)$ (where henceforth $\cB(\Sigma)$ is the family of Borel sets on $\Sigma$)  is their joint spectral measure. A normalized vector in   $\psi \in \sP_\Sigma(\Delta)\cH$ immediately yields a probability distribution  
$\langle \psi |P(E) \psi \rangle$ 
 which vanishes if  $E \cap \Delta= \emptyset$,  fitting the hypothesis of Hegerfeldt's theorem.
 If we do not accept superluminal propagation of detection probability, on account of Hegerfeldt achievement, we are  committed to reject any description of the spatial localization in terms of PVM, i.e., selfadjoint operators. The first victim of this reasoning is  the very triple of Newton-Wigner operators  $X^k=N^k_\Sigma$, $k=1,2,3$\footnote{As already found in \cite{M2},  in the new picture arising from modern achievements as the present work, the Newton-Wigner observable is however recovered in terms of the  {\em  first moments} of permitted  localization  observables which are   described in terms of POVMs.}.

{\em Should we rule out spatial localization described by PVMs on the ground of the theoretical evidence of superluminal propagation of  probability?}

The answer needs a certain analysis. What we can control in laboratories  in practice are just  macroscopic objects and devices. At macroscopic level, superluminal propagation of information is forbidden. So, a better perspective to tackle this  issue is wondering {\em whether or not the  superluminal propagation of probability predicted by the Hegerfeldt theorem can be used to propagate superluminal  macroscopic information}. The answer is positive, in case of 
states whose probability distribution at $t=0$ vanishes outside a  bounded spatial region. A corresponding ideal experiment   was discussed  in \cite{M2} by one of the authors of this work.
The conclusion is that,  to describe spatial localization  on a  rest $3$-space $\Sigma$ at some time of an inertial observer,
the use of  a suitable {\em normalized positive-operator  valued  measure} (POVM) is compulsory.
A  $\cH$-POVM  \cite{Buschbook} on $\Sigma$ is a map
\beq \cB(\Sigma) \ni \Delta \mapsto \sA_\Sigma(A) \in \gB(\cH)\:\label{POVM}\eeq
such that
$0\leq \sA_\Sigma(\Delta) \leq I$ for every $\Delta \in \cB(\Sigma)$, where 
$\cB(\Sigma) \ni \Delta \mapsto \langle \psi|\sA_\Sigma(A)\psi'\rangle$ is a complex measure for every $\psi,\psi' \in \cH$.
$\sA_\Sigma$ is {\em normalized} if $\sA_\Sigma(\Sigma)=I$.
If $\psi \in \cH$ is normalized as well, the meaning of  $\langle \psi|\sA_\Sigma(\Delta) \psi \rangle$ is the probability to find the system in $\Delta$, when its state is $\psi$.

 There are POVMs which do {\em not} admit probability distributions $\cB(\Sigma) \ni E \mpasto \langle \psi| \sA_\Sigma(E) \psi \rangle$  localized in bounded regions for any conceivable pure  state $\psi \in \cH$. Contrarily to PVMs  which always  permits probability distributions localized in bounded sets as viewed above. Hegerfeldt's thoerem does not exclude POVMs  to describe the spatial localization of relativistic quantum particles. More precisely, it permits  the POVMs whose probability distributions to detect the particle  is not supported in a bounded region for every pure state $\psi \in \cH$.
It is obviously  necessary that  probability distributions localized in bounded regions can be  arbitrarily approximated by permitted  distributions. It in fact happens in concrete POVMs describing spatial localization \cite{C0,M2,C}.

Mathematically speaking, at this juncture, it is convenient to give the following definition, where $\bM$ denotes the Minkowski spacetime and $\cC^{sf}_\bM$ is the family of all possible rest $3$-spaces $\Sigma$ of inertial observers at any given instant of  their proper time.

\begin{definition}\label{SLO1} {\em Given a quantum system described in the Hilbert space $\cH$, a {\bf spatial localization observable}  of it in $\bM$ is a collection $\sA := \{\sA_\Sigma\}_{\Sigma \in \cC^{sf}_\bM}$,  where each $\sA_\Sigma$ is  a normalized $\cH$-POVM on $\cB(\Sigma)$
which is absolutely continuous with respect to the Lebesgue measure on $\Sigma$.
}\hfill $\blacksquare$
\end{definition}

\noindent  In principle, a {\em coherence condition} should be also imposed if one wants the probability to find the particle in $\Delta$ be independent of the rest space $\Sigma$ containing it. One should require that, if $\Delta \subset \Sigma\cap \Sigma'$, then $\sA_\Sigma(\Delta) = \sA_{\Sigma'}(\Delta)$.
 However, the POVMs $\sA_\Sigma$
 are absolutely continuous with respect to the Lebesgue measure of $\Sigma$ by hypothesis.
  Since  $\Delta \subset \Sigma\cap \Sigma'$ has zero measure if $\Sigma\neq \Sigma'$,
we have   $\sA_\Sigma(\Delta)= 0 = \sA_{\Sigma'}(\Delta)$ in any cases.
Actually,  all physically relevant  spatial localization observables  are also invariant under translations in $\Sigma$ due to general   Poincar\'e covariance requirements and this fact automatically implies Lebesgue absolute continuity (see (d) Theorem 2.20 in \cite{Rudin}).

Definition \ref{SLO1} is nothing but Definition 18 stated in \cite{M2} for a {\em relativistic spatial localization observable} without the requirement of  $\cP_+$-covariance.

 An important remark is the following one. Consider the case  $\Delta, \Delta' \in \Sigma$ and $\Delta \cap \Delta' = \emptyset$.
  Referring to quantum relativistic particles (with positive energy), operators $\sA_\Sigma(\Delta)$ and $\sA_\Sigma(\Delta')$ cannot commute, although $\Delta$ and $\Delta'$ are causally separated in the considered case. That is due to general no-go results\footnote{See Lemma 4 on p.13 of \cite{CL}  and its proof in particular.} \cite{Buschloc,HC,CL, Beck}. This is a delicate issue somehow related to the {\em Reeh-Schlieder property} if embedding all the theoretical description in the framework of local algebras of observables. Roughly speaking, the operators  $\sA_\Sigma(\Delta)$ cannot belong to a  local algebra in the sense of {\em Haag-Kastler}. Whether or not this non-commutativity  permits superluminal propagation of macroscopic information should be analyzed in a perspective (as in \cite{FewsterVerch,FJR}) where {\em measurement instruments} are explicitly studied. This analysis would concern post measurement states which are outside of the goals of this paper.

To get closer to the results achieved in  this paper, we observe that there is  a more sophisticated version of the Hegerfeldt theorem \cite{Hegerfeldt2}: If a state $\psi \in \cH$ produces a probability distribution which vanishes sufficiently rapidly at infinity at $t=0$, then a more specific type of   {\em superluminal propagation of  probability}  takes place again at later times.
 
The modern overall description  of the types of violation  was only recently formalized by Castrigiano. Let us quickly review the {\em causality condition}
and the {\em causal time evolution condition}, introduced in \cite{C0} for localization POVMs in Minkowski spacetime $\bM$, to reformulate the second  form of Hegerfeldt's theorem \cite{Hegerfeldt2}.

In the following,  $\cL(\Sigma)$ is the  family of Lebesgue sets of  $\Sigma \in \cC^{sf}_\bM$, and the considered POVMs 
are defined on this $\sigma$-algebra larger that the more  usual\footnote{The use of the Lebesgue $\sigma$-algebra $\cL(\Sigma)$ in place of the Borel one $\cB(\Sigma)$ is compulsory here. In fact, as discussed in \cite{C0}, if defining the relevant POVMs on $\cB(\Sigma)$ rather than on $\cL(\Sigma)$  it is not guaranteed that $\Delta' \in \cB(\Sigma')$ when $\Delta \in \cB(\Sigma)$, making meaningless (\ref{CCC}).
It is however true   that $\Delta' \in \cL(\Sigma')$ for every $\Delta \subset  \Sigma$ (even if $\Delta \not \in \cB(S)$!) as established in \cite{C0}. The use of $\cL(\Sigma)$ is not a serious issue, even if initially  dealing with POVMs defined on the Borel sets, since the standard completion procedure  uniquely extends them to POVMs on $\cL(\Sigma)$.} $\cB(\Sigma)$.
 Finally,  if $\Sigma,\Sigma'\in \cC^{sf}_\bM$  and $\Delta \in \cL(\Sigma)$, we define the {\bf region of influence of $\Delta$ on $\Sigma'$}
\beq
\Delta' := (J^+(\Delta) \cup J^-(\Delta)) \cap \Sigma' \label{Delta'old}\:.
\eeq 
($J^\pm(\Delta)$ are the standard {\em causal sets} emanated from $\Delta$ as defined in Section \ref{SECCAUS}.)
The physical meaning of $\Delta'$ should be obvious. {\em It is the largest spatial region on $\Sigma'$ which may have macroscopic causal relations with $\Delta\subset \Sigma$}, when assuming  the causal structure of Minkowski spacetime based on the bounded propagation velocity of macroscopic  physical signals.  

Let us  consider  a spatial localization observable $\sA := \{\sA_\Sigma\}_{\Sigma \in \cC^{sf}_\bM}$, where $A_\Sigma: \cL(\Sigma)\to \gB(\cH)$.

\begin{definition}  \label{DEFC0}
{\em  $\sA$
 satisfies  {\bf Castrigiano's causality condition} ({\bf CC}) \cite{C0} if
\beq
\sA_{\Sigma'}(\Delta') \geq \sA_\Sigma(\Delta)\quad \mbox{when   $\Delta \in \cL(\Sigma)$,} \label{CCC}
\eeq
 for every $\Sigma, \Sigma' \in \cC^{sf}_\bM$. \hfill $\blacksquare$}
 \end{definition}
 
 The physical meaning of the condition above is evident: detection probability cannot propagate faster than the light speed.
 A weaker condition which is implied by CC but does not imply CC, is the {\em causal time evolution}. This condition -- and not CC -- is actually  the specific subject of Hegerfeldt's investigation. It only considers the case where $\Sigma$ and $\Sigma'$ are rest spaces of a {\em common inertial observer} $n$, and thus they are related  by means of the time evolution proper of $n$. Geometrically speaking, that is equivalent to saying that the normal vectors $n_\Sigma$ and $n_{\Sigma'}$ coincide (with $n$).
\begin{definition}  \label{DEFC0C}
{\em  $\sA$
 satisfies the  {\bf causal time evolution condition} condition ({\bf CT}) \cite{C0} if
\beq
\sA_{\Sigma'}(\Delta') \geq \sA_\Sigma(\Delta)\quad \mbox{when   $\Delta \in \cL(\Sigma)$,}\label{CCCT}
\eeq
 for every $\Sigma, \Sigma' \in \cC^{sf}_\bM$ with $n_\Sigma =n_{\Sigma'}$.\hfill $\blacksquare$}
 \end{definition}

The afore-mentioned  second version of the Hegerfeld theorem \cite{Hegerfeldt2} (see also further formalizations in \cite{C}) can be stated as follows, for a free quantum relativistic particle of any spin and   mass $m\geq 0$. (The result was actually extended  \cite{Hegerfeldt2} to more  general relativistic systems in, also self-interacting, but satisfying energy positivity). 

\begin{theorem}[Hegerfeldt’s theorem on relativistic time
evolution \cite{Hegerfeldt2}]\label{HTeo}
Consider a spatial localization observable $\sA$ for a  Wigner particle of mass $m\geq 0$  described in the Hilbert space $\cH$.
If there is a unit vector $\psi \in \cH$  such that the associated  probability to find 
the particle outside a ball  of radius $r$  in a certain $\Sigma$  is bounded by  $K_1e^{-K_2r}$, for all $r>0$ and constants $K_1\geq 0$ and $K_2 > 2 \frac{mc}{\hbar}$,
then  $\sA$ fails to satisfy CT.
\end{theorem}

\begin{remark}{\em As consequence,  $\sA$  fails to satisfy CC as well. }\hfill $\blacksquare$\end{remark}

Castrigiano proved in \cite{C0} that spatial localization observables exist which satisfy CC for massive and massless spin $1/2$ fermions. 
Another   family of {\em causal} -- namely  satisfying CC-- spatial localization observable $\{\sM^n_\Sigma\}_{\Sigma\in \cC^{sf}_\bM}$ was obtained in \cite{M2} by one of the authors of this work for scalar, real, massive Klein-Gordon particles. That observable was constructed by generalizing a more specific  observable introduced by D. Terno \cite{T} and also studied in \cite{M2}, where it was rigorously established  that it satisfies CT.
More recently, D. Castrigiano rigorously proved in \cite{C} that the spatial localization observables for Klein-Gordon particles with causal kernels  $\{\sT^g_\Sigma\}_{\Sigma\in \cC^{sf}_\bM}$  introduced in  \cite{P,HW} 
satisfy CC. These three types of causal spatial  localization observables also enjoy natural covariance properties with respect to the orthochronous Poincar\'e group in $\bM$ as expected. In this sense, they  are {\em relativistic} spatial localization observables (we shall come back to those features later in the paper).

A natural issue, discussed in the first part of this work,  is  whether or not we can  physically admit  a spatial localization observable that  {\em violates} CC. In other words, {\em is it possible to transmit superluminal macroscopic information by exploiting  it?}

 It is not easy to re-adapt the ideal experiment presented in \cite{M2} to the case of a POVM and a state whose   probability distribution {\em is  not  supported in a bounded region} but  vanishes sufficiently fast at infinity to switch on Hegerfeldt’s theorem on relativistic time
evolution (Theorem \ref{HTeo}), thus violating CT and CC.  It seems to the author of \cite{M2} that a   re-adaptation of the reasoning in \cite{M2} should involve some choice about a post-measurement state prescription or about some quantum instrument.  As declared before, we do not want to pursue that route in this paper, since we do not intend  to deal with  more sophisticated theoretical notions than POVMs.   

As a matter of fact,  CC can be seen as a consequence of a more general causality condition we shall introduce in Section \ref{SECCP}.
The crucial  overall idea\footnote{D. Castrigiano communicated to V.M. that R.F. Werner independently had a similar idea to prove CC.} is to extend the notion of spatial localization to a broader class of $3$-dimensional  surfaces in Minkowski spacetime.
Minkowski spacetime is  in fact {\em globally hyperbolic}. In other words, it admits certain special $3$-dimensional  surfaces, called {\em Cauchy surfaces}, which can be used as the place where assigning initial data of causal  (hyperbolic)  PDEs. From the geometric perspective, each of these  surfaces are met by all (inextendible) causal curves, exactly once by timelike ones. No macroscopic physical information can be transmitted from a region on a {\em spacelike} Cauchy surface to another separated  region on the {\em same} Cauchy surface.  Rest spaces  of inertial reference frames at given times  are just a special flat case of smooth Cauchy surfaces. 
{\em  From a relativistic perspective, it seems natural to think of spacelike Cauchy surfaces as a generalization of the notion of space at a given time, where to localize quantum systems like particles\footnote{However, see \cite{RP} and below in the text for a different viewpoint on this issue.}.}
We therefore  extend the notion of spatial localization observable to a more general notion  where the possible regions $\Delta$, where a particle can be  detected,  are subsets of a generic spacelike Cauchy surface $S$. 
To this end, each (generally curved)  Cauchy surface $S$ is  equipped with a {\em normalized} POVM $\sA_S$.  

A {\em spacelike Cauchy localization observable} (Definition \ref{GSL}) is the family all these normalized POVMs, when $S$ varies in the collection of all spacelike Cauchy surfaces.
An important {\em coherence condition} is also imposed: if $\Delta \subset S\cap S'$, then $\sA_S(\Delta) = \sA_{S'}(\Delta)$. In other words,  the probability to find the particle in $\Delta$ is independent of the Cauchy surface, but it is a function of $\Delta$ and the state $\psi \in \cH$ only.
Finally we also require that the probability to find a particle in a zero measure set on $S$ vanishes.

At this juncture, a  {\em general causality condition} GCC  (Definition \ref{DEFC}) for  spacelike Cauchy localization observables  can be stated, in the spirit of Definition \ref{DEFC0}.  The general causality condition is just obtained by  replacing the flat spacelike Cauchy surfaces $\Sigma, \Sigma'$ in Definition \ref{DEFC0}
for generic spacelike Cauchy surfaces $S,S'$.
This general  causality condition evidently implies CC when restricting $\sA$ to the  POVMs $\sA_\Sigma$ defined on the subfamily  flat Cauchy surfaces $\Sigma$.  A natural issue pops out here:

{\em How is  selective the general causality condition?}

\noindent Quite surprisingly, the answer is that there is {\em no} selection at all: {\em every  spacelike Cauchy localization observable automatically satisfies the general  causality condition} (Theorem \ref{TONE1}). 
As a corollary, {\em if a spatial localization observable $\{\sA_\Sigma\}_{\Sigma\in \cC^{sf}_\bM}$ can be extended to a  spacelike Cauchy localization observable $\{\sA_S\}_{S\in \cC^s_\bM}$, then the former automatically  satisfies CC}.  

It seems quite remarkable that the  result  only uses a general, and very 
natural, notion of localization. The achieved result proves in fact that  the general causality relation should not be imposed as an independent postulate because  it is physics that asks for it. In summary, the presented result proves that mere localizability implies causality independently of  any kinematical argument.

The strategy to prove our main result stated above takes  advantage of some general technical facts in Lorentzian geometry on  existence of spacelike Cauchy surfaces  adapted to given submanifolds with boundary  \cite{BS,BS3}.  

Other approaches exist where,  also considering many particles, one renounces to energy positivity and defines a notion of spatial localization on generic  Cauchy surfaces in terms of PVMs \cite{ADD1,ADD2}, finding similar causality conditions.
A different general approach to the Born rule for the spatial localization of a particle formulated in a generic spacetime is discussed in \cite{RP}. There, a more general notion of $3$-space at given time is introduced and analyzed (also relying on technical results in \cite{MEHH}). That notion of space  is adapted (transverse)  to the considered  probability current and,  for this reason,  it is not necessarily spacelike nor a Cauchy surface.

The second  part of the work is mostly devoted to prove that spacelike Cauchy localization observables {\em do exist}. Some of them are (uniquely determined) extensions
 $\{\sT^g_S\}_{S\in \cC^{s}_\bM}$ and $\{\sM^n_S\}_{S\in \cC^{s}_\bM}$
 of the respective above mentioned spatial localization observables $\{\sT^g_\Sigma\}_{\Sigma\in \cC^{sf}_\bM}$ and $\{\sM^n_\Sigma\}_{\Sigma\in \cC^{sf}_\bM}$.
 The method to demonstrate the results in the second part  uses  ideas physically formulated in \cite{P,J,HW},  made rigorous in \cite{M2,C} and here further improved.

Some further results about the interplay of spatial localization observables, the NW operators, and the Heisenberg inequality  appear in several spots of the main text and in a cumulative proposition in the last section.

In details, the achievements of this work are as follows.
\begin{itemize}
\item[(1)]   {\bf Theorem \ref{TONE1}}: {\em Every spacelike Cauchy localization observable satisfies the general causality condition.}\\ {\bf Corollary \ref{TONE1COR}}: {\em If a spatial localization observable can be extended to a  spacelike Cauchy localization observable, then it satisfies Castrigiano's causality condition}.

\item[(2)]  {\bf  Theorems \ref{MAIN}} and {\bf \ref{MAIN3}}: {\em  Both $\{\sT^g_\Sigma\}_{\Sigma\in \cC^{sf}_\bM}$ and  $\{\sM^n_\Sigma\}_{\Sigma\in \cC^{sf}_\bM}$ uniquely  extend to corresponding spacelike Cauchy localization observables $\{\sT^g_S\}_{S\in \cC^{s}_\bM}$ and  $\{\sM^n_S\}_{S\in \cC^{s}_\bM}$. The former can be even  defined for smooth Cauchy surfaces $S$ which are not spacelike.}

 \item[(3)] {\bf Theorems \ref{1MNW1}}, {\bf \ref{1MNW2}},  {\bf Proposition \ref{PROPVAR}}: {\em Both $\{\sT^g_\Sigma\}_{\Sigma\in \cC^{sf}_\bM}$ and  $\{\sM^n_\Sigma\}_{\Sigma\in \cC^{sf}_\bM}$  give rise to the Newton Wigner operators as their first moments on every inertial rest space
$\Sigma$.  The NW operators are  the unique selfadjoint operators which satisfy this property}. {\em A generalized Heisenberg inequality holds in all cases}.
\end{itemize}

\subsection{Structure of the paper}
After a list of  notations and conventions adopted  in this work, Section \ref{SECCAUS} concerns a recap of the causal structure of globally hyperbolic spacetimes and Minkowski spacetime in particular. Section \ref{SECCP}, after introducing the relevant definitions, presents our main results about causality: every spacelike Cauchy localization observable is causal.  Section \ref{INTER} collects  some technical notions and results necessary to pass to the second set of achievements. These results are established in Sections \ref{SECPVMC}
and \ref{SECMMM}: the possibility to extend some relevant spatial localization observables defined in \cite{C} and \cite{M2} to corresponding spacelike Cauchy localization observables.
In the same section, we shall discuss some features of  spatial localization observables. 
Section \ref{NWApp} is devoted to  discuss some general facts about Newton Wigner observables, Heisenberg inequality and spatial localization observables.
After a final discussion in Section \ref{DISC}, the appendices contain the proofs of several technical lemmata and propositions  asserted  in the main text.

\subsection{General definitions, notations, and conventions}
Barring few changes (like the symbol for the flat Cauchy surfaces), we shall adopt the same notation as in \cite{M2}.

Throughout  $\bR_+ := [0,+\infty)$, $\overline{\bR_+}:= \bR_+ \cup \{+\infty\}$, and {\bf smooth} means $C^\infty$. The light speed is $c=1$. The normalized  Planck constant is $\hbar=1$.  Furthermore we shall take advantage of the following notation:  $$\vec{x} \equiv (x^1,x^2,x^3) \in \bR^3\quad  \mbox{and}\quad  
\nabla f = (\partial_{x^1} f, \partial_{x^2} f, \partial_{x^3} f)$$  if $f= f(\vec{x})$ is defined on a suitable open domain of $\bR^3$.

 The {\em Lebesgue measure} (also restricted to the Borel sets) on $\bR^n$ will be denoted by $d^nx$, when $x^1,\ldots, x^n$ are orthonormal Cartesian coordinates on $\bR^n$.  We use also the notation $|B| := \int_B 1 d^nx$. The family of {\em Borel sets} on a topological space $X$ will be denoted by $\cB(X)$.
The {\em Lebesgue $\sigma$-algebra} on $\bR^n$ will be indicated by $\cL(\bR^n)$.

We assume   the following  normalisation  convention concerning  volume $n$-forms in $\bR^n$
\beq
\int_A dx^1 \wedge \cdots \wedge dx^n := \int_A  dx^1\cdots dx^n \:. \label{VF} 
\eeq

The inner product $\cH \times \cH \ni (x,y) \mapsto \langle x|y\rangle \in \bC$ in a complex Hilbert space $\cH$ is assumed to be linear in the {\em right} entry.
 
If $A: D(\cH)\to \cH$ is an operator in the complex  Hilbert space $\cH$ with domain given by the linear subspace $D(A)\subset \cH$, saying that $A$ is {\bf positive}, written $A\geq 0$, means $\langle \phi|A \phi \rangle \geq 0$ for all $\phi\in D(A)$. Furthermore, if $D(A)=D(B) \subset H$ for a corresponding pair of operators $A,B$, then $A\geq B$  (also written $B\leq A$) means $A-B \geq 0$.

$\gB(\cH)$ throughout denotes the {\em unital $C^*$-algebra of bounded operators} \cite{M} $A: \cH\to \cH$, where  $\cH$ is a complex Hilbert space.

Let  $\cM(X)$ be a  $\sigma$-algebra  on $X$ and $\cH$ a complex Hilbert space. A {\bf $\cH$-POVM} (Positive-Operator Valued Measure) on  $\cM(X)$ \cite{Buschbook} is a map $$\sA :\cM(X)\ni \Delta \to \sA(\Delta) \in \gB(\cH)$$ such that
$0\leq \sA(\Delta)\leq I$ and $\Sigma(X)\ni \Delta \mapsto \langle \psi| \sA(\Delta) \phi\rangle \in \bC$ is a $\sigma$-additive  complex measure (with finite total variation) for every $\psi,\phi \in \cH$.  It turns out that a POVM is also $\sigma$-additive in the strong operator topology. $\sA$ is {\bf normalized} if $\sA(X)=I$.

A normalized  ${\cal H}$-POVM  $\sA$ on $\cM(X)$ is a ${\cal H}$-{\bf PVM} (Projector Valued Measure) on $\cM(X)$ if $\sA(\Delta)^2 = \sA(\Delta)$ (i.e., $\sA(\Delta)$ is an {\em orthogonal projector} in ${\cH}$) for every $\Delta \in \cM(X)$.  We shall use standard notions and constructions of operator theory and spectral theory (see, e.g., \cite{M}).

An ${\cal H}$-POVM  $\sA$ on $\cM(X)$   is {\bf absolutely continuous} with respect to a positive measure $\mu : \cM(X) \to \overline{\bR_+}$, written $\sA <\sp< \mu$,  if $\langle \psi|\sA(E) \psi \rangle =0$ for every $\psi \in \cH$ when $\mu(E)=0$, for $E\in \cM(X)$.

\section{Causal structures of spacetimes, Minkowski spacetime} \label{SECCAUS}
This section is devoted to quickly  introduce the basic causal structures we shall use in the rest of the paper (see, e.g., \cite{ON}). The reader who is already familiar with these notions and 
wants to achieve the main results  soon, may temporarily skip this part, coming back to it when necessary, or for understanding a specific notation or a  definition.

\subsection{Globally hyperbolic spacetimes and their causal structure}
A {\em spacetime} is the most general scenario where formulating any macroscopic physical theory according with standard notions of causality (see, e.g., \cite{ON}).  It is a continuous set of {\em events}, a smooth manifold,  equipped with geometric structures which account for causal relations. The crucial notion underpinning these notions is the  Lorentzian metric tensor field on the spacetime.

\begin{definition}{\em A $n$-dimensional {\bf  spacetime} $(M,g)$ is a connected  $n$-dimensional  ($n\geq 2$) smooth manifold $M$ equipped with a {\bf Lorentzian metric} tensor field $g$. It is a smooth assignment to every $T^*_pM\otimes T^*_pM$ of  nondegenate symmetric tensors $g_p$  with constant signature $-1,+1,\ldots, +1$.\\
If $p\in M$,  $v\in T_pM\setminus\{0\}$ is {\bf spacelike}, {\bf timelike}, {\bf lightlike} if  respectively $g_p(v,v) >0$, $g_p(v,v)<0$, $g_p(v,v)=0$.    {\bf Causal vectors} are both timelike and lightlike. The zero vector $0\in T_pM$ is spacelike {\em per definition}.
 \hfill $\blacksquare$}
 \end{definition}

We adopt the same terminology for  smooth curves $(a,b) \ni s \mapsto \gamma(s) \in M$: they are {\bf spacelike}, {\bf timelike}, {\bf lightlike}, {\bf causal}  according to the character of their tangent vector $\dot{\gamma}(s)$ {\em supposed to be uniform along $\gamma$}.

Co-vectors in $T_p^*M$ are classified as {\bf spacelike}, {\bf timelike}, {\bf lightlike}, {\bf causal}, according to the associated elements of $T_pM$  through the  natural isomorphism $T_pM \ni v\mapsto g_p(v, \cdot ) \in  T_p^*M$. The  indefinite inner product induced on $T^*_pM$ by that isomorphism  will be denoted by $g_p^\#$.

As usual, a $C^k$ {\bf submanifold} ($k=0,1,\ldots, \infty$) of a given smooth manifod $M$ is a subset  $N\subset M$ equipped with its own structure of $C^k$ manifold. The submanifold $N\subset M$ is said to be {\bf embedded}\footnote{We shall not consider {\em immersed submanifolds} in this work.} and of {\bf dimension} $m\leq n:=dim(M)$ if (a) for  every $p\in N$ there is a local chart $(U,\psi)$  {\em of $M$}, such that  $p\ni U$ and $\psi(U\cap S)$ is the intersection of $\psi(U)$ and $\bR^{m}$ (viewed as a standard  $m$-plane in  $\bR^n$); (b) 
the topology on $N$ is the one induced by $M$; and (c)  $(U\cap N, \psi|_{U\cap N}: U\cap N \to \bR^m)$ is a local $C^k$ chart on $N$. Evidently, $T_pN$ turns out to be a linear subspace of $T_pM$ when $k\geq 1$.

\begin{definition}{\em  In a $n$-dimensional  spacetime $(M,g)$, an $m$-dimensional ($m\leq n$) embedded smooth submanifold $S$ (possibly with boundary) is  {\bf spacelike} if the  tangent vectors at every point (including the ones tangent to the boundary if any) are spacelike.\\ If $m=n-1$  that is equivalent to saying that the co-normal vector  to $S$ is {\em timelike} everywhere. \hfill $\blacksquare$}
\end{definition}

 To go on, observe that the set of timelike vectors in $T_pM$ is made of two disjoint open cones $V_p$  and $V'_p$.
 
\begin{definition}\label{DEFST} {\em A spacetime $(M,g)$ is {\bf time oriented} if there exists a continuous timelike vector field $T$ on $M$.  If $p\in M$,
\begin{itemize}
\item[(a)] the
open cone $V^+_p\subset T_pM$ of  the pair  $V_p$  and $V'_p$ which contains $T_p$   is called
the  {\bf future open cone at $p$};
the causal vectors of $\overline{V^+_p}\setminus\{0\}$ are said to be {\bf future-directed};
\item[(b)] the remaining cone $V^-_p$ is the {\bf past open cone at $p$}; the causal vectors of $\overline{V^-_p}\setminus\{0\}$ are said to be {\bf past-directed}. 
\end{itemize}
A global  continuous choice of $V_p^+$  for every $p\in M$ as above is a {\bf temporal orientation} of $(M,g)$.\\
 A causal co-vector is {\bf future-directed} or {\bf past-directed} if it is the image of a, respectively, future-directed or past-directed causal vector through  $T_pM \ni v\mapsto g_p(v, \cdot ) \in  T_p^*M$. \hfill $\blacksquare$}
\end{definition}

\begin{remark}
{\em \begin{itemize}
\item[(1)] As $M$ is connected, there are  two possible temporal orientations or none. 

\item[(2)] If $t\in V^+_p$  and $u \in T_pM$ is causal, then
 $u$ future-directed $\Leftrightarrow  g_p(t,u)<0\:.$ \\
(The same result is valid  for  $t$   future-directed lightlike and non-parallel to $u$.) \hfill $\blacksquare$
\end{itemize}}
\end{remark}

\begin{definition} {\em  Let $(M,g)$ be a time oriented spacetime.
If $A\subset M$, 
\begin{itemize}
\item[(a)] its {\bf chronological future} $I^+(A) \subset M$ is the set of $q \in M$ such that there is a future-directed timelike  smooth curve $\gamma: (a,b)\to M$ with  $\gamma(t_1) \in A$ and $\gamma(t_2) =q$, for $t_1< t_2$;
 \item[(b)] its {\bf causal  future} $J^+(A) \subset M$ is the set  of $q \in M$ such that  either $q\in A$ or  there is a future-directed causal smooth curve $\gamma : (a,b)\to M$ with $\gamma(t_1) \in A$ and $\gamma(t_2) =q$, for $t_1< t_2$. 
\end{itemize} The {\bf chronological past} $I^-(A)$ and the {\bf causal past} $J^-(A)$ are defined similarly. 
\begin{itemize} 
\item[(c)] $A$ is {\bf achronal} if $A \cap I^+(A) = A \cap I^-(A) = \emptyset$\:,
\item[(d)] $A$ is {\bf acausal} if there is no causal smooth curve $\gamma : (a,b)\to M$ with $\gamma(t_1), \gamma(t_2) \in A$ and $t_1 \neq t_2$. \hfill $\blacksquare$
\end{itemize}
 }
 \end{definition}
 Physically speaking, macroscopic physical information is transported along causal (future-directed) curves. The existence of the sets $J^\pm(p)$ (and also $V_p^\pm$) is the mathematical description of the finite propagation speed of physical information exiting  (or entering) $p$.

\begin{definition} {\em Let $(M,g)$ be a time oriented spacetime. A  future-directed causal smooth curve $\gamma : (a,b)\to M$, with $a,b \in [-\infty,+\infty]$, is {\bf future}, resp. {\bf past}, {\bf inextendible} if there is no $p\in M$ such that, respectively, $\gamma(s)\to p$ for $s\to b$ or $s\to a$. \\$\gamma$ is {\bf inextendible} if it is both future and past inextendible.  \hfill $\blacksquare$}
\end{definition}

We shall now focus attention on the so called {\em globally hyperbolic spacetimes}. This  kind of spacetimes is of the  utmost physical interest for many reasons, in particular because a wide family of, roughly speaking,  {\em hyperbolic  PDEs} of great physical relevance --  as the Einstein equations, Klein-Gordon equations, Dirac equations  -- admit existence and uniqueness theorems.  Cauchy data are given on special subsets called {\em Cauchy surfaces}. Very interestingly, the definition of globally hyperbolic spacetime and Cauchy surface is not related to PDEs, but only relies on the above geometric  causal structures. {\em Spacelike} Cauchy surfaces are also the natural representation of instantaneous rest spaces of globally extended observers.

\begin{definition} {\em  A time oriented spacetime $(M,g)$ is said to be  {\bf globally hyperbolic} (e.g., see \cite{BS}) if it includes a {\bf Cauchy surface}. That is  a set $S\subset M$ which intersects every inextendible  timelike smooth curve exactly once.  
 \hfill $\blacksquare$}
\end{definition}

\begin{remark}\label{RCSu}
{\em There are  features of Cauchy surfaces in a $n$-dimensional  globally hyperbolic spacetime  $(M,g)$ which deserve mention also because they are technically important for this paper.
\begin{itemize}
 
\item[(a)] Every Cauchy surface is in particular achronal. It  is also met  by every {\em causal} inextendible curve, but not necessarily once \cite{ON}.

\item[(b)] As it was established by Geroch, a Cauchy surface $S$ is a  closed (in $M$) $C^0$ embedded  submanifold  of co-dimension $1$ of $M$ which, in turn, is homeomorphic to $\bR\times S$. All Cauchy surfaces are homeomorphic. \cite{ON}.

\item[(c)]  If a Cauchy surface $S$ is also a smooth 
embedded co-dimension $1$
 submanifold of $M$, then its tangent vectors at each point  must be either spacelike or likghtlike, since $S$ does not contain timelike curves.
  These Cauchy surfaces are said {\bf smooth Cauchy surfaces}.

\item[(d)] As established by Bernal and S\'anchez \cite{BS}, every globally hyperbolic spacetime  admits 
{\em spacelike} smooth Cauchy surfaces 
 (necessarily closed in $M$ for (b)). In turn, $M$ is diffeomorphic to $\bR \times S$. \\
The following further facts  are  of crucial interest to our work.
\begin{itemize}
\item[(d1)] (Proposition 5.18 in \cite{BS3}) {\em Let $S_-,S_+$ be disjoint  spacelike smooth Cauchy surfaces of $(M,g)$ with $S_- \subset I^-(S_+)$. If $S \subset I^+(S_-)\cap I^-(S_+)$ is a closed (in $M$) connected spacelike smooth  embedded $1$-codimensional submanifold of $M$, then $S$ is a Cauchy surface of $(M,g)$ as well.   }

\item[(d2)] (Theorem 1.1 and Remark 4.14 of \cite{BS3}) {\em If $K\subset M$ is a compact spacelike acausal smooth $p$-dimensional  ($p=1,2,\ldots, n-1$) (embedded) submanifold with
 boundary\footnote{We refer to \cite{Lee} for the notion and the use of manifold with boundary.} 
 of $(M,g)$ not necessarily connected, then there is a spacelike smooth Cauchy surface $S\supset K$.}
\end{itemize}

\item[(e)] A  
{\em spacelike} smooth 
Cauchy surface $S$ meets exactly once also every  inextendible causal curve \cite{BS}. In particular, $S$ is therefore  acausal. \hfill $\blacksquare$
\end{itemize}}
\end{remark}

\begin{remark}{\em $\cC_M$ will henceforth denote the family of all smooth Cauchy surfaces of $(M,g)$.\hfill $\blacksquare$}
\end{remark}

\subsection{Minkowski spacetime geometric and causal structures}\label{SECM}
 \begin{definition} {\em  {\bf Minkowski spacetime}  $(\bM, \sV, \gor)$ is a 4-dimensional real affine space whose space of translation --  denoted by  $\sV$ -- is equipped with  a  Lorentzian inner product with signature\footnote{Differently from the choice of \cite{C0,C}}   $(-1,+1,+1,+1)$ indicated  by $\gor: \sV\times \sV \ni (k,p) \mapsto \gor(k,p) =:  k\cdot p  \in  \bR$. \hfill $\blacksquare$}
 \end{definition}

\noindent This definition is consistent with Definition \ref{DEFST} because  $\bM$ is automatically a smooth 4-dimensional manifold with respect to the natural smooth structure induced by its affine structure. It is   defined by requiring that, upon the choice of an  origin $o\in \bM$, the bijective  map $ I^o_p :\bM \ni p \mapsto p-o \in \sV$
 is a diffeomorphism.  $\imath_p:= dI^o_p :T_p\bM \to \sV$ is a natural isomorphism independent of the choice of $o$.
 $\gor$  can therefore be  exported to each tangent space $T_p\bM$ as  $g_p(u,v) := {\bf g}(\imath_p u, \imath_p v)$. The construction  defines a smooth manifold diffeomorphic to $\bR^4$ endowed  with a Lorentzian metric tensor field --  called {\bf Minkowski metric} -- still denoted by $\gor$.  In the rest of the paper,  for $u,v \in T_p\bM$,  we use the notation  ${\bf g}(u,v):=u\cdot v :=g_p(u,v)$.

Minkowski spacetime $(\bM,\gor)$ is assumed  to be  {\em time oriented} and $d\imath_p(V_p^+)=: \sV_+\subset \sV$ (which does not depend on $p$) is   the  {\bf future open cone} in $\sV$ by definition.

The {\bf orthochronous Poincar\'e group}  $\cP_+$ is the group of time-orientation preserving  isometries of $(\bM, {\bf g})$. $\cP_+$ turns out to be  the semidirect product $\sV \rtimes  \cL_+$ of $\sV$ (the abelian group  of displacements of $\bM$)  and the  {\bf orthochronous Lorentz group} $\cL_+$ consisting of the $\sV_+$-preserving linear  isometries of $(\sV, {\bf g})$.  Upon the choice of an origin $o\in \bM$ and referring to the map $I^o_p$ above, the action of $(v,\Lambda) \in  \cP_+$ is 
\beq (v,\Lambda)  : \bM \ni p \mpasto o + v + \Lambda (p-o) \in \bM\:.\label{AAP}\eeq
This action also defines the structure of  semidirect product $\cP_+ = \sV \rtimes  \cL_+$.

A {\bf Minkowski reference frame}, physically corresponding to an {\em inertial observer}, is defined by a future-directed timelike {\em unit} vector $n$. The set of these unit vectors will be denoted by $\sT_+\subset \sV_+$.
Take $o\in \bM$ and a {\bf Minkowski basis}, i.e., $e_0,e_1,e_2,e_3 \in \sV$  such that $e_0\in \sT_+$ and $\gor(e_\mu,e_\nu)=\eta_{\mu\nu}$, where $[\eta_{\mu\nu}]_{\mu,\nu=0,1,2,3}:= \mbox{diag}(-1,1,1,1)$.
The global (bijective Cartesian)  chart $\bM \ni p \mapsto (x^0(p), x^1(p),x^2(p),x^3(p)) \in  \bR^4\quad \mbox{such that} \quad  p= o + \sum_{\mu=0}^3 x^{\mu}(p) e_\mu$
is a  {\bf Minkowski chart} on $\bM$ (with {\bf origin} $o$ and {\bf axes} $e_0,e_1,e_2,e_3$) by definition. 
The  vectors  $\partial_{x^\mu}|_p \in T_p\bM$ of the local bases associated to the coordinates are mapped to $e_\mu$ by $\imath_p$ and it holds both  $\gor(\partial_{x^\mu},\partial_{x^\nu}) = \eta_{\mu\nu}$ 
and $\gor^\#(dx^\mu,dx^\nu) =\eta^{\mu\nu}:=  \eta_{\mu\nu}$ constantly everywhere. These identities imply in particular  that the metric tensor field $\gor$ is {\em globally flat}.
A Minkowski chart $x^0,x^1,x^2,x^3$ is {\bf adapted} to a Minkowski reference frame $n\in \sT_+$ if $\partial_{x^0}=n$ everywhere (where again  the natural isomorphism $\imath_p$ is understood).
%
%
A {\bf rest space} $\Sigma_n$ of $n\in \sT_+$ is any  $3$-dimensional plane orthogonal to $n$.  Rest spaces are  smooth spacelike 3-dimensional embedded submanifolds. 
The surjective coordinate function $x^0: \bM \to \bR$ of a Minkowski chart adapted to $n$  defines a {\bf global time coordinate} of $n$. 
The possible global time coordinates of $n$ are defined up to an arbitrary additive constant.
 If a Minkowski chart $x^0,x^1,x^2,x^3$ is adapted to $n$, the {\bf time slices} $$ \bR^3_{x^0_0}:=\{(x^0_0,x^1,x^2,x^3)\:|\: (x^1,x^2,x^3)\in \bR^3\}$$  at constant time $x^0=x_0^0$ 
are the coordinate representation of the rest spaces $\Sigma_{n, x^0_0}$  of $n$ (adopting the notation used in \cite{M2}). 
The rest spaces of $n$ are, in fact,  bijectively labelled by the values of  $x^0$ itself  and  the remaining coordinates $x^1,x^2,x^3$ define an admissible  global chart on each submanifold  $\Sigma_{n,x^0}$. This chart is an orthogonal Cartesian coordinate system on  $\Sigma_{n,x^0}$ with respect to the (Euclidean) metric induced by ${\bf g}$ on $\Sigma_{n,x^0}$ and the affine structure induced by the one of $\bM$. So that, for instance, the {\em Lebesgue $\sigma$-algebra} $\cL(\Sigma_{n,x^0})$ and the {\em Lebesgue measure} are univocally defined on each $\Sigma_{n,x^0}$ independently of the Minkowski chart adapted to $n$.  The Lebesgue measure  on  $\Sigma_{n,x^0}$ coincides with the completion of the Borel  measure canonically induced  by ${\bf g}$ on its embedded submanifolds of $\bM$ (see (\ref{muS}) below).

Given Minkowski chart $x^0,x^1,x^2,x^3$, the {\bf spatial components}  of $k\in T_p\bM \equiv \sV$ along the local basis 
$\partial_{x^0}|_p, \partial_{x^1}|_p,\partial_{x^2}|_p,\partial_{x^3}|_p \in T_pM$
(i.e., the Minkowski basis $e_0,e_1, e_2,e_3\in \sV$ associated to the said Minkowski chart) are  $\vec{k}:=(k^1,k^2,k^3)$ and  the {\bf temporal component} is $k^0$. 
 Referring to the spatial components of  $k,p \in \sV$, their inner product in  $\bR^3$ is again denoted  by the dot $\vec{k}\cdot \vec{h}$.

An important feature of Minkowski spacetime $\sM$ is that the structure  of $J^\pm(A)$ and $I^\pm(A)$ simplify because   $\bM$ is convex.   The following result is true, whose elementary proof is left to the reader.

\begin{proposition} \label{JJ} If $A\subset \sM$, the following holds.
\begin{itemize}
\item[(a)]
 $I^+(A)$ is made of the points $q \in \sM$ such that there is $p\in A$ for that the  $q-p\in \sV_+$.
 \item[(b)] $J^+(A) $ is made of the points $q \in \sM$ such that  there  is $p\in A$ for that  $q-p\in \overline{\sV_+}$ (so $p=q$ is admitted).
\end{itemize}
Analogous  facts are true for $J^-(A)$ and $I^-(A)$.
\end{proposition}

\begin{proposition}[\cite{ON}] {Minkowski spacetime is globally hyperbolic since  the rest spaces   of every  Minkowski reference frame  are (spacelike) Cauchy surfaces. }
\end{proposition}

\begin{definition} {\em A {\bf spacelike Cauchy surface} in $\bM$ is a  spacelike {\em smooth} Cauchy surface of $(\bM, {\bf g})$.   Their family is denoted by $\cC^s_\bM$.\\
A  {\bf spacelike flat  Cauchy surfaces} in $\bM$ is  a rest space of any Minkowski reference frame. Their  family will be  denoted by $\cC^{sf}_\bM$.  \hfill $\blacksquare$}
\end{definition}

\begin{remark} {\em In the following, $S$ will denote a smooth Cauchy surface of $\bM$, i.e., an element of $\cC_\bM$.  However,  {\em in case $S$ is flat and spacelike}, i.e. $S\in \cC^{sf}_\bM$, we shall  very often use the symbol $\Sigma$ in place of $S$, especially when viewing them as $3$-rest spaces of inertial observers  according to the notation of \cite{M2}. \hfill $\blacksquare$}
\end{remark}


In addition to  spacelike flat Cauchy surfaces, there are many other types of smooth Cauchy surfaces in $\bM$. The following proposition concern their description.

\begin{proposition}\label{propS}  Consider a given Minkowski chart $x^0,x^1,x^2,x^3$ and  $S\in \cC_\bM$.
\begin{itemize}
\item[(a)] $S$ is  determined by a smooth map $x^0 = t_S(\vec{x})$, $\vec{x}\in \bR^3$ with \beq
|\nabla t_S(\vec{x})| \leq 1\quad \mbox{for every  $\vec{x}\in \bR^3$.  Here,  $<$ replaces $\leq$ if $S$ is spacelike.} \label{DIS2}\eeq

\item[(b)] $S$ is diffeomorphic to $\bR^3$ as the coordinates $(x^1,x^2,x^3)\in \bR^3$ define a  global chart on it.
\end{itemize}
\end{proposition}

\begin{proof} See Appendix \ref{APPENDIXA}. \end{proof}

Let us now focus on  the case where  $S\in \cC^s_\bM$ , i.e., the smooth Cauchy surface $S$ is spacelike.
\begin{itemize}
\item[(1)] The future-directed unit normal  vector and co-vector  to $S$ at $(t_S(\vec{x}), \vec{x})$ are respectively 
\beq n_S(\vec{x}) := \frac{\partial_{x^0} + \sum_{k=1}^3  \frac{\partial t_S}{\partial x^k} \partial_{x^k}}{\sqrt{1 -  |\nabla t_S(\vec{x})|^2}}\:,
\quad \gor(n_S, \cdot)(\vec{x}) := \frac{-dx^0 + \sum_{k=1}^3  \frac{\partial t_S}{\partial x^k} dx^k}{\sqrt{1 -  |\nabla t_S(\vec{x})|^2}}\:.\eeq
\item[(2)] The Riemannian metric $h^S$ induced on $S$ and represented in terms of the local coordinates $\vec{x} \in \bR^3$ is the pullback of $g$ through the embedding function. It reads
\beq h^S_{ab} := \sum_{\alpha \beta=0}^3 \eta_{\alpha\beta} \frac{\partial x^\alpha}{\partial x^a} \frac{\partial x^\beta}{\partial x^b} = \delta_{ab} - \frac{\partial t_S}{\partial x^a} \frac{\partial t_S}{\partial x^b}\eeq
where $a,b =1,2,3$ and $x^0= t_S(\vec{x})$.
\item[(c)] Correspondingly,  the Borel measure on $S$  induced by the metric is equivalent to the Lebesgue measure $d^3x$ on $\bR^3$ (restricted to the Borel sets) and reads\footnote{We take advantage of $\det(I + cd^t) = 1+ d^tc$ for $c,d \in \bR^n$ and $I$ the identity of $M(n,\bR)$. }
\beq \label{muS} \nu_S(A) := \int_{A} \sqrt{\det h^S}\:  d^3 x  =  \int_{A} \sqrt{1- |\nabla t_S(\vec{x})|^2}\:  d^3 x \quad \mbox{for every  $A\in \cB(S)$},\eeq
where (and occasionally henceforth) the integration region  $A$ in the integrals actually  denotes the projection of $A\in \cB(S)$ onto $\bR^3$.
\end{itemize}

\noindent In case  the smooth Cauchy surface $S$ is generic, we can still define {\em non vanishing}, {\em causal} and {\em future-directed} normal  vectors and co-vectors  to $S$ at $(t_S(\vec{x}), \vec{x})$ but, generally speaking,  we cannot normalize them 
\beq v_S(\vec{x}) := \partial_{x^0} + \sum_{k=1}^3  \frac{\partial t_S}{\partial x^k} \partial_{x^k}\:,
\quad \gor(v_S, \cdot)(\vec{x}) := -dx^0 + \sum_{k=1}^3  \frac{\partial t_S}{\partial x^k} dx^k\:.\label{Vv}\eeq
The induced metric turns out to be degenerate where these vectors are lightlike.

\section{Spacelike Cauchy localization observables in $\bM$ and a general causality condition}\label{SECCP}
We pass to present our generalized notion of spatial localization  and a corresponding generalized causality condition. The latter   extends  CC stated in Definition \ref{DEFC0}. The general notion of spatial localization will be given in terms of families of POVMs on spacelike Cauchy surfaces. This notion  extends  the analogous 
Definition \ref{SLO1}, where only flat Cauchy surfaces were considered.

\subsection{A general notion of spatial  localization in terms of POVMs}

The POVMs we shall use will be defined on a {\em completion} of the Borel 
$\sigma$-algebra on $S\in \cC_\bM^s$. This is necessary because the {\em region of influence} $\Delta' \subset S'$ of a set $\Delta \in \cB(S)$ is not necessarily in $\cB(S')$, but it necessarily  stays in the  $\nu_{S'}$ completion of $\cB(S')$ (see below). This fact was already true \cite{C0,C,M2} when dealing with flat Cauchy surfaces $\Sigma$,  where the said completion of $\cB(\Sigma)$ was nothing but  the Lebesgue $\sigma$-algebra.
If  $S\in \cC_\bM^s$, we denote by    \beq \cM(S):= \overline{\cB(S)}^{\nu_S}\:, \quad \overline{\nu_S}\label{COMPBn}  \eeq the  completion  $\sigma$-algebra of $\cB(S)$ and the completion measure of $\nu_S$ with respect to the positive Borel measure $\nu_S$ (\ref{muS}) induced on $S$  by the spacetime metric.

\begin{definition}\label{GSL} {\em Consider a quantum system described in the Hilbert space $\cH$. A   {\bf  spacelike Cauchy localization observable}  (for short {\bf spacelike Cauchy localization}) of the system  in $\bM$ is a family of normalized $\cH$-POVMs  $\sA :=\{\sA_S\}_{S\in \cC_\bM^s}$
where $\sA_S : \cM(S) \to \gB(\cH)$, such that
\begin{itemize}
\item[(a)] $\sA$ satisfies the {\bf coherence condition}
\beq\label{COHERENCE}
\sA_S(\Delta) = \sA_{S'}(\Delta) \quad \mbox{if  $S,S'\in \cC^s_\bM$, $\Delta\subset S\cap S'$ and  $\Delta \in \cM(S) \cap \cM(S')$;}
\eeq 
\item[(b)] $\sA_S$ is absolutely continuous with respect to $\overline{\nu_S}$, for every $S\in \cC_\bM^s$ (in formulae $\sA_S <\sp < \overline{\nu_S}$).
\end{itemize}  \hfill $\blacksquare$}
\end{definition}

\noindent Physically speaking, requirement (b)  $\sA_S <\sp < \overline{\nu_S}$ means that there is no chance to find a particle in a spatial region  with zero measure. This condition will play a crucial technical role in establishing Theorem \ref{TONE1}. This property is closely related
to the analogous  property of measurs associated to {\em spatial} localization observables  as $\overline{\nu_S}$ is equivalent to the image of the Lebesgue
measure on $\bR^3$.

When restricting to flat Cauchy surfaces, the definition of spacelike Cauchy localization observable  boils down to the definition of spatial localization observable as in  Definition \ref{SLO1}.

\begin{remark} \label{REMREM2}
 {\em 
 \begin{itemize} 
 \item[(1)]  Definition \ref{GSL}, in principle,  is valid for a generic  quantum system and not necessarily for a Wigner particle.
 \item[(2)]  A different  approach \cite{C2} in defining our  localization observable would concern  an assignment of effects on a suitable family of acausal subsets of $\bM$, without declaring that they gives rise to families of  POVMs on the Cauchy surfaces of $\cC^s_\bM$, but recovering this fact at a more advanced stage  of the theory.   This physically deeper  approach would avoid to impose  the coherence condition (\ref{COHERENCE}), since it  would be encapsulated in the formalism automatically. On the other hand, this perspective would turn out technically involved when  proving Theorem \ref{TONE1} below, in view of used mathematical technology which relies upon features of  smooth  spacelike Cauchy surfaces.\hfill $\blacksquare$
 \end{itemize}
 }\end{remark}

Since the examples of  POVMs we shall consider later are initially  defined on $\cB(S)$, the following extension results are of relevance to our work.

\begin{proposition}\label{ACM} Let  $\sA_S$ be a normalized $\cH$-POVM defined on $\cB(S)$ for $S\in \cC_\bM^s$. If $\sA_S <\sp < \nu_S$, then there exists  
 a unique  normalized $\cH$-POVM $\tilde{\sA}_S$ on $\cM(S)$ which extends $\sA_S$ and such that 
$\tilde{\sA}_S<\sp < \overline{\nu_S}$.
\end{proposition}

\begin{proof} See Appendix \ref{APPENDIXA}. \end{proof}

\begin{proposition}\label{PROPCOEES} Suppose that the family of POVMs $\{\sA_S\}_{S\in \cC_\bM^s}$, where  $\sA_S : \cB(S) \to \gB(\cH)$,  satisfies 
\begin{itemize}
\item[(a)]$\sA_S(\Delta) = \sA_{S'}(\Delta)$ for every  $\Delta  \in \cB(S) \cap \cB(S')$;
\item[(b)]  $\sA_S <\sp < \nu_S$. 
\end{itemize}
 Then, extending each POVM according to Proposition \ref{ACM}, we obtain a spacelike Cauchy localization observable $\{\tilde{\sA}_S\}_{S\in \cC_\bM^s}$ . 
\end{proposition}

\begin{remark} {\em In the rest of the paper we shall use the same symbol $\sA_S$ also for the extension to $\cM(S)$ above denote by $\tilde{\sA}_S$. \hfill $\blacksquare$.}
\end{remark}

 In principle, it would be  possible to define from scratch  a  notion of Cauchy localization for smooth Cauchy surfaces $S$ {\em which are not necessarily spacelike} on the corresponding Borel $\sigma$-algebra $\cB(S)$. As a matter of fact, we shall define such type of families of POVMs $\sT^g$ in Sect.\ref{SECPVMC}.  However Theorem \ref{TONE1} below needs several general results about {\em spacelike} smooth Cauchy surfaces. The generalisation of the results presented  in this section to families of POVMs defined on generic Cauchy surfaces (also non-smooth) will be investigated elsewhere.

\subsection{The general causality condition for spacelike Cauchy localization observables}
We can state a natural generalization of CC. 

\begin{definition}\label{DEFC} {\em  
 If  $S,S'\in \cC_\bM^s$ and $\Delta \in \cM(S)$, its {\bf region of influence on $S'$} is
\beq
\Delta' := (J^+(\Delta) \cup J^-(\Delta)) \cap S' \label{2Delta'}\:.
\eeq 
A spacelike Cauchy localization  $\sA$ satisfies the {\bf general causality condition} ({\bf GCC}) if,  for every $S, S' \in \cC_\bM^s$,
\beq
\sA_{S'}(\Delta') \geq \sA_S(\Delta) \label{CC}
\eeq
when   $\Delta \in \cM(S)$
 satisfies
\beq
\Delta' \in \cM(S')\:. \label{CC2}
\eeq \hfill $\blacksquare$}
\end{definition}

 If $S =\Sigma$ and $S'= \Sigma'$ are spacelike {\em flat} Cauchy surfaces  of Minkowski spacetime,  the completed measures $\overline{\nu_\Sigma}$ and $\overline{\nu_{\Sigma'}}$ are nothing but the Lebesgue measures on $\Sigma$ and $\Sigma'$. In this situation, as already said,  it turns out that \cite{C0} $\Delta' \in \cM(\Sigma')$ whenever $\Delta \in \cM(\Sigma)$. We do not  known if this fact is general. It is however possible to prove the following   fact, whose  proof 
is inspired by analogous ideas and proofs in \cite{C2}.  To this end, we 
observe that $S\in \cC_\bM$ is a {\em Polish space} \cite{Cohn}. It can be proved in various ways, the most economic  way  is to use the diffeomorphism of $S$ and  $\bR^3$ according to Proposition \ref{propS}.
%
 
 \begin{proposition}Consider  $S,S'\in \cC_\bM^s$.  If $\Delta\in   \cB(S)$, then $\Delta' \in  \cM(S')$.
 \end{proposition}
 
 \begin{proof} 
Define the following continuous  function
$$\eta : S \times S' \ni (p,q) \mapsto (p, {\bf g}(p-q,p-q)) \in S \times \bR\:.$$
The set $\Delta\times (-\infty, 0]$ is in $\cB(S \times \bR)$ trivially. Since $\eta$ is Borel measurable (as it is continuous), $$\Delta \times \Delta' = \eta^{-1}(\Delta\times (-\infty, 0])$$
is  in $\cB(S\times S')$.  A this point, we can use Proposition 8.4.4 in \cite{Cohn} obtaining that the projection of the above set onto $S'$ is {\em universally measurable}: it stays in the $\sigma$-algebra obtained by completing $\cB(S')$ with respect to any positive finite Borel measure  on $S'$.
To conclude we prove that $\Delta'$ must therefore belong to $\cM(S')$.  Upon the identification of $S'$ with $\bR^3$ through the global chart defined in (b) of Proposition \ref{propS}, consider the Borel finite measure  $\nu_n(E) = \int_{E} \chi_{B_n} \sqrt{1- |\nabla t_{S'}|^2} d^3x$, where $B_n$ is the open ball of radius $n\in \bN$ and center the origin of $\bR^3$.
$\nu_n|_{\cB(B_n)}$ is  equivalent to the $\bR^3$ Lebesgue measure restricted to $B_n$.
Since $\Delta'\cap B_n \in \overline{\cB(S')}^{\nu_n}$, due to Lemma \ref{LEMMAM0}, it must be  $\Delta'\cap B_n = E_n \cup N_n$ where $E_n \in \cB(B_n) \subset \cB(S')$ and $N_n \subset Z_n$ with $Z_n \in  \cB(B_n) \subset \cB(S')$ and $|Z_n|=0$, that is $\nu_{S'}(Z_n)=0$.
In summary, $\Delta' = \cup_{n\in \bN} E_n \cup N_n = ( \cup_{n\in \bN} E_n) \cup (\cup_{n\in \bN} N_n)$ with $\cup_{n\in \bN} E_n \in \cB(S')$ and $\cup_{n\in \bN} N_n\subset \cup_{n\in \bN} Z_n \in \cB(S')$, where $0\leq \nu_{S'}(\cup_{n\in \bN} Z_n ) \leq \sum_{n} \nu_{S'}(Z_n)=0$, so that $\nu_{S'}(\cup_{n\in \bN} Z_n )=0$. By construction $\Delta' \in \overline{\cB(S')}^{\nu_{S'}}=: \cM(S')$.
\end{proof}

\subsection{The general causality condition is valid for every spacelike Cauchy localization observable}

We pass to prove that GCC is actually valid for every spacelike Cauchy localization. The proof consists of some steps. The general structure of the demonstration enjoys many similarities  with some proofs originally introduced  in \cite{M2} and later generalized in \cite{C}. However, to extend these constructions to generic spacelike Cauchy surfaces,  we shall also take advantage of   some fundamental achievements   by Bernal and S\'anchez  \cite{BS,BS3} here specialized to Cauchy surfaces of  Minkowski spacetime.

\begin{lemma} \label{LEMMADEFCOMP} Consider a Minkowski chart on $\bM$ and consider a pair of spacelike  smooth $3$-dimensional embedded manifolds  $S,S'$ respectively described as the sets 
$$S := \{(x^0=t_S(\vec{x}), \vec{x})\:|\: \vec{x}\in \bR^3 \}\:,\quad  S' := \{(x^0=t_{S'}(\vec{x}), \vec{x})\:|\: \vec{x}\in \bR^3 \}$$
for a pair of smooth functions $t_S : \bR^3 \to \bR$ and  $t_{S'} : \bR^3 \to \bR$ (where necessarily $|\nabla t_S|<1$ and $|\nabla t_S|<1$ everywhere).
If $S\in \cC_\bM^s$ and $t_{S'}(\vec{x})= t_S(\vec{x})$ outside a compact set in $\bR^3$, then $S'\in \cC_\bM^s$  as well.
\end{lemma}

\begin{proof}  $S'$ is diffeomorphic to $\bR^3$ through the projection map $S'\ni (x^0, t_{S'}(\vec{x})) \mapsto \vec{x}\in \bR^3$. Therefore $S'$ is connected and closed as it is the preimage of the closed set $\{0\}$
according to the continuous (actually smooth) map $f:  \bR^4 \ni (x^0, \vec{x}) \mapsto  x^0-t_{S'}(\vec{x}) \in \bR$.
  Let us define $T= \max \{ |t_S(\vec{x})- t_{S'}(\vec{x})| \:|\: \vec{x} \in \bR^3\}$. Notice that $T$ is finite because the function is continuous and compactly supported. Since the curves tangent to the $x^0$ axis are timelike, we have that $S' \subset I^+(S_{-2T}) \cap I^-(S_{2T})$, where 
$$S_\tau := \{(x^0=\tau+ t_S(\vec{x}), \vec{x})\:|\: \vec{x}\in \bR^3 \}\:.$$
Similarly, $S_{-2T} \subset I^-(S_{2T})$ and also $S_{\pm 2T}$ are spacelike Cauchy surfaces because they are obtained trough isometries out of the spacelike Cauchy surface $S$. According to (d1) in Remark \ref{RCSu}, $S'$ is a spacelike Cauchy surface. \end{proof}

\begin{proposition}  \label{PFIN1} 
If $\sA$ is a spacelike Cauchy localization,  then (\ref{CC2}) and (\ref{CC}) are true when $\Delta \subset S \in \cC_\bM$ is a spacelike  (embedded) compact smooth $3$-submanifold with boundary.
\end{proposition}

\begin{proof}  We henceforth choose a Minkowski global chart and we describe $S$ and $S'$  through the global chart as in  (b) of Proposition \ref{propS}.  In particular  $\psi' : S' \ni p \mapsto  \vec{x}(p) \in \bR^3$  will be denote the said   global chart on $S'$.\\
As $\Delta$ is compact, it belongs to $\cM(S)$. Furthermore
$\Delta' = (J^+(\Delta) \cup J^-(\Delta)) \cap S' \label{Delta'}$ is compact as well  (Corollary A.5.4 \cite{BGP}), so that it belongs to $\cM(S')$.
Next we consider a sequence of open sets $F'_n \subset S'$ such that (a) $ \partial F'_n \cap \Delta' = \emptyset$,  (b) $\overline{F'_{n+1}}\subset F'_n$, and (c) $\cap_{n\in \bN} F'_n = \Delta'$. This sequence is recursively  constructed by extracting a finite subcovering of a covering of $\Delta'$  made of coordinate open balls $\psi'^{-1}(B_{r_n}(\psi'(p_i)))$ of radius $r_n$, $p_i \in \Delta'$, where $r_n \to 0$ for $n\to +\infty$, and
$r_{n+1} < \mbox{dist}(\Delta', \partial F'_n)$, where $\mbox{dist}(A,B) := \inf_{(x,y)\in A\times B} \mbox{dist}(\psi'(x),\psi'(y))$ for $A,B\subset S'$. The latter distance being the standard distance of couples of points in $\bR^3$ (notice that $ \mbox{dist}(\Delta', \partial F'_n)$  is strictly positive and finite by construction).  The boundary of each $F'_n \subset S'$ is a $C^0$ submanifold which is also smooth up to a zero-measure subset of the boundary of the balls made of part of the intersections of the boundaries of  a finite number of open balls. Each
 $\partial F'_n$  can be therefore slightly locally smoothed  in order to  transform  each $F'_n$ to a corresponding open set  $R'_n$, such that  
$\overline{R'_n}$ is a compact  smooth submanifold with boundary of the spacelike Cauchy surface $S'$ and, as before, $\Delta' \cap  \partial R'_n=\emptyset$,
(b) $\overline{R'_{n+1}}\subset R'_n$, and (c) $\cap_{n\in \bN} R'_n = \Delta'$.
Let us now focus on the further family of open relatively compact sets
$$dR'_n := R'_{n}\setminus \overline{R'_{n+1}}\:.$$
$\overline{dR'_n}$ is a compact smooth submanifold with boundary of $S'$ and  $\overline{dR'_n}$ has no intesection with $J^+(\Delta) \cup J^-(\Delta)$ by construction. Since  $\Delta$ is a smooth submanifold with boundary of the spacelike Cauchy surface $S$, we conclude that  $\Delta \cup \overline{dR_n'}$ is a (non connected) spacelike and acausal compact smooth $3$-submanifold with boundary  of $\bM$. According to (d2) in Remark \ref{RCSu},
 there is a spacelike Cauchy surface $S''_n$ which includes $\Delta \cup \overline{dR_n'}$.
Finally consider the set $$S_n:= \{(x^0 = t_{S_n}(\vec{x}), \vec{x})\:|\: \vec{x}\in \bR^3 \}\:,$$
where  the map $t_{S_n}: \bR^3 \to \bM$ is constructed as follows:
$$t_{S_n}(\vec{x}) := \begin{cases} t_{S''_n}(\vec{x})\quad \mbox{if $\vec{x} \in \psi'(R'_n)$,} \\
 t_{S'}(\vec{x})\quad\:  \mbox{if $\vec{x} \not\in \psi'(R'_n)$.} \end{cases}$$
This map is smooth by construction and $|\nabla t_{S_n}(\vec{x})| <1 $ everywhere, in particular $dx^0 - dt_{S_n} \neq 0$ everywhere,  so that the $S_n$ of $t_{S_n}$ is a spacelike smooth submanifold. This is a spacelike compact deformation of the Cauchy surface $S'$ as it coincides to it outside the compact $\overline{R'}_n$. According to Lemma \ref{LEMMADEFCOMP}, $S_n$ is a spacelike Cauchy surface as well.  \\
To go on, we pass to consider the two normalized $\cH$-POVM $\sA_S$ and $\sA_{S'}$. We can decompose $$I=\sA_{S_n}(S_n)= \sA_{S_n}(\Delta) + \sA_{S_n}\left({\psi''_n}^{-1}\circ\psi'\left(R'_n\right)\setminus\Delta\right) + \sA_{S_n}(S_n \setminus {\psi''_n}^{-1}\circ\psi'\left(R'_n\right))$$
where $\psi''_n : S''_n \to \bR^3$ is the global chart on $S''_n$ according to (b) of Proposition \ref{propS}. Similarly,
$$I=\sA_{S'}(S')= \sA_{S'}(R'_n) + \sA_{S'}(S' \setminus R'_n)\:.$$
We know that $S_n \setminus {\psi''_n}^{-1}\circ\psi'\left(R'_n\right) =S' \setminus R'_n$ because $S'$ and $S_n$ coincide thereon. As a consequence of the coherence condition (\ref{COHERENCE}),
$ \sA_{S'}(S' \setminus R'_n)=\sA_{S_n}({S_n \setminus\psi''_n}^{-1}\circ\psi'( R'_n))$ so that,
$$\sA_{S_n}(\Delta) + \sA_{S_n}({\psi''_n}^{-1}\circ\psi'(R'_n)\setminus \Delta)= \sA_{S'}(R'_n)$$
Using again the coherence condition with $\Delta \subset S \cap S_n$, we have obtained that
$$\sA_{S}(\Delta) + \sA_{S_n}({\psi''_n}^{-1}\circ\psi'(R_n\setminus \Delta))= \sA_{S'}(R'_n)$$
so that, for every $\psi\in \cH$, we achieve the relation between positive finite  measures
$$\langle \psi| \sA_{S}(\Delta)\psi\rangle +\langle \psi| \sA_{S_n}({\psi''_n}^{-1}\circ\psi'(R'_n\setminus \Delta))\psi\rangle= \langle \psi|\sA_{S'}(R'_n)\psi\rangle\:.$$
At this juncture, 
$\langle \psi|\sA_{S'}(R'_n)\psi\rangle <+\infty$,
 $\cap_{n\in \bN} R'_n = \Delta'$, $R'_{n+1}\subset R'_n$,  and external continuity yield
 $$\langle \psi| \sA_{S}(\Delta)\psi\rangle +\inf_{n\in \bN}\langle \psi| \sA_{S_n}({\psi''_n}^{-1}\circ\psi'(R'_n\setminus \Delta))\psi\rangle= \inf_{n\in \bN}\langle \psi|\sA_{S'}(R'_n)\psi\rangle= \langle \psi|\sA_{S'}(\Delta')\psi\rangle\:.$$
As $\inf_{n\in \bN}\langle \psi| \sA_{S_n}({\psi''_n}^{-1}\circ\psi'(R'_n\setminus \Delta))\psi\rangle \geq 0$, we have proved that
 $$\langle \psi| \sA_{S}(\Delta)\psi\rangle \leq  \langle \psi|\sA_{S'}(\Delta')\psi\rangle\:.$$
This is the thesis by arbitrariness of $\psi \in \cH$. \end{proof}

\begin{proposition} \label{PGY}
If $\sA$ is a spacelike Cauchy localization,  then (\ref{CC2}) and (\ref{CC}) are true when $\Delta \subset S \in \cC^s_\bM$ is open.
\end{proposition}

\begin{proof}  Consider $S,S' \in \cC^s_\bM$ and refer to the definition (\ref{2Delta'}) of $\Delta'$.  Generally speaking, if $\Delta \subset S$ is open set, from Prop. \ref{JJ} and the fact that $S$ is a smooth submanifold (a $C^0$ embedded submanifold is enough actually), it is not difficult to prove that
$J^\pm(\Delta)$ are open sets\footnote{This is generally false in other spacetimes!} in $\bM$. Hence  $\Delta' \subset S'$ is open as well. Therefore  (\ref{CC2})  is true in the considered case.  \\
Let us denote by $\psi : S \to \bR^3$ the global chart on $S$ constructed through a function $t_S$ in a Minkowski chart $x^0,\vec{x}$ as prescribed in (b) of Proposition \ref{propS}. We shall use this identification to see the points of $S$ as points of $\bR^3$. This identification enjoys an important property.
The {\em completed} measures $\overline{\nu}_S$ and the standard Lebesgue measure on $\bR^3$ turn out to be {\em equivalent} through the identification $\psi$ as a straightforward consequence of the fact that, in  (\ref{muS}), the map $\bR^3 \ni \vec{x} \mapsto \sqrt{1-|\nabla t_S(\vec{x})|^2}$ is continuous and strictly positive. The impact of this remark relies upon the following known result  (Theorem  1.26 in \cite{EG}).  If $A \subset \bR^n$ is open,  $\delta>0$, there exist a countable collection $\{B_j\}_{j=1,2,\ldots }$  of disjoint (non-empty) closed balls $B_j \subset A$ with diameter less than  $\delta$, such that $$  \left|A \setminus \bigcup_{j\in \bN} B_j\right|  =0\:,$$
the bar denoting the Lebesgue measure on $\bR^3$.
Passing to $S$, it means that (for every $\delta>0$) there exist a countable  family of pairwise disjoint  compact smooth  spacelike submanifolds with boundary $\Delta_j := \psi^{-1}(B_j) \subset \Delta \subset S$, such that 
$$
\overline{\nu}_S \left( \Delta \setminus \bigcup_{j\in \bN} \Delta_j\right) = 0\:.
$$
Therefore, since the positive measure  $\langle \psi|\sA_S(\cdot)\psi \rangle$ is finite and is absolutely continuous with respect to $\nu_S$ according to Prop. \ref{ACM}:
\beq
\langle \psi|\sA_S(\Delta)\psi \rangle = \left\langle \psi \left|\sA_S\left(   \bigcup_{j\in \bN} \Delta_j\right) \right. \psi\right\rangle\:,\label{NEW}
\eeq
if $\psi\in \sH$.
If we define the compact set $\Delta_N := \bigcup_{j=0}^N \Delta_j \subset \Delta$ and $\Delta'_N$ correspondingly, we can apply Proposition \ref{PFIN1} obtaining that, for every $\psi \in \cH$,
\beq
\langle \psi|\sA_{S'}(\Delta'_N) \psi\rangle  \geq \langle \psi|\sA_{S'}(\Delta_N) \psi\rangle\:.
\eeq
As both sequences are non-decreasing, the limit for $N\to +\infty$ exist and, using inner continuity,  (\ref{NEW}), and monotony, we find
$$\langle \psi|\sA_{S'}(\Delta') \psi\rangle  \geq \lim_{N\to +\infty}\langle \psi|\sA_{S'}(\Delta'_N) \psi\rangle  \geq
\lim_{N\to +\infty} \langle \psi|\sA_{S'}(\Delta_N) \psi\rangle =  \langle \psi|\sA_{S'}(\Delta) \psi\rangle\:, \quad \forall \psi \in \cH\:.$$
This is the thesis by arbitrariness of $\psi \in \cH$.\end{proof}

\begin{proposition} \label{PGY2}
If $\sA$ is a spacelike Cauchy localization,  then (\ref{CC2}) and (\ref{CC}) are true when $\Delta \subset S \in \cC_\bM$ is compact.
\end{proposition}

\begin{proof}  (\ref{CC2})  is valid just because,  as already observed in the proof of Proposition \ref{PFIN1}, $\Delta'$ is compact if $\Delta$ is. So we have to establish the validity of
 (\ref{CC}) only.   We refer to a global chart $\psi : S \to \bR^3$  on $S$ constructed out of the map $t_S$ in a Minkowski chart $x^0,\vec{x}$ as in  (b) of Proposition \ref{propS}.     If $\Delta \subset S$ is compact  we can construct a sequence of open sets $F_n \supset \Delta$ such that $F_{n+1} \subset F_n$ and $\cap_{n\in \bN} F_n = \Delta$. $F_n$ is the finite union of coordinate balls of radius $r_n \to 0$ such that $r_{n+1}<\mbox{dist}(\Delta, F_n)$ where we adopted the same notation as in the proof of Prop. \ref{PFIN1}.    In view of Prof \ref{PGY2}, taking account of the external continuity of the involved  measures 
$$\langle \psi| \sA_S(\Delta) \psi \rangle = \lim_{n\to +\infty}\langle \psi| \sA_S(F_n) \psi \rangle \leq \lim_{n\to +\infty}\langle \psi| \sA_{S'}(F'_n) \psi \rangle = \langle \psi| \sA_{S'}(\cap_{n\in \bN} F'_n) \psi \rangle\:.$$
The proof ends if proving that $\cap_{n\in \bN} F'_n= \Delta'$. In, fact, obviously $\cap_{n\in \bN} F_n\supset \Delta'$.
On the other hand, if $p \in \cap_{n\in \bN} F'_n$ then there is a causal segment from $p$ to $q_n\in F_n \subset \overline{F_0}$ which is compact and stays in $S$ because it is closed. As a consequence there is a subsequence $q_{n_k} \to q \in  \overline{F_0}$ for $k\to +\infty$. It is easy to see that 
$q\in \Delta$ (otherwise $q$ would stay at some distance from the compact $\Delta$ and this is not admitted in view of the very construction of the sets $F_{n_k}\ni q_{n_k} \to q$).  The limit $p-q$ of the causal segments $p-q_n$ is still causal or $0$ (as the set of causal vectors and $\{0\}$ is closed). We conclude that $p \in (J^+(\Delta) \cup J^-(\Delta))\cap S = \Delta'$. We have established that 
$\cap_{n\in \bN} F_n\subset \Delta'$, so that $\cap_{n\in \bN} F_n= \Delta'$ concluding the proof.
\end{proof} 

We are in a position to prove the first main result of this work:  spacelike Cauchy localizations always satisfy the 
Cauchy causality requirement in Definition \ref{DEFC}. Thus, in particular, also Castrigiano's causal requirement when restricting to spacelike flat Cauchy surfaces.

\begin{theorem} \label{TONE1} Let $\sA := \{\sA_S\}_{S\in \cC^s_\bM}$ be a spacelike Cauchy localization observable, then it satisfies the general causal condition  in  Definition \ref{DEFC}.\\
More generally, if $S,S' \in \cC_\bM^s$,   $\Delta \in \cM(S)$, and $\psi \in \cH$, then
\beq \langle \psi| \sA_S(\Delta) \psi\rangle \leq \sup \left\{ \langle \psi| \sA_{S'}(K') \psi\rangle\:|\: K \subset \Delta, \:\:K\:  \mbox{compact}\right\}\label{F1}\:,\eeq
even if $\Delta' \not \in \cM(S')$.  (Above $K' := (J^+(K) \cup J^-(K)) \cap S'$ as usual.)
\end{theorem}

\begin{proof}  The positive Borel  measure $\cB(S)\ni \Delta \mapsto \langle \psi |\sA_S(\Delta)\psi \rangle$ is regular (because $S$  is Hausdorff, locally compact  and every open set is a countable union of compacts with finite measure,  according to Theorem 2.18 in \cite{Rudin}).  As a consequence,   the completion measure  $\cM(S)\ni \Delta \mapsto \langle \psi |\sA_S(\Delta)\psi \rangle$ is regular as well (Prop. 1.59 in \cite{Cohn}).
 If $\Delta \in \cM(S)$, internal  regularity yields
  $$ \langle \psi| \sA_S(\Delta) \psi\rangle = \sup \left\{ \langle \psi| \sA_{S}(K) \psi\rangle\:|\:  K \subset \Delta \subset S, \:\: K\:  \mbox{compact}\right\}\:.$$
At this juncture Prop. \ref{PGY2} entails (\ref{F1}).\\
 If we also know that $\Delta' \in \cM(S')$, noticing that the sets $K'$ are compact  and satisfy $K'\subset \Delta'$, but they are not necessarily all compact sets in $S'$ satisfying the latter condition, internal regularity entails
$$ \sup \left\{ \langle \psi| \sA_{S'}(K') \psi\rangle\:|\:  K \subset \Delta \subset S, \:\: K\:  \mbox{compact}\right\} \leq \langle \psi| \sA_{S'}(\Delta') \psi\rangle\:. $$
However  Prop. \ref{PGY2} implies
$$ \sup \left\{ \langle \psi| \sA_{S}(K) \psi\rangle\:|\:  K \subset \Delta \subset S, \:\: K\:  \mbox{compact}\right\} \leq  \sup \left\{ \langle \psi| \sA_{S'}(K') \psi\rangle\:|\:  K \subset \Delta \subset S, \:\: K\:  \mbox{compact}\right\}\:.$$
Therefore   $ \langle \psi| \sA_S(\Delta) \psi\rangle\leq  \langle \psi| \sA_{S'}(\Delta') \psi\rangle$. This is just  (\ref{CC}) due to  arbitrariness of $\psi \in \cH$.  \end{proof}

\begin{corollary}\label{TONE1COR} If a spatial localization observable $\{\sA_\Sigma\}_{\Sigma\in \cC^{sf}_\bM}$ can be extended to a spacelike Cauchy localization observable $\{\sA_S\}_{S\in \cC^s_\bM}$, then the former automatically  satisfies CC in  Definition \ref{DEFC0}.
\end{corollary}

\begin{proof} Consider only spacelike flat Cauchy surfaces, noticing that there $\cM(\Sigma)= \cL(\Sigma)$ and that
 $\Delta' \in \cL(\Sigma')$ for every $\Delta \subset  \Sigma$ (even if $\Delta \not \in \cB(S)$!) as established in \cite{C0}. In that case GCC  in Definition \ref{DEFC} boils down to CC in   Definition \ref{DEFC0}.
\end{proof}

\begin{remark}\label{REMLASTC}
{\em In the proof of Theorem \ref{TONE1},  the coherence condition (\ref{COHERENCE}) of Definition \ref{GSL} is of utmost relevance. This is one of the reasons why we inserted that requirement already in  Definition \ref{GSL}.  However, an alternative approach would consist  of removing this condition from the  Definition \ref{GSL}  and to directly add it in the hypotheses of Theorem \ref{TONE1},  proving again  the validity of GCC.
  At this juncture it is worth stressing that GCC  trivially implies the coherence condition \ref{COHERENCE} (since $\Delta= \Delta'$ if $\Delta \subset S \cap S'$).  The conclusion is that, if using from scratch a weaker version of Definition \ref{GSL}, i.e.,  without the requirement of coherence (\ref{COHERENCE}), the general causal condition  (\ref{CC})-(\ref{CC2}) and the coherence condition (\ref{COHERENCE}) would be  {\em equivalent}\footnote{We are grateful to D.P.L. Castrigiano for this observation.}.} $\hfill \blacksquare$
\end{remark}

\section{Massive KG particles, conserved currents}\label{INTER}
In this section, we shall introduce some basic notion and results to construct examples of spatial Cauchy localization observables.
The former section is a recap about the one-particle structure of the massive real Klein-Gordon particle. The latter deals with some general properties of conserved currents and associated exact volume 3-forms. 

\subsection{One-particle Hilbert space of  real massive KG particles in $\bM$}\label{OPS}

According to Section \ref{SECM}, we fix a preferred origin $o\in \bM$, so that the map  $\sV \ni x \mapsto o+x \in \bM^4$ defines a one-to-one correspondence between points of $\bM$ and vectors in $\sV$ (the differential of this map being  $\imath_p$). This identification is very useful when dealing with the Fourier transformation on $\bM$, where the product $p\cdot x$ enters the play. All the theory developed in this work does not depend on the choice of $o$.
With this structure, the active action (\ref{AAP}) of $(y, \Lambda) \in  \cP_+$ on $\bM$ takes the form
$\bM\ni o + x \mapsto o+ y+ \Lambda x \in \bM\:.$
Given a mass constant $m>0$, the {\bf future mass-shell} in $\sV$ is 
$\sV_{m,+}:= \{k\in V_+ \:|\: k\cdot k = -m^2\}$.
For every  Minkowski chart   $x^0,x^1,x^2,x^3$,
\beq   k\in \sV_{m,+} \quad \Leftrightarrow \quad   k \in \sV  \quad \mbox{and} \quad k^0 = \sqrt{\vec{k}^2+m^2} >0\label{ms}\:.\eeq

The {\em one-particle Hilbert space} of a Klein-Gordon particle of mass $m>0$ is isomorphic to $L^2(\bR^3, d^3p)$ upon the choice of a  Minkowski chart $x^0,x^1,x^2,x^3$. Here $\bR^3$, is the space of momenta $\vec{p}$ and the {\em momentum representation} wavefunctions $\phi=\phi(\vec{p})$ are normalized elements of  $L^2(\bR^3, d^3p)$.
To have a completely covariant formulation, one  takes advantage of  the canonical Hilbert space isomorphism  
\beq
F:  L^2(\sV_{m,+}, \mu_m)  \ni \psi \mapsto \phi_\psi \in L^2(\bR^3, d^3k) \quad \mbox{where} \quad \phi_\psi(\vec{k}): = \frac{\psi(k)}{\sqrt{k^0(\vec{k})}} \label{PHI}
\eeq
There,   
\beq
\cH := L^2(\sV_{m,+}, \mu_m)\quad \mbox{with\:  $d\mu_m(k) := \frac{d^3k}{k^0(\vec{k})}$} \label{COVH}
\eeq
is the  (covariant) {\bf one-particle Hilbert space} of a real Klein-Gordon particle with mass $m>0$, where $\mu_m$  is the $\cP_+$-{\bf invariant measure} on the mass shell $\sV_{m,+}$.
Covariance is here manifest because,
  the standard unitary representation of $\cP_+$ on the one-particle space   \beq \cP_+ \ni h \mapsto U_h\in \gB(\cH)\label{UNI}\eeq   takes the {\em equivariant}  form \beq (U_h \psi)(k) = e^{-ik\cdot y_h} \psi(\Lambda_h^{-1}k)\:, \quad \mbox{for every $k\in \sV_{m,+}$, $\psi \in \cH$,  $h= (y_h,\Lambda_h) \in \cP_+$}\:. \label{UNI2}\eeq
For future convenience, we define the dense subspaces   ${\cal D}(\cH)  \subset {\cal S}(\cH)  \subset \cH$
\beq {\cal D}(\cH) := F^{-1}(C_c^\infty(\bR^3)) \quad \mbox{and}\quad   {\cal S}(\cH) := F^{-1}({\cS}(\bR^3)) \:, \label{D} \eeq
where ${\cS}(\bR^3)$ is the usual {\em Schwartz} space on $\bR^3$. 
It easy to prove that, in view of (\ref{UNI2}), the definition of ${\cal D}(\cH)$ and ${\cal S}(\cH)$ do not depend on the choice of the Minkowski chart used to construct $F$.
If $\psi \in \cD(\cH)$, or more generally $\psi \in {\cal S}(\cH)$, the associated (complex) {\bf covariant wavefunction} is
\beq
\varphi_\psi(x) := 
\int_{\sV_{m,+}}  \frac{\psi(p)e^{i p\cdot x} }{(2\pi)^{3/2}} d\mu_m(p)\:. \label{wavePSI}
\eeq
Notice that $\varphi_\psi \in C^\infty(\bM)$,  it is also bounded with all of its derivatives, and it solves the {\bf Klein-Gordon equation} in $\bM$ \beq (\Box - m^2)\varphi_\psi=0\:, \quad \mbox{where $\Box:= \eta^{\mu\nu} \partial_{\mu}\partial_{\nu}$ in every Minkowski chart.}\label{KG}\eeq
Furthermore, the action of $U$  on $\varphi_\psi$ is straightforward and explains the adjective ''covariant'':
\beq
\varphi_{U_h\psi}(x) = \varphi_\psi(h^{-1}x) \quad \forall h\in \cP_+\:, \forall x\in \bM\:.
\eeq
Due to (\ref{wavePSI}),(\ref{KG}) and Proposition \ref{PROPJ}  below applied to the bounded  conserved  smooth current
$$ J^{\psi, \psi'}_\mu:= \overline{\varphi_{\psi}} \partial_{\mu} \varphi_{\psi'} - \varphi_{\psi'} \partial_{\mu} \overline{\varphi_{\psi}}\:,$$
we have that the scalar product of $\cH$ satisfies, for $\psi,\psi' \in {\cal S}(\cH)$,
\beq
\langle \psi|\psi' \rangle =\frac{i}{2} \int_{S} \left(\overline{\varphi_{\psi}}(x) n^\mu_S(x)  \partial_\mu \varphi_{\psi'}(x) - \varphi_{\psi'}(x) n^\mu_S(x) \partial_\mu\overline{\varphi_{\psi}}(x)\right) \: d\nu_S(x)\:,
\eeq
for every spacelike smooth Cauchy surface $S$ in $\bM$. An analogous formula for generic smooth Cauchy surfaces can be established on account of Propositions \ref{PINT} and \ref{PROPJ} below.

Various issues concerning the possibility to directly or indirectly interpret $\varphi_\psi$ as a wavefunction in classical sense  somehow related to notions of spatial localization, and the failure of these expectations accumulated over the years, were discussed in \cite{M2}.  

\subsection{Conserved quantities on generic smooth Cauchy surfaces of $\bM$}
From now on, we assume that $\bM$ is oriented.  Consider a smooth vector field $J$ in $\bM$. We can associate a $3$-form to it\footnote{$\omega^J$ can equivalently be written as the {\em interior product} of $J$ and the volume $4$-form of $(\bM, {\bf g})$.}, {\em where we henceforth take advantage of the summation convention over repeated  Greek indices},
\beq
\omega^J(x) := \frac{1}{3!} J^\delta(x) \epsilon_{\delta \alpha\beta \gamma} dx^\alpha \wedge dx^\beta \wedge dx^\gamma \label{omegaJ}\:.
\eeq
Above  $\epsilon_{\delta \alpha\beta \gamma}$ is the {\bf Levi-Civita (pseudo)tensor} which, in Minkowski coordinates, coincides to the sign of the permutation $(\delta, \alpha,\beta, \gamma)$ of $(0,1,2,3)$, or vanishes in case of  repetitions. 

 Let $S$ be a smooth Cauchy surface determined by the map $t_S:\bR^3 \to \bR$ in Minkowski coordinates $x^0,x^1,x^2,x^3$ where $v_S$ is the future-directed normal vector to $S$ as in (\ref{Vv}). If 
$J \equiv (J^0, \vec{J})$  is either zero or causal and future directed depending on  $x\in \bM$, then we define
\beq
\int_A \omega^J := \int_A (J^0 - \vec{J}\cdot \nabla t_S) d^3x = -\int_{A} J \cdot v_S d^3x \in [0,+\infty] \:,  \quad \forall A \in \cB(S) \label{VVN}
\eeq
 In the last two integrals, $A$ is interpreted as a subset of  $\bR^3$ according to the canonical projection $\bR^4 \ni (x^0, \vec{x}) \mapsto \vec{x} \in \bR^3$ and the  integrals are  well defined  in $[0,+\infty]$ because
\beq
-J \cdot v_S \geq 0\:.
\eeq

\begin{remark} {\em The integral of the form $\omega^J$  as defined by the right-most term  in (\ref{VVN}) is a standard integral  in the sense of measure theory for a generic Borel set $A$. This fact permits to use standard arguments of measure theory.  On the other hand, if $A$ is compact, the definition agrees with the standard integral of smooth $n$-forms on compact sets and we can take advantage of standard results in this context as the Poincar\'e theorem. \hfill $\blacksquare$}
\end{remark}

We have a preliminary proposition whose elementary proof is left to the reader.

\begin{proposition} \label{PINT} If $J$ is a smooth vector field on $\bM$ and $\omega^J$ is defined as in (\ref{omegaJ}), then the following facts are valid (where (\ref{VF}) has been taken into account)
\begin{itemize}
\item[(1)] $J$ is conserved iff $\omega^J$ is closed:
\beq
\partial_\alpha J^\alpha(x)  =0 \quad \Leftrightarrow \quad  d \omega^J(x) =0\:.
\eeq

\item[(2)]
If $S\in \cC_\bM^s$, then
\beq
\int_A \omega^J = -\int_A J\cdot n_S \: d\nu_S \:. \label{FJ}
\eeq
\end{itemize}
\end{proposition}

\noindent We pass now to state  a folklore statement which is actually technically very useful. The unexpectedly technical  proof is in the appendix.

\begin{proposition} \label{PROPJ} Let $J$ be a smooth vector field on $\bM$ such that
\begin{itemize}
\item[(1)]  it is {\bf conserved}, i.e.,  $\partial_\alpha J^\alpha =0$ everywhere;
\item[(2)]  if $x\in \bM$, then either $J(x)=0$ or $J(x)$ is causal and future-directed, i.e., $J(x) \in \overline{\sV_+}$;
\item[(3)]  it is {\bf bounded}, i.e., its components in a (thus every)  Minkowski chart are  bounded functions.
\end{itemize}
Then, if $S,S' \in \cC_\bM$,
\beq
\int_S \omega^J  =  \int_{S'} \omega^J  \in [0,+\infty] \label{FIRST}\:.
\eeq
If, e.g. $S'\in \cC^s_\bM$, then  the above identity can be re-written as 
\beq
\int_{S} \omega^J =- \int_{S'} J \cdot n_{S'} d\nu_{S'}   \:. \label{SECOND}
\eeq
\end{proposition}

\begin{proof}   See Appendix \ref{APPENDIXA}. \end{proof}

\section{Spacelike Cauchy localization observables out of causal kernels of massive KG particles}\label{SECPVMC}
We now pass to make contact with  a family of normalized POVMs, for a massive real Klein-Gordon particle, defined on spacelike flat Cauchy surfaces and introduced by Gerlach, Gromes and Petzold
 \cite{P} and  Henning, Wolf \cite{HW}  and, much more recently, rigorously studied by Castrigiano in \cite{C}, proving in particular that these POVMs satisfy CC. Next we pass to extend these POVMS to the full family of smooth Cauchy surfaces showing that, when considering only the spacelike Cauchy surfaces,  they induce corresponding spacelike Cauchy localization observables. Before, we have to recall some basic definitions.

\begin{remark}
{\em We stress that the localization observable  we shall construct is more general than the notion of spacelike Cauchy localization. In fact the POVMs we shall construct are also defined for smooth Cauchy surfaces which are not spacelike. This feature  may have some consequences in the analysis  of the structure of  causal regions of $\bM$ in terms of orthocomplemented  lattices \cite{C0}.  \hfill $\blacksquare$}
\end{remark}

\subsection{A notion of spatial  localization for massive KG particles out of  causal kernels}

Consider a Minkowski chart $x^0,x^1,x^2,x^3$.
We want to introduce an $L^2(\bR^3, d^3p)$  POVM rigorously discussed  in \cite{C} -- there named  POL --   on the rest space at time $x^0$  denoted below by $\bR^3_{x^0}$. 
The POVM  $\sT_{\bR^3_{x^0}}$ defines  the position  in $\bR^3_{x^0}$  of a scalar real Klein Gordon particle of mass $m>0$  -- whose pure  states are defined by functions $\phi=\phi(\vec{p})$  in momentum representation with Hilbert space $L^2(\bR^3, d^3p)$. It satisfies the condition 
\beq
\langle \phi | \sT_{\bR^3_{x^0}}(\Delta) \phi\rangle =\frac{1}{(2\pi)^3} \int_\Delta \int_{\bR^3}\int_{\bR^3} \frac{K(\vec{k},\vec{p})}{2\sqrt{p^0k^0}} e^{i(\vec{p}-\vec{k})\cdot \vec{x}- i(p^0-k^0)x^0} \overline{\phi(\vec{k})}\phi(\vec{p})  d^3p d^3k d^3x\:,  \Delta\in \cB(\bR_{x^0}^3)\:.\label{POVMC}
\eeq
Above,
$\vec{x}:= (x^1,x^2,x^3)$, the components  $k^0$ and  $p^0$  are determined by $\vec{k}, \vec{p} \in \bR^3$ as prescribed in (\ref{ms}) and the vectors $\phi$ stay  in a suitable dense  subspace of  $L^2(\bR^3, d^3p)$.
More specifically, referring to the discussion of Sect 11 \cite{C}:

\begin{definition} \label{DEFCP} {\em Consider a Minkowski chart $x^0,x^1,x^2, x^3$ on $\bM$. A   {\bf POL with causal kernel}\footnote{According to  \cite{C}.} on $\bR^3_{x^0}$  is  a  normalized  $L^2(\bR^3, d^3p)$-POVM  $T^g_{\bR^3_{x^0}}$ such that\footnote{In \cite{C}, a further denominator $1/\sqrt{k^0p^0}$ should take place under $k^0+p^0$ in (\ref{KC}). This factor has been moved in (\ref{POVMC}) here, so that the final formulas of the POVM are actually identical. Obviously, this  does not affect the positivity of that kernel.}  
\begin{itemize}
 \item[(a)]  (\ref{POVMC})  is valid when $\phi \in L^1(\bR^3, dk^3)\cap L^2(\bR^3, dk^3)$;
\item[(b)]  $K :=K_g$ is a {\bf causal kernel}, i.e., it  has the structure
\beq
K_g(\vec{k},\vec{p})  :=(k^0+p^0) g(-k\cdot p)\:,\quad k,p \in \sV_{m,+}\label{KC}
\eeq
for  $g: [m^2,+\infty) \to \bR$ which is continuous, normalized to $g(m^2)=1$ and it is such that $K_g: \bR^3 \times \bR^3 \to \bC$ is {\em positive definite}.   \hfill $\blacksquare$ 
\end{itemize}}
\end{definition}

\begin{remark}  
{\em \begin{itemize}
\item[(1)]  We recall the reader that $K: X\times X \to \bC$ is a {\bf positive definite kernel} if
\beq \sum_{i,j=1}^N \overline{c_i}c_j K(k_i, k_j) \geq 0\:, \quad \forall \{c_j\}_{j=1,\ldots, N}\subset \bC\:,  \forall \{k_j\}_{j=1,\ldots, N}\subset X \:, \forall N=1,2,\ldots\:.\label{POSK}\eeq
Notice that a positive definite kernel is necessarily {\bf Hermitian}: $K(p,q)= \overline{K(q,p)}$.
\item[(2)]  We stress that in (\ref{KC}), $k = (k^0(\vec{k}), \vec{k})$ and $p = (p^0(\vec{p}), \vec{p})$ in accordance with (\ref{ms}), so that we can  see $K_g$ either  as a function on $\bR^3\times \bR^3$ or on $\sV_{m,+}\times \sV_{m,+}$ indifferently. \hfill $\blacksquare$
\end{itemize}}
\end{remark}

\noindent It is possible to prove  (see \cite{C} for details) that $\cP_+$-covariant localization POVMs with causal kernel do exist. In particular, a family of functions $g$ as in the above definition which give rise to corresponding $\cP_+$-covariant POVM with causal kernel is 
\beq
g_r(z) :=   \frac{(2m^2)^r}{(m^2+z)^r}\: \quad z\geq m^2\:, r \geq 3/2\:. \label{gr}
\eeq
Finite convex combinations of functions $g_r$ also define  causal kernels.
   Even some pointwise limits of these convex combinations do the job as discussed in  \cite{P} and \cite{HW}:
$$ g(z) = \left(\frac{(1+c)m^2}{cm^2+z }\right)^n\:,  \quad z \geq m^2 \quad \mbox{for $-1< c \leq 1$ and $\bN \ni n>1$.}$$

 It is convenient to re-write  the POVM $\sT^g_{\bR^3_{x^0}}$ of Definition \ref{DEFCP} into a different form, where (a) the dependence on  a chart disappears, (b) the spacelike flat Cauchy surface normal to a reference frame remains, and (c) some covariance properties explicitly show up. A similar reformulation already appeared in \cite{C}, but we use here a slightly different approach and notation, more useful for our final goals. To this end, as in \cite{C}, we first define a current $j_g(k,p) \in \sV_+$ by means of
\beq
j_g(k,p)  :=\frac{1}{2}(k+p)g(-k\cdot p)\:, \quad p,k \in \sV_{m,+}\:,\label{jkp}
\eeq
so that (\ref{POVMC}) can be rephrased to an equivalent form
\beq
\langle \phi | \sT^g_{\bR^3_{x^0}}(\Delta) \phi\rangle =-\frac{1}{(2\pi)^3} \int_\Delta \int_{\bR^3}\int_{\bR^3} \frac{j_g(p,k) \cdot n}{\sqrt{p^0k^0}} e^{i(p-k)\cdot x} \overline{\phi(\vec{k})}\phi(\vec{p})  d^3p d^3k d\nu_{\bR^3_{x^0}}(x)\:, \quad \Delta\in \cB( \bR^3_{x^0})\:.\label{POVMC2}
\eeq
Above, the Minkowski chart  $x^0,x^1,x^2,x^3$ is adapted to $n:= \partial_{x^0}$. Furthermore, up to a spacetime displacement of the origin of the coordinates, we can always assume that $(0,0,0,0)$ corresponds to the origin $o\in \bM$ initially fixed\footnote{Changes of  the origin $o\in \bM$ can be re-absorbed in the definition of wavefuctions through an obvious unitary transformation.} (beginning  of Section \ref{SECPVMC}). The vector $x\in \sV$  in the exponent in (\ref{POVMC2}) is  such that the point $o+x \in \bM$ has  coordinates $x^0,x^1,x^2,x^3$.
 We might therefore  write $d\mu_{\bR^3_{x^0}}(o+x)$ rather than  $d\mu_{\bR^3_{x^0}}(x)$. However,  this misuse of notation does not produces troubles because the map  $\sV \ni x \mpasto o+v \in \bM$ is one-to-one.
The factor $1/\sqrt{p^0k^0}$ in (\ref{POVMC2})  stems from  the choice of the  the Hilbert space $L^2(\bR^3,d^3p)$.  It can be removed by passing   to covariant one-particle Hilbert space $\cH:=  L^2(\sV_{m,+}, \mu_m)$ (\ref{COVH}). 
With this prescription, if  $\Sigma\in \cC_\bM^{sf}$,  (\ref{POVMC2}) can be rephrased to
$$\langle \psi | \sT^g_\Sigma(\Delta) \psi\rangle =$$
\beq
 - \int_\Delta \int_{\sV_{m,+}}\int_{\sV_{m,+}} \frac{ j_g(p,k)\cdot n_\Sigma}{(2\pi)^3}  e^{i(p-k)\cdot x} \overline{\psi(k)}\psi(p)d\mu_m(p) \mu_m(k)  d\nu_{\Sigma}(x)\:,  \quad \Delta\in \cB(\Sigma)\label{POVMC2'}
\eeq
It is now evident that no choice of a {\em Minkowski chart} enters  (\ref{POVMC2'}) and the only spot where  a {\em Minkowski reference frame} $n_\Sigma$ takes place is when  assigning  the spacelike flat Cauchy surface $\Sigma$, since   $n_\Sigma$ is  the  future directed unit normal vector to $\Sigma$. The vectors $x$ entering the exponential  satisfy  $o+x\in S$, which, in turn, is the integration space of the external integral.

\subsection{Properties of $\sT^g_\Sigma$, for spacelike flat Cauchy surfaces $\Sigma$: Covariance,  causality, no strict localizability, Newton-Wigner, Heisenberg inequality,}
If we fix a function $g$ as in Definition \ref{DEFCP},  the arising spatial localization observable  $\{\sT^g_\Sigma\}_{\Sigma\in \cC^{sf}_\bM}$  is {\em $\cP_+$-covariant}  in the sense  discussed in \cite{C}: If $U$ is the unitary representation of $\cP_+$ in the one-particle space introduced in (\ref{UNI}) and (\ref{UNI2}), it holds
\beq U_h \sT^g_\Sigma(\Delta)U^{-1}_h  = \sT^g_{h\Sigma}(h\Delta)  \:, \quad \forall \Sigma \in \cC^{sf}_\bM\:, \forall \Delta \in \cL(\Sigma)\:, \quad \forall h \in \cP_+\label{ACOVNWT20}\:.\eeq
 Notice that $h\Sigma$ is an analogous spacelike flat  Cauchy surface normal  to $n_{h\Sigma}= \Lambda_hn_\Sigma$ if $h=(y_h,\Lambda_h)$.
We do not enter into the details of these properties because  we shall not use them in the rest of the paper.  The  spatial localization observable
$\{\sT^g_\Sigma\}_{\Sigma\in \cC^{sf}_\bM}$ complies with  the definition of {\em relativistic} spatial localization observable proposed in \cite{M2} (Definition 18 therein).
  Denoting the POVMs extended from $\cB(\Sigma)$  to $\cL(\Sigma)$  with the same symbol $\sT^g_\Sigma$, we can prove the following result.

\begin{theorem}[\cite{C}]  If    $g: [m^2,+\infty) \to \bR$  continuous\footnote{Actually continuity of $g$ already follows from measurability of $g$, see \cite{C} (55) Corollary}, normalized to $g(m^2)=1$ such that the kernel $K_g: \bR^3 \times  \bR^3 \to \bC$ in (\ref{KC}) is positive definite, then spatial localization observable  $\{ \sT^g_\Sigma\}_{\Sigma \in \cC^{sf}_\bM}$
 satisfies CC in Definition \ref{DEFC0}.   
 \end{theorem}

\begin{proof}  Theorem 56 \cite{C}. \end{proof} 

\begin{corollary} 
If $\psi \in \cH$, the localization probability associated to the  spatial localization observable 
$\{ \sT^g_\Sigma\}_{\Sigma \in \cC^{sf}_\bM}$ cannot be  zero outside a bounded set in $\Sigma$. 
\end{corollary}

\begin{proof}
If localized states as above exist   CT  would fail as a consequence of Hegerfeldt’s theorem on relativistic time
evolution (Theorem \ref{HTeo}) and thus also CC would be false. (See \cite{M2} for a discussion on this point.)
\end{proof}

 \begin{remark} {\em In spite of this obstruction, it is possible to show \cite{C} that probability distributions localized in bounded sets can be arbitrarily well approximated by probability distributions arising by suitable sequences of state $\psi_n$.  In the case of
causal kernels it is an important open problem whether there are {\em point localized} sequences of states, i.e., roughly speaking whether the system
can be localized within a bounded region as accurately as desired. This
problem is discussed in \cite{C} sec. 18 Discussion.} \hfill $\blacksquare$
\end{remark}
 
There is a interesting  relation between the {\em first moment} of the  POVM $\sT^g_\Sigma$ on a spacelike flat Cauchy surface $\Sigma$ and the {\em Newton-Wigner}  selfadjoint operators $N_\Sigma^1,N_\Sigma^2,N_\Sigma^3$ \cite{M2} associated to a Minkowski chart $x^0,x^1,x^2,x^3$ such that the slice $x^0=0$ coincides with $\Sigma$.  Similarly to  Thm 26 in \cite{M2}, one sees that  the following result is valid. Below 
$\Delta_\psi  x^a_{\sT^g_\Sigma}$ denotes the standard deviation for the coordinate $x^a$ referred to  the probability distribution $\langle \psi| \sT^g_\Sigma(\cdot) \psi\rangle$ constructed out of the POVM $\sT^g_\Sigma$ in the state defined by the unit vector $\psi$.

\begin{theorem}\label{1MNW1}  Suppose that $g$ in the definition of $\sT^g_\Sigma$ is real, bounded, and smooth, then the following is true.
\begin{itemize}
\item[(a)]   The $a$-first moment of $\sT^g_\Sigma$ is defined for every  $\psi\in \cD(\cH)$ with $||\psi||=1$ and 
\beq \int_{\Sigma} x^a \langle \psi| \sT^g_\Sigma(d^3x) \psi\rangle =  \langle \psi| N_\Sigma^a \psi\rangle \:,\quad \mbox{where $a=1,2,3$.}
\label{NTN}\eeq
In particular, $N^a_\Sigma$  is  the unique selfadjoint operator in $\cH$  which satisfies the indenty above.
\item[(b)] The {\bf Heisenberg inequality} turns out to be corrected as, for $\psi \in \cD(\cH)$,
$$\Delta_\psi  x_{\sT^g_\Sigma}^a \Delta_\psi P_a  \geq \frac{\hbar}{2} \sqrt{1+4 (\Delta_\psi P_a)^2 \langle \psi|\sK^{\sT^g_\Sigma}_{a} \psi\rangle}\:, \quad a=1,2,3\:.$$
$\sK^{\sT^g_\Sigma}_{a}\in \gB(\cH)$ is a selfadjoint  operator,  which is a (spectral)  function of the four  momentum observable $P$ with the form (\ref{kappa}), such that $\sK^{\sT^g_\Sigma}_{a}\geq 0$.
\item[(c)] If $g$ is of the form (\ref{gr}), or convex combinations of them, then (a) and (b)  also hold  for $\psi \in {\cal S}(\cH)$.
\end{itemize}
\end{theorem}

\noindent{\em Proof}. See Appendix  \ref{APPENDIXA}. \hfill $\Box$\\

 Notice that the left-hand side of (\ref{NTN}) does not depend on $g$.

\begin{remark}\label{NAICOM} {\em
\begin{itemize}
\item[(1)] If also the identity $$ \int_{\Sigma} (x^a)^2 \langle \psi| \sT^g_{\Sigma}(d^3x) \psi \rangle =  \langle \psi| (N^a_\Sigma)^2\psi\rangle\: \quad \mbox{(false!)},$$ 
 were valid one could apply a known theorem by Naimark about the decomposition of maximally symmetric operators (here $N^a_\Sigma$) in terms of POVMs (see Theorem 23  in \cite{DM} and the discussion about it) obtaining  $\sT^g_\Sigma= \sQ_{\Sigma}$. This is obviously false and it is also reflected by  the appearance of the  term $ \langle \psi|\sK^{\sT^g_\Sigma}_{a} \psi\rangle$ in the modified  Heisenberg inequality.
\item[(2)]   If $U^{(n)}_t$ is the unitary time evolutor corresponding to the time evolution along $n$ in $\bM$,
  it is  easy to see that the Heisenberg evolution  $U^{(n)}_t N_\Sigma^a U^{(n)\dagger}_t$  of $N_\Sigma^a$ on the right-hand side of (\ref{NTN}) equals  the integral on the left-hand side over  the correspondingly temporally translated time slice $\Sigma_t$. As already observed in \cite{M2}, this  fact implies that the worldline $\bR \ni t \mpasto (t,  \int_{\Sigma_t} \vec{x} \langle \psi| \sT^g_{\Sigma_t}(d^3x) \psi\rangle)$ is {\em timelike} (Corollary 14 in \cite{M2}) as expected by massive particles. \hfill $\blacksquare$
 \end{itemize}} 
 \end{remark}

\subsection{A spacelike Cauchy localization  observable  $\sT^g= \{\sT^g_S\}_{S\in \cC_\bM^s}$ for a  massive KG particle}
Let us focus again on the POVM $T_{g, \bR^3_{x^0}}$ satisfying  (\ref{POVMC}).
The equivalent form (\ref{POVMC2'}) of  (\ref{POVMC})  is actually already prompt to be generalized to any smooth Cauchy surface $S$.
Heuristically, if   $S\subset \bM$ is a spacelike Cauchy surface defined by  $x^0=t_S(\vec{x})$, we expect that the current $j_{g}$ also  defines a  normalized POVM  whose  expectation value on $\Delta \in \cB(S)$ is:
\beq
\langle \psi | \sT^g_S(\Delta) \psi\rangle =- \int_\Delta \int_{\sV_{m,+}}\int_{\sV_{m,+}}\sp \sp \sp \frac{j_g(p,k)  \cdot n_S(x)}{(2\pi)^3} e^{i(p-k)\cdot x} \overline{\psi(k)}\psi(p) d\mu_m(p)  d\mu_m(k) d\nu_S(x)\:,  \quad \Delta\in \cB(S) \label{POVMS'}
\eeq
when $\psi$ belongs to a suitable subspace of $\cH$ and where $n_S$ is the normal future-directed unit vector to $S$.
In case $S$ is smooth but not spacelike, we can  expect the version in terms of forms be valid
\beq
\langle \psi | \sT^g_S(\Delta) \psi\rangle =\int_\Delta  \omega^{J_{g,\psi}}\:, \quad \Delta\in \cB(S) \label{POVMS''}\:,
\eeq
where
\beq J_{g,\psi}(x) := \frac{1}{2(2\pi)^3}  \int_{\sV_{m,+}}\int_{\sV_{m,+}} (k+p)g(-k\cdot p)  e^{i(p-k)\cdot x} \overline{\psi(k)}\psi(p)  d\mu_m(p)d\mu_m(k)\:.\label{POVMS''a} \eeq
As we shall see shortly, this current is conserved: $\partial_\mu J_{g, \psi}^\mu =0$, so that $d  \omega^{J_{g,\phi}}=0$ and that it is zero or causal and future directed at each point of $\bM$.

To extend the definition of $\sT^g_S$ to every smooth Cauchy surface of $\bM$, we cannot directly  follow the  approach of \cite{C} -- based on a smart  decompositions of $\sT(\Delta)$ in terms of isometries and the standard PVM of the position observable on $\bR^3$ --  because of the appearance of  various spurious  terms containing $t_S(\vec{x})$ and $\nabla t_S(\vec{x})$ in the expression of $\omega^{J_{g,\psi}}$.  Nevertheless, the wanted extension is feasible through  a more indirect way.

There are various lemmata we shall exploit to achieve the wanted result.

\begin{lemma}\label{LEMMA0}
Consider   $g: [m^2,+\infty) \to \bR$  continuous, normalized to $g(m^2)=1$ such that the kernel $K_g: \bR^3 \times  \bR^3 \to \bC$ in (\ref{KC}) is positive definite. \\
If $\psi \in  \cD(\cH)$  then  the  current $J_{g,\psi}$ defined in (\ref{POVMS''a}) is smooth, bounded, conserved, and zero or causal and future-directed at any  point of $\bM$.
\end{lemma}

\begin{proof}   Fix a Minkowski chart $x^0,x^1,x^2,x^3$.  Let us take $\phi \in  C_c^\infty(\bR^3) \subset L^2(\bR^3, d^3k)$  and let us consider  the vector field in $\bM$ defined in (\ref{POVMS''a}), now in equivalent terms of $\phi$ instead of $\psi$ according to (\ref{PHI}),
$$J_{g,\phi}^\mu(x) :=  \int_{\bR^3}\spa \int_{\bR^3} \spa\frac{(p^\mu+k^\mu) }{2(2\pi)^3\sqrt{k^0p^0}} g(-k\cdot p) e^{i(p-k)\cdot x} \overline{\phi(\vec{k})}\phi(\vec{p}) d^3p d^3k \:.  $$
With the said hypotheses on $\phi$, a direct use of the Lebesgue  dominate convergence theorem permit us  to pass the $x^\nu$ derivatives of every order under the integration symbol proving that $J$ is bounded and smooth. However, the first derivative  yields
$$\partial_\mu J_{g,\phi}^\mu(x) = i\int_{\bR^3}\spa \int_{\bR^3} \spa\frac{(p^\mu+k^\mu)(p_\mu-k_\mu) }{2(2\pi)^3\sqrt{k^0p^0}} g(-k\cdot p) e^{i(p-k)\cdot x} \overline{\phi(\vec{k})}\phi(\vec{p}) d^3p d^3k =0$$
since $(p^\mu+k^\mu)(p_\mu-k_\mu)  = -m^2+m^2 +k^\mu p_\mu - p^\mu k_\mu =0$. 
Finally, as established in \cite{C}, $J(x)$ is either zero or causal and future-directed for every $x\in \bM$  (see Theorem 52 in  \cite{C} and (c) in its  proof). \end{proof}

\begin{lemma} \label{LEMMA3}
Consider  $g$ as in Lemma \ref{LEMMA0}.
If $S$ is a smooth Cauchy surface, it holds
\beq \int_{S}  \omega^{J_{g,\psi}} =\langle \psi|\psi \rangle\:. \label{NT0}\eeq
for every $\psi\in \cD(\cH)$.
\end{lemma}

\begin{proof}  In view of Lemma \ref{LEMMA0} , since $J_{g,\psi}$ satisfies all required hypotheses,
we can apply Proposition \ref{PROPJ} for $S'$ given by the time-slice at $x^0=0$,  taking (\ref{VVN}) into account and finally  obtaining that 
$$\int_{S}  \omega^{J_{g,\phi}} =  \int_{S'}  \omega^{J_{g\phi}} =\int_{\bR^3} \spa\int_{\bR^3}\spa \int_{\bR^3} \spa\frac{(p^0+k^0)}{2(2\pi)^3\sqrt{k^0p^0}} g(-k\cdot p) e^{i(\vec{p}-\vec{k})\cdot \vec{x}} \overline{\phi(\vec{k})}\phi(\vec{p})  d^3p d^3k d^3x = \langle \phi|\phi \rangle_{L^2(\bR^3)}\:,$$
which is the thesis.
The last identity is due to the fact that, as established in \cite{C}, the last integral is nothing but $\langle \phi| \sT_{\bR^3_0}(\bR^3) \phi \rangle$ where $\cB(\bR^3)\ni \Delta \mapsto \sT_{\bR^3_0}(\Delta)$ is a normalized POVM in $\bR^3_{x^0=0}$, because is a POVM defined as in Definition \ref{DEFCP}.\end{proof}

\begin{lemma} \label{LEMMAAGG} Let $\cK$ be a complex Hilbert space and  $D\subset \cK$ a dense subspace.  Consider a 
Hermitian  form
 $\Lambda: D \times D \to \bC$ such that $|\Lambda (x,x)| \leq C||x||^2$ for some constant $C\geq 0$ and $\Lambda(x,x) \geq 0$ for every $x\in D$.
Then a unique operator $A \in \gB(\cK)$ exists such that 
$\Lambda(x, x) = \langle x | A x \rangle$ for all $x \in D$.
Furthermore  $||A|| \leq C$ and  $A\geq 0$.
\end{lemma}

\begin{proof}  The Cauchy-Schwartz inequality implies $|\Lambda (x,y)| \leq C ||x||\: ||y||$ if $x,y \in D$.
 Thus $\Lambda$ is continuous  and, 
since $D\subset \cH$ is dense,  $\Lambda$ continuously extends to a Hermitian form on $\cH \times \cH$
 satisfying the same bound as above, as the reader easily proves.
The proof ends by direct application of Corollary after Theorem II.4 in \cite{RS}. \end{proof}
%
%

We are in a position to prove that a  normalized POVM induced by the operators $\sT^g_S(\Delta)$ exists on every smooth  Cauchy surface of $\bM$ and satisfies the natural coherence requirement (\ref{COHERENCE}) also in that case. 

\begin{theorem} \label{MAIN}  Let $S\in \cC_\bM$ and suppose that $g: [m^2,+\infty) \to \bR$ is  continuous,  $g(m^2)=1$ and  the kernel  in (\ref{KC}) is positive definite.  Then,
 \begin{itemize}
\item[(a)] there is a unique   $\cH$-POVM, still indicated by $\{\sT^g_S(\Delta)\}_{\Delta \in \cB(S)}$, which  satisfies (\ref{POVMS''})  when $\psi \in {\cal D}(\cH)$:
$$
\langle \psi | \sT^g_S(\Delta) \psi\rangle =\int_\Delta  \omega^{J_{g,\psi}}\:, \quad \Delta\in \cB(S) \:,
$$
so that, $\sT^g_S < \sp < \nu_S$ if $S \in \cC_\bM^s$ due to (\ref{FJ}).

\item[(b)] $\sT^g_S$ is normalized: $\sT^g_S(S)=I$; 
\item[(c)] $\sT^g_S$ satisfies Definition \ref{DEFCP} (i.e., it is a POL with causal kernel according to  \cite{C}) when $S=\Sigma$ is a time slice of a Minkowski coordinate  system;

\item[(d)] if $\Delta \in S\cap S'$ is a Borel set, where $S'\in \cC_\bM$, then
$$\sT^g_S(\Delta) =\sT^g_{S'}(\Delta)\:.$$
\end{itemize}
The family $\sT^g := \{\sT^g_S\}_{S\in \cC^s_\bM}$, where we extend each  POVMs to $\cM(S)$ according to Proposition \ref{ACM}, is a spacelike Cauchy localization  observable  in the sense of Definition \ref{GSL}.
\end{theorem}

\begin{proof}
Item (c) is a trivial consequence of (a),(b), and \cite{C}, so we prove (a), (b), and (d).  The last statement is evident in view of the items (a), (b), (d) and Def. \ref{GSL}\\
We prove (a) and (b) together. Let   $x^0=t_S(\vec{x})$ be the map which defines $S$ by identifying it with the spatial $\bR^3$ in a given   Minkowski system of coordinates $x^0, \vec{x}$ according to Proposition \ref{propS}. To make easier the notation we shall omit $^g$ in $\sT^g_S$.   
If $\Delta \in \cB(\bR^3)$ is bounded, so that the integral below is defined,  consider the Hermitian form 
$$\Lambda_\Delta (\psi,\psi')$$ $$ := \sp\int_\Delta \spa \int_{\bR^3}  \spa \int_{\bR^3} \spa\frac{(p^0+k^0) - (\vec{p}+ \vec{k})\cdot \nabla t_S(\vec{x})}{2(2\pi)^3\sqrt{k^0p^0}} g(-k\cdot p) e^{i(\vec{p}-\vec{k})\cdot \vec{x}-i(p^0-k^0)t_S(\vec{x})}\overline{\phi_\psi}(\vec{k}) \phi_{\psi'}(\vec{p}) d^3p d^3k  d^3x$$
for $\psi,\psi' \in \cD(\cH)$ and $\phi_\psi, \phi_{\psi'}$ are the corresponding elements in the isomorphic  Hilbert space $L^2(\bR^3,d^3p)$. Since (use (\ref{VVN}))
$$\Lambda_\Delta(\psi,\psi) = \int_\Delta \omega^{J_{g,\psi}}\:,$$
we conclude that 
\beq 0\leq \Lambda_\Delta(\psi,\psi) \leq \langle \psi|\psi \rangle\:. \label{INEQ}\eeq
The former inequality arises from the last identity in (\ref{VVN}) when remembering that $J_{g,\psi}$ is zero or causal and future directed for Lemma \ref{LEMMA0}. The latter inequality in (\ref{INEQ}) is consequence of the positivity of the integrand $ \omega^{J_{g,\psi}}$ in coordinates (again for (\ref{VVN})) and of Lemma \ref{LEMMA3}.
We can finally apply Lemma \ref{LEMMAAGG} proving that there exists $\sT_S(\Delta) \in \gB(\cH)$ such that $0\leq T_S(\Delta)\leq I$
and \beq \langle \psi|T_S(\Delta) \psi\rangle = \int_\Delta \omega^{J_{g,\psi}}\:. \label{INTFI}\eeq
Eq.(\ref{POVMS''}) is therefore satisfied for $\Delta$ bounded.
To go on  we define $\cF(S)$ as the ring of Borel sets of $S$ which are bounded in $\bR^3$.  It is clear that the $\sigma$-algebra generated by $\cF(S)$ is $\cB(S)$ itself.
For the elements  $\Delta \in \cF(S)$  we can use the positive operator\footnote{This positivity property is by no means obvious and it may fail when trying to define localization POVMs in terms of objects which are {\em classically} positive like energy \cite{PRD}. See Remark \ref{remarkPRD} below.}   $\sT_S(\Delta) : \cH \to \cH$
bounded by $I$ defined above.
We want to prove that, if $\psi\in \cH$, the map $$\nu_\psi : \cF(S) \ni \Delta \mapsto \langle \psi |\sT_S(\Delta) \psi \rangle \in [0, ||\psi||^2]\; \quad \psi \in \cH$$ is a $\sigma$-additive premeasure on $\cF(S)$. We stress that  $\sigma$-additivity is guaranteed by the monotone convergence theorem referred to the integral in (\ref{INTFI}) -- using the fact  that the integrand  is positive. However this argument works {\em only when $\psi \in \cD(\cH)$} since  (\ref{INTFI}) is given for this type of functions. On the other hand, the very structure of  (\ref{INTFI}) implies (simple)  additivity of $\nu_\psi$  for $\psi \in \cH$ just by the continuity of $\sT_S(\Delta)$.  To prove $\sigma$-additivity for the general $L^2$ case, consider  $\Delta \in \cF(S)$ which is the countable union of pairwise-disjoint  sets $\Delta_n \in \cF(S)$.  If $\psi \in \cH$, since all operators $\sT_S(\Delta_n)$ are positive and additivity holds, we conclude that
$$\sT_N := \sum_{n=0}^N \sT_S(\Delta_n) \geq 0 \quad \mbox{as well as} \quad \sT_N \leq \sT_{N+1} \leq I\:.$$
The latter inequality arises from $\sT_N = \sT_S(\cup_{n=0}^N \Delta_n)\leq I$. A known result on increasing sequences of positive operators (Proposition 3.76 \cite{M}) proves that there exist a bounded  operator $P: \cH \to \cH$ such that $0\leq P \leq I$ and
\beq \langle \psi |\sT_N\psi \rangle \to \langle \psi| P \psi \rangle \quad \mbox{as $N\to +\infty$ for all $\psi \in \cH$}\label{LIM}\eeq
However, $\sigma$-additivity for $\psi \in \cD(\cH)$ guarantees that
$$\sum_{n=0}^{N} \nu_\psi(\Delta_n)= \langle \psi |\sT_N\psi \rangle \to \langle \psi| \sT_S(\Delta) \psi \rangle \quad \mbox{as $N\to +\infty$ for all $\psi \in \cD(\cH)$}\:.$$
Comparing the found limits, we conclude that  
$ \langle \psi| \sT_S(\Delta) \psi \rangle  =  \langle \psi| P \psi \rangle $ for every $\psi\in \cD(\cH)$. By continuity of $P$ and $\sT_S(\Delta)$, this identity extends to $\psi \in \cD(\cH)$ which, in turn, yields $\sT_S(\Delta)= P$ since  the Hilbert space is complex. In summary, (\ref{LIM})
can be rephrased to
$$\sum_{n\in \bN} \langle \phi | \sT_S(\Delta_n) \phi \rangle = \langle \phi | \sT_S(\cup_{n\in \bN}\Delta_n) \phi \rangle\quad \mbox{for all $\phi \in \cH$,}$$
that is the wanted $\sigma$-additivity property for $\nu_\phi$.
 As a consequence of the Carath\'eodory extension theorem, there is a positive $\sigma$-additive  measure
$\overline{\nu_\phi} : \cB(S) \to  [0,+\infty]$ which extends $\nu_\phi$, for every given $\phi \in L^2(\bR^3, d^3k)$. This measure is unique because $S$ is countable union of sets in $\cF(S)$ thus with  premeasure $\nu_\phi$ 
 finite.  $\overline{\nu_\phi}$ is  finite by inner continuity: If $\Delta_n \in \cF(S)$ satisfy $\Delta_n \subset \Delta_{n+1}$ and $\cup_{n\in \bN} \Delta_n = S$, we have 
 $$\overline{\nu_\phi}(S) = \sup_{n \in \bN} \overline{\nu_\phi}(\Delta_n)
 =  \sup_{n \in \bN} \nu_\phi(\Delta_n) \leq ||\phi||^2$$
Now we use $\nu_\phi$ to extend the definition of $\sT_S(\Delta)$ to the case of $\Delta \in \cB(S)$ without the constraint $\Delta \in \cF(S)$.\\
Consider $\Delta \in \cB(S)$ and a sequence $\Delta_n \in\cF(S)$ such that $\Delta_{n} \subset \Delta_{n+1} \subset \Delta$ and $\cup_{n\in \bN} \Delta_n = \Delta$. For instance, $\Delta_n = \Delta \cap B_n(0)$, where $B_n(0)$ is the open ball of radius $n$ centered at the origin of $\bR^3\equiv S$.  By construction, using additivity, $0\leq \sT_S(\Delta_n) \leq \sT_S(\Delta_{n+1})\leq I$. Therefore, again for Proposition 3.76 \cite{M}, there exists a bounded everywhere defined operator \beq \sT_S(\Delta) := s\mbox{-}\lim_{n\to +\infty} \sT_S(\Delta_n) \label{LIMSS}\eeq such that $0\leq \sT_S(\Delta) \leq I$. With an argument strictly similar to the one used before, this operator satisfies
$$\overline{\nu_\phi}(\Delta) = \langle \phi | \sT_S(\Delta)\phi\rangle \quad \mbox{for every $\phi\in\cH$.}$$
This identity proves both that $\sT_S(\Delta)$ extends the definition for $\Delta \in \cF(S)$ and that $\sT_S(\Delta)$ does not depend on the used sequence of sets $\Delta_n\in \cF(S)$ to define it. 
Now fix  $\Delta \in \cB(S)$ with $\Delta$ unbounded. Taking advantage of a sequence $\Delta_n := \Delta \cap B_n(0)$, according to (\ref{LIMSS}) and (\ref{INTFI}), 
we can prove  that (\ref{POVMS''}) is valid also if $\Delta$ unbounded for $\psi \in \cD(\cH)$:
$$\langle \psi| \sT_S(\Delta) \psi \rangle =  \lim_{n\to +\infty} \langle \psi|\sT_S(B_n(0)\cap \Delta) \psi\rangle 
=\lim_{n\to +\infty}\int_{B_n(0)\cap \Delta}  \omega^{J_{g,\psi}} = \int_{\Delta}  \omega^{J_{g,\psi}}\:.$$
In the last limit,  we exploited the  monotone convergence theorem (the integrand being positive). 
Finally, let us consider the case $\Delta = S\equiv \bR^3$. Again, if $\psi \in \cD(\cH)$
$$\langle \psi| \sT_S(S) \psi \rangle =  \lim_{n\to +\infty} \langle \psi|\sT_S(B_n(0)) \psi\rangle 
=\lim_{n\to +\infty}\int_{B_n(0)}  \omega^{J_{g,\psi}} = \int_{S}  \omega^{J_{g,\psi}} =\langle\psi|\psi\rangle$$
The last limit, due to the monotone convergence theorem (the integrand being positive) coincides to the integral on the whole $S$ which is $\langle\psi|\psi\rangle$ according to  Lemma \ref{LEMMA3}. Hence $\langle\psi|\sT_S(S)-I)\psi \rangle=0$ for every $\psi \in \cD(\cH)$. The standard density and continuity argument (in a complex Hilbert space) permits to conclude that $\sT_S(S) =I$.
This result also proves  (b) if $\sT_S$ satisfies the other requirements of a POVM. To end the proof, it is sufficient to prove that $\{\sT_S(\Delta)\}_{\Delta \in \cB(S)}$ is a POVM in all cases. It is equivalent to prove that  $\{\langle \psi|\sT_S(\Delta)\phi\rangle \}_{\Delta \in \cB(S)}$ is a complex measure for every choice of $\phi,\psi \in \cD(\cH)$. The only fact to be proved is that $\cB(\bR^3) \ni \Delta \mapsto \langle \psi|\sT_S(\Delta)\phi\rangle \in \bC$ is $\sigma$-additive. This fact immediately arise by  $\sigma$-additivity  of the positive Borel measure $\cB(S) \ni \Delta \mapsto \langle \chi |\sT_S(\Delta)\chi\rangle \in [0, \||\chi||^2]$, by choosing $\chi = \psi\pm \phi$ and $\chi = \psi \pm i\phi$ and taking advantage of the polarization identity.\\
The proof of (d) is trivial due to (a), polarization,  and an obvious density argument.\\
 The last statement of the thesis  immediately arises from Proposition \ref{PROPCOEES}.
\end{proof} 

\section{Spacelike Cauchy localization observables out of the stress-energy tensor of massive KG particles}\label{SECMMM}
In the final Sect 7 of \cite{M2}, a second type of spatial localization  observable was introduced by  generalizing an idea by D. Terno \cite{T} also analysed   in the first part of \cite{M2}. These POVMs denoted by $\sM^{n'}_{n,t}$ (with $n, n'\in \sT_+$) were constructed on spacelike flat Cauchy surfaces out of the stress energy tensor of the Klein Gordon field. 

The goal of the remaining part of this section is to prove that this different notion of POVM can be defined on spacelike Cauchy surfaces of $\bM$, giving rise to a Cauchy  localization  according to Def. \ref{GSL}.
  We expect that  the construction can be  generalized to any static (or stationary) globally hyperbolic  spacetime referring to the Hadamard static vacuum, since the theoretical construction does not depend on the use of Fourier transform, at least at heuristic level.  This conjecture will be analysed elsewhere.

\subsection{The POVM $\sM^n_{\Sigma}$ for massive KG particles}
In Section 7 of \cite{M2}, extending a notion introduced by D. Terno \cite{T},  a family of POVMs was introduced on all spacelike flat Cauchy surfaces $\Sigma$ of $\bM$ for a common choice of a reference frame $n\in \sT_+$. (If $n=n_\Sigma$ one obtains the very notion introduced by Terno, studied and made rigorous in the first part of \cite{M2} that is a special case of the following discussion. A rigorous proof of CT for that observable appears in \cite{M2}).

At the level of 2nd quantization of the massive real Klein-Gordon field, the considered  POVM is formally defined as follows on  $\Sigma \in \cC^{sf}_\bM$ (we readapt the notation to the  choices of the present work) 
\beq\label{Ternoinformal}
\sM^n_{\Sigma}(\Delta):= \frac{1}{\sqrt{H_n}}P_1 \int_\Delta :\spa\hat{T}_{\mu\nu}\spa:\spa(x) n^\mu n_\Sigma^\nu \:  d \nu_\Sigma(x) P_1\frac{1}{\sqrt{H_n}} \:, \quad \Delta \in \cL(\Sigma)\:,
\eeq
where $P_1 : \gF_+({\cal H}) \to {\cal H}$ is the orthogonal projector onto the one-particle space of the symmetric Fock space $\gF_+({\cal H}_m)$ constructed upon the Minkowski vacuum state with the Hilbert space ${\cal H}$ defined as in (\ref{COVH}) as the one-particle subspace. $:\spa \hat{T}_{\mu\nu}\spa:\spa(x)$ is the {\em normally ordered} stress energy tensor operator.  $H_n$ is the Hamiltonian operator of the quantum field in the reference frame $n\in \sT_+$.
$\Sigma$ is a {\em flat} spacelike Cauchy surface orthogonal with constant normal unit vector  $n_\Sigma\in \sT_+$.

The overall idea at the basis of \cite{T} and \cite{M2} is that a physical procedure to detect a particle in a region of a flat Cauchy surface may exploit  the energy of the particle.  However,  we have many ways to  synchronize a net of detectors and,  as discussed in \cite{M2}:  $\langle \psi|\sM^n_{\Sigma}(\Delta)\psi\rangle$ accounts for  the probability to find a particle of state $\psi$ in $\Delta \subset \Sigma$ {\em using a net of detectors which are (a) at rest in $n$ but (b) synchronized on  $\Sigma$}.  This possibility naturally arises from the observation \cite{M2} that 
\beq P_1 \int_\Delta :\spa\hat{T}_{\mu\nu}\spa:\spa(x) n^\mu n_\Sigma^\nu \:  d \nu_\Sigma(x) P_1 \geq 0,\label{TNAIVE}\eeq
 for {\em every} choice of $n,n_\Sigma\in \sT_+$ and $\Delta \in \cL(\Sigma)$,
even if positivity fails  when removing the one-particle space projectors $P_1$.

\begin{remark}\label{remarkPRD} {\em A  similar local positivity property {\em does not hold} when dealing with massive Dirac  particles, even if $n=n_\Sigma$ \cite{PRD}: A localization POVM cannot be constructed for this type of fermions in terms of energy on a given rest space of a Minkowski reference frame according to  the approach  of \cite{T,M2}. The natural physical object, useful to this goal,  is instead the fermionic  current operator $:\spa\hat{J}^\mu \spa: \spa(x)$ \cite{C0}.}  \hfill $\blacksquare$
\end{remark}

\noindent The rigorous  definition of the normalized  POVM $\sM^n_{\Sigma}$ in $\cH$
corresponding to the formal object (\ref{Ternoinformal})
 was given in Thm 37 of \cite{M2} in terms of a kinematic deformation of the  PVM 
$\sQ_\Sigma(\Delta)$
of the {\em Newton-Wigner position operator} \cite{M2} (see (\ref{PNW}) below) on the spacelike flat Cauchy surface $\Sigma$:
$$
\sM^n_{\Sigma}(\Delta) :=\frac{1}{2}\left( \sqrt{\frac{H_{n_\Sigma}}{H_{n}}}\sQ_{n_\Sigma}(\Delta)
 \sqrt{\frac{H_{n}}{H_{n_\Sigma}}} +  \sqrt{\frac{H_{n}}{H_{n_\Sigma}}} \sQ_{n_\Sigma}(\Delta) \sqrt{\frac{H_{n_\Sigma}}{H_{n}}}\right)
$$
\beq 
-\frac{n\cdot n_\Sigma}{2}  \sqrt{\frac{H_{n_\Sigma}}{H_{n}}}\left(\eta^{\mu\nu}\frac{P_{\mu}}{H_{n_\Sigma}}  \sQ_{n_\Sigma}(\Delta) \frac{P_{\nu}}{H_{n_\Sigma}}  + \frac{m}{H_{n_\Sigma}} \sQ_{n_\Sigma}(\Delta)  \frac{m}{H_{n_\Sigma}} \right)\sqrt{\frac{H_{n_\Sigma}}{H_{n}}}\:.\label{MPOVM}
\eeq
Above 
$\Delta \in \cL(\Sigma)$, $H_r := -P \cdot r$ (for $r\in \sT_+$) 
is the Hamiltonian  in the Minkowski reference frame $r$. 
The various everywhere-defined {\em bounded}  composite operators $H_n/H_{n_\Sigma}$,  $P_{\nu}/H_{n_\Sigma}$ etc. are defined in terms of the joint spectral measure of the four momentum operator $P^\mu$ and standard spectral calculus. The components $P^\mu$ are referred to a Minkowski chart adapted to $n$.

\subsection{Properties of  $\sM_\Sigma^n$ for spacelike flat Cauchy surfaces $\Sigma$: Covariance,  causality, no strict localizability, Newton-Wigner, Heisenberg inequality,}
Though its clear  from (\ref{MPOVM})  that $\sM^n_{\Sigma}(\Delta)\in \gB(\cH)$ and that $\sM^n_{\Sigma}(\Sigma)=I$ (notice that $ \sQ_{n_\Sigma}(\Sigma)=I$
and $P_\mu P^\mu + m^2I=0$),  it  is not evident that $\sM^n_{\Sigma}(\Delta)\geq 0$, nor the connection between (\ref{MPOVM}) and (\ref{Ternoinformal}) seems  straightforward. We spend  this section  about these issues  because the discussion will turn out useful  when we shall generalize $\sM^n_\Sigma$ to generally curved spacelike Cauchy surfaces $S$.

From Eq.(17) in \cite{M2}, we know  that, if $\Sigma\in \cC^{sf}_\bM$,  \beq \langle \psi|\sQ_{\Sigma}(\Delta) \psi\rangle =  \int_{\sV_{m,+}}\sp\sp \overline{\psi(p)}\int_{\Delta}   \int_{\sV_{m,+}}  \sp\sp \frac{e^{i(q-p)\cdot x}}{(2\pi)^3} \sqrt{E_{n_\Sigma}(p) E_{n_\Sigma}(q)}\:\psi(q)   d\mu_m(q) 
d\nu_{\Sigma}(x)  d\mu_m(p)\:,\quad \psi \in {\cal S}(\cH)\:.\label{PNW}\eeq
We adopted the notation
 \beq E_r(p) := -r\cdot p\:, \quad \mbox{for every $r\in \sT_+$ and $p\in \sV$.} \label{EMP}\eeq
($E_r(p)$ is nothing but the component $p^0$ of $p\in \sV$ in any Minkowski chart adapted to $r$.) 
 At this juncture, 
it is not difficult to see from (\ref{MPOVM})  that,  if $\psi \in {\cal D}(\cH)$ or more generally $\psi \in {\cal S}(\cH) (\supset  {\cal D}(\cH))$ as used in \cite{M2}, then
$$\langle \psi|\sM^n_{\Sigma}(\Delta)\psi\rangle =$$
\beq  \int_{\sV_{m,+}} \sp\sp\overline{\psi(p)}\int_{\Delta}   \int_{\sV_{m,+}}  \sp\sp \frac{e^{i(q-p)\cdot x}}{(2\pi)^3} \frac{E_n(p)E_{n_\Sigma}(q)+ E_n(q) E_{n_\Sigma}(p) - n\cdot n_\Sigma(p\cdot  q + m^2)}{2 \sqrt{ E_n(q) E_n(p)}} \psi(q)   d\mu_m(q) d\nu_{\Sigma}(x)  d\mu_m(p)\:,\label{MM}\eeq
Conversely, since  ${\cal D}(\cH)$ and ${\cal S}{\cal \Sigma}(\cH)$ are dense in $\cH$ and $\sM^n_{\Sigma}(\Delta) \in \gB(\cH)$, identity (\ref{MM})  completely determines $\sM^n_{\Sigma}(\Delta)$ by polarization and continuity.

First of all,  we prove that   (\ref{MM}) can be written into an equivalent form already used in \cite{M2}, which eventually  leads to both  the requested positivity condition
and the relation with  (\ref{Ternoinformal}).\\

\begin{lemma}\label{LEMMASWAP} If $\psi\in {\cal S}(\cH)$, the right-hand side of  (\ref{MM}) can be equivalently written with the first two integrals interchanged:
\beq  \int_{\Delta} \sp \int_{\sV_{m,+}}   \int_{\sV_{m,+}}  \sp\sp \frac{e^{i(q-p)\cdot x}}{(2\pi)^3} \frac{E_n(p)E_{n_\Sigma}(q)+ E_n(q) E_{n_\Sigma}(p) - n\cdot n_\Sigma(p\cdot  q + m^2)}{2 \sqrt{ E_n(q) E_n(p)}} \overline{\psi(p)}\psi(q)   d\mu_m(q)  d\mu_m(p)  d\nu_{\Sigma}(x)\label{MM2} \eeq
\end{lemma}

\begin{proof}  See Appendix \ref{APPENDIXA}. \end{proof}

\noindent To go no, if $\psi \in \cD(\cH)$ (or more generally, $\psi \in {\cal S}(\cH)$)  and $n\in \sT_+$, we define the solution of the Klein-Gordon equation 
\beq
\Phi^\psi_n(x) := 
\int_{\sV_{m,+}}  \frac{\psi(p)e^{i p\cdot x} }{(2\pi)^{3/2}\sqrt{E_n(p)}} d\mu_m(p)\:. \label{wavePHI}
\eeq
$\Phi^\psi_n=\Phi^\psi_n(x)$ is smooth on $\bM$ and bounded with all of its derivatives. More strongly $\Phi_n^\psi \in \cS(\Sigma)$ for every flat $\Sigma\in \cC^{sf}_\bM$, where $\cS(\Sigma)$ is the usual Schwartz space on $\bR^3\equiv \Sigma$ referring to any Minkowski  chart adapted to $n_\Sigma$.
We stress that (\ref{wavePHI}) is {\em not} the standard covariant  Klein-Gordon wavefuntion $\varphi_\psi$ (\ref{wavePSI}) associated to a state $\psi \in \cH$, since the integrand above includes a further ''anomalous'' factor $E_p^{-1/2}(p)$. This latter can be traced back to the factors $H_n^{1/2}$ in (\ref{TNAIVE}).
\begin{definition} {\em  If $\psi \in \cD(\cH)$ and $n\in \sT_+$, the  $n$-{\bf normalized stress energy operator} is, in components of a Minkowski chart,
$$T^{\psi}_{\mu\nu}(x)_n := \frac{1}{2}\left(\partial_\mu \overline{\Phi^\psi_n(x)}\partial_\nu\Phi^\psi_n(x) 
 +\partial_\mu \Phi^\psi_n(x)\partial_\nu\overline{\Phi^\psi_n(x)}\right)$$ \beq - \frac{1}{2}\eta_{\mu\nu} \left( \partial^\alpha \overline{\Phi^\psi_n(x)} \partial_\alpha \Phi^\psi_n(x) + m^2 \overline{\Phi^\psi_n(x)} \Phi^\psi_n(x) \right)\:,\label{TTT}
\eeq
and 
the associated current
\beq
J^{\psi\mu}_{n}(x) := n^\nu T^{\psi \mu}_{\nu}(x)_n \:.\label{53}
\eeq
\hfill $\blacksquare$}
\end{definition}

\begin{lemma} \label{LEMMALAST} With the said definitions, for $\psi \in \cD(\cH)$ and $n\in \sT_+$,
\begin{itemize}
\item[(a)] $J^{\psi}_{n}$ is smooth, belongs  to  $\cS(\Sigma)$ for every $\Sigma \in \cC^{sf}_\bM$, and  is bounded;
\item[(b)] $J^{\psi}_{n}$ is conserved;
\item[(c)]  $J^{\psi}_{n}(x)$ is either zero or causal and past-directed if $x\in \bM$, so that 
\beq
J_n^{\psi}(x) \cdot r \geq 0 \quad \forall x\in \bM\;, \forall r \in \sT_+\:. \label{POSJ}
\eeq
\end{itemize}
\end{lemma}

\begin{proof}  (a) is a trivial consequence of the definition. (b) Follows from the Klein Gordon equation which is satisfied by $\Phi^\psi_n$. (c) was established in Prop. 30 \cite{M2}.\end{proof} 

We are prompt to prove that $\sM^n_{\Sigma}$ is a normalized POVM when $\Sigma$ is a spacelike flat Cauchy surface. In particular $\sM^n_\Sigma(\Delta)\geq 0$.

\begin{theorem}\label{MAIN2} If $\Sigma\subset \bM$ is a spacelike flat Cauchy surface, the family of operators (\ref{MPOVM}), when $\Delta \in \cL(\Sigma)$, defines a normalized $\cH$-POVM. Furthermore, if $\psi \in \cD(\cH)$ (or more generally $\psi \in  {\cal S}(\cH)$),
\beq \langle \psi|\sM^n_{\Sigma}(\Delta)\psi\rangle = \int_\Delta  T^{\psi}_{\mu\nu}(x)_n  n^\mu n^\nu_\Sigma d\nu_\Sigma(x)\:,  \quad  \Sigma \in \cC^{sf}_\bM\:. \label{Ternoinformal2}\eeq
\end{theorem}

\begin{proof} 
First of all  $\langle \psi|\sM^n_{\Sigma}(\Delta)\psi\rangle \geq 0$  if $\psi\in {\cal S}(\cH)$. Indeed,  from Lemma \ref{LEMMASWAP} and expanding (\ref{TTT}) as prescribed in (\ref{wavePHI}), we find
$$\langle \psi|\sM^n_{\Sigma}(\Delta)\psi\rangle =$$
$$\int_{\Delta} \sp \int_{\sV_{m,+}}   \int_{\sV_{m,+}}  \sp\sp \frac{e^{i(q-p)\cdot x}}{(2\pi)^3} \frac{E_n(p)E_{n_\Sigma}(q)+ E_n(q) E_{n_\Sigma}(p) -
  n\cdot n_\Sigma(p\cdot  q + m^2)}{2 \sqrt{ E_n(q) E_n(p)}} \overline{\psi(p)}\psi(q)   d\mu_m(q)  d\mu_m(p)  d\nu_{\Sigma}(x) $$
$$= \int_\Delta J^\psi_n\cdot n_\Sigma\: d\nu_\Sigma(x) = \int_\Delta  n^\nu T^{\psi \mu}_{\nu}(x)_n\cdot n_{\Sigma\mu}\: d\nu_\Sigma(x) \geq 0
$$ 
on account of (\ref{POSJ}) and where the integral is finite because $J_n^{\psi} \in \cS(\Sigma)$.
(The found identity also establishes (\ref{Ternoinformal2}).)
 As a consequence $\sM^n_{\Sigma}(\Delta) \geq 0$ because, as already observed, $\langle \psi|\sM^n_{\Sigma}(\Delta)\psi\rangle$ is the limit of analogous matrix elements with $\psi \in \cD(\cH)$ or ${\cal S}(\cH)$. Normalization of the POVM has been already discussed beforehand, and   (\ref{MPOVM}) itself implies weak $\sigma$-additivity from the analogous property  of $\sQ_\Sigma$.
\end{proof} 

 Regarding the connection between  (\ref{MPOVM}) and (\ref{Ternoinformal}),  from  $E_p^{-1/2}(p)$ in (\ref{wavePHI}) and the expression  (\ref{TTT}) of the stress energy tensor, it is not difficult to see that (\ref{Ternoinformal2}) is nothing but the matrix element of 
 (\ref{Ternoinformal}) with respect a one-particle state $\psi \in \cH$.

  Referring to the spatial localization observable $\{\sM^n_\Sigma\}_{n\in \sT_+, \Sigma\in \cC^{sf}_\bM}$, the $\cP_+$-covariance relations analogous to (\ref{ACOVNWT20}) are valid (Thm 37 in \cite{M2})
\beq 
 U_{h}  \sM^n_{\Sigma}(\Delta) U_{h}^{-1} = \sM^{\Lambda_h n}_{h\Sigma}(h\Delta) \:, \quad \forall \Sigma \in \cC^{sf}_\bM\:, \forall \Delta \in \cL(\Sigma)\:, \quad \forall h=(v_h, \Lambda_h) \in \cP_+\label{ACOVNWT2}\:.
\eeq
The validity of these relations is actually required in the spirit of the  very definition of {\em relativistic} spatial localization observable assumed in \cite{M2} (Definition 18 therein). And, in fact,  $\{\sM^n_\Sigma\}_{n\in \sT_+, \Sigma\in \cC^{sf}_\bM}$ satisfies that definition in a broader sense, due to the presence of the further specification\footnote{If $n=n_\Sigma$, giving rise to the {\em Terno spatial localization observable} studied in first part of \cite{M2}, it fully satisfies that definition} $n\in \sT_+$.\\

Let us pass to discuss causality properties.

\begin{theorem}[\cite{M2}]  For a given choice of $n\in \sT_+$, the spatial localization observable  $\{\sM^{n}_{\Sigma}\}_{\Sigma\in \cC^{sf}_\bM}$ satisfies the CC in Definition \ref{DEFC0}.
\end{theorem}

\begin{proof}  Theorem 39 \cite{M2}.\end{proof}

\begin{corollary}  If $\psi \in \cH$, the localization probability associated to the  spatial localization observable 
$\{\sM^{n}_{\Sigma}\}_{\Sigma\in \cC^{sf}_\bM}$ cannot be  zero outside a bounded set in $\Sigma$. 
\end{corollary}

\begin{proof}
If states producing localized probability distributions  as above exist   CT  would fail as a consequence of Hegerfeldt’s theorem on relativistic time
evolution (Theorem \ref{HTeo}) and thus also CC would be false. (See \cite{M2}) for a discussion on this point.)
\end{proof}

 \begin{remark} {\em   In spite of this obstruction, it is possible to show  \cite{M2} that probability distributions localized in bounded sets can be arbitrarily well approximated by probability distributions arising by suitable sequences of state $\psi_k$
  when $n=n_\Sigma$. However it is not difficult to generalize this result to $n\neq n_\Sigma$.} \hfill $\blacksquare$
\end{remark}

  As  already observed for the POVMs $\sT^g_\Sigma$, there is again  the interesting  relation between the {\em first moment} of the  POVM $\sM^n_\Sigma$ on a spacelike flat Cauchy surface $\Sigma$ and the {\em Newton-Wigner}  selfadjoint operators $N_\Sigma^1,N_\Sigma^2,N_\Sigma^3$ \cite{M2} associated to a Minkowski chart $x^0,x^1,x^2,x^3$ such that the slice $x^0=0$ coincides with $\Sigma$.  Generalizing  Thm 26 in \cite{M2}, one sees that  the following theorem is true. Below, 
$\Delta_\psi  x^a_{\sM^n_\Sigma}$ denotes the standard deviation for the coordinate $x^a$ referred to  the probability distribution $\langle \psi| \sM^n_\Sigma(\cdot) \psi\rangle$ constructed out of the POVM $\sM^n_\Sigma$ in the state defined by the unit vector $\psi$.

\begin{theorem} \label{1MNW2} Referring to the POVM $\sM^n_\Sigma$ for $\Sigma \in \cC^{sf}_\bM$, the following facts are true.
\begin{itemize}
\item[(a)]  The  $a$-first moment of $\sM^n_\Sigma$ is defined for  every $\psi\in {\cal S}(\cH)$ with $||\psi||=1$ and 
\beq \int_{\Sigma} x^a \langle \psi| \sM^n_\Sigma(d^3x) \psi\rangle =  \langle \psi| N_\Sigma^a \psi\rangle\quad \mbox{for $a=1,2,3$.}\label{NMN} \eeq
In particular, $N^a_\Sigma$  is  the unique selfadjoint operator in $\cH$  which satisfies the indenty above.
\item[(b)] The {\bf Heisenberg inequality} turns out to be corrected as,  for $\psi\in {\cal S}(\cH)$,
$$\Delta_\psi  x_{\sM^n_\Sigma}^a \Delta_\psi P_a  \geq \frac{\hbar}{2} \sqrt{1+4 (\Delta_\psi P_a)^2 \langle \psi|\sK^{\sM^n_{\Sigma}}_{a} \psi\rangle}\:, \quad a=1,2,3 $$
where $\sK^{\sM^n_{\Sigma}}_{a}\in \gB(\cH)$ is a selfadjoint  operator  which is a (spectral)  function of the four  momentum observable $P$ with the form (\ref{kappa}) and $\sK^{\sM^n_{\Sigma}}_{a}\geq 0$.
\end{itemize}
\end{theorem}

\begin{proof} See Appendix \ref{APPENDIXA}. \end{proof}

\noindent Observe  that the result in (\ref{NMN})  does not depend on $n$.

\begin{remark}\label{NAICOM2} {\em
\begin{itemize}
\item[(1)] If also the identity $$ \int_{\Sigma} (x^a)^2 \langle \psi| \sM^n_{\Sigma}(d^3x) \psi \rangle =  \langle \psi| (N^a_\Sigma)^2\psi\rangle\quad \mbox{(false!)}\:,$$ 
 were valid one could apply a known theorem by Naimark about the decomposition of maximally symmetric operators (here $N^a_\Sigma$) in terms of POVMs (see Theorem 23  in \cite{DM} and the discussion about it) obtaining  $\sM^n_\Sigma= \sQ_{\Sigma}$. This is obviously false and it is also reflected by  the presence of the  term $ \langle \psi|\sK^{\sM^n_\Sigma}_{a} \psi\rangle$ in the modified  Heisenberg inequality.
\item[(2)] If $U^{(n)}_t$ is the unitary time evolutor corresponding to the time evolution along $n$ in the spacetime $\bM$,
  it is  easy to see that the Heisenberg evolution  $U^{(n)}_t N_\Sigma^a U^{(n)\dagger}_t$  of $N_\Sigma^a$ on the right-hand side of (\ref{NMN}) equals the integral on the left-hand side over the correspondingly temporally translated time slice $\Sigma_t$. As already observed in \cite{M2}, this  fact implies that the worldline $\bR \ni t \mpasto (t,  \int_{\Sigma_t} \vec{x} \langle \psi| \sM^n_{\Sigma_t}(d^3x) \psi\rangle)$ is {\em timelike} (Corollary 14 in \cite{M2}) as expected by massive particles. \hfill $\blacksquare$
 \end{itemize}} 
 \end{remark}

\subsection{A spacelike Cauchy localization observable $\sM^n= \{\sM^n_S\}_{S\in \cC_\bM^s}$
 for massive KG particles.}
We are in a position to prove that a  normalized POVM  $\sM^n_S$ exists on every spacelike Cauchy surface $S$ of $\bM$. We shall also obtain  that the elements of the  POVMs  do not depend on the Cauchy surface they belong to. In other words 
we have a spacelike Cauchy localization  $\sM^n$.

\begin{theorem} \label{MAIN3}   Consider $S\in \cC^s_\bM$  and $n\in \sT_+$. Then,
 \begin{itemize}
\item[(a)] there is a unique   $\cH$-POVM, still indicated by $\{\sM^g_S(\Delta)\}_{\Delta \in \cB(S)}$, which  satisfies (\ref{Ternoinformal2})  also for generic spacelike Cauchys surface $S$  when $\psi \in {\cal D}(\cH)$:
\beq  \langle\psi|\sM^n_{S}(\Delta)\psi\rangle = \int_\Delta  T^{\psi}_{\mu\nu}(x)_n  n^\mu n^\nu_S d\nu_S(x)\:,\quad \Delta \in \cB(S)\:,  \label{SNF}\eeq
so that, $\sM^n_S < \sp < \nu_S$.
\item[(b)] $\sM^n_S$ is normalized: $\sM^n_S(S)=I$; 
\item[(c)] $\sM^n_\Sigma$ satisfies (\ref{MPOVM}) when $\Sigma \in \cC_\bM^{sf}$;

\item[(d)] if $\Delta \in S\cap S'$ is a Borel set, where $S\in \cC^s_\bM$, then
$$\sM^n_S(\Delta) =\sM^n_{S'}(\Delta)\:.$$
\end{itemize}
The family $\sM^n := \{\sM^n_S\}_{S\in \cC^s_\bM}$,  where we extend each  POVMs to $\cM(S)$  according to Proposition \ref{ACM}, is a spacelike Cauchy localization  observable  according to Definition \ref{GSL}.
\end{theorem}

\begin{proof} Item (c) is a trivial consequence of (a) and (b), so we prove (a), (b), and (d).\\
Let   $x^0=t_S(\vec{x})$ be the map which defines $S$ by identifying it with the spatial $\bR^3$ in a given   Minkowski system of coordinates $x^0, \vec{x}$ according to Proposition \ref{propS}.  
If $\Delta \in \cB(\bR^3)$ is bounded consider the Hermitian form 
$$\Lambda(\psi,\psi') := \int_\Delta \omega^{J^{\psi\psi'}_n}\:,$$
where, taking (\ref{wavePHI}) into account,
$$T^{\psi\psi'}_{\mu\nu}(x)_n := \frac{1}{2}\left(\partial_\mu \overline{\Phi^\psi_n(x)}\partial_\nu\Phi^{\psi'}_n(x) 
 +\partial_\mu \Phi^{\psi'}_n(x)\partial_\nu\overline{\Phi^{\psi}_n(x)}\right)$$ \beq - \frac{1}{2}\eta_{\mu\nu} \left( \partial^\alpha \overline{\Phi^\psi_n(x)} \partial_\alpha \Phi^{\psi'}_n(x) + m^2 \overline{\Phi^\psi_n(x)} \Phi^{\psi'}_n(x) \right)\:,\label{TTT2}
\eeq
and
\beq
J^{\psi\psi'\mu}_{n}(x) := n^\nu T^{\psi \psi' \:\mu}_{\nu}(x)_n \:,\label{53F}
\eeq
As $\Delta$ is bounded, the integral is well defined. Since (use (\ref{VVN}))
we conclude that 
\beq 0\leq \Lambda(\psi,\psi) \leq \langle \psi|\psi \rangle\:. \label{INEQ2}\eeq
The former inequality arises from the last identity in (\ref{VVN}) {\em when observing  that (see (\ref{53})) $J^{\psi\psi}_n= J^{\psi}_n$}  is zero or causal and future directed as established in Lemma \ref{LEMMALAST}. The latter inequality in (\ref{INEQ2}) is consequence of the positivity of the integrand $ \omega^{J^{\psi}_n}$ in coordinates (again for (\ref{VVN})) and (\ref{POSJ}) which first of all  imply (we write  $J^{\psi}_n= J^{\psi\psi}_n$)
$$\int_\Delta \omega^{J^{\psi}_n} \leq \int_S \omega^{J^{\psi}_n}\:.$$
On the other hand
$$ \int_S \omega^{J^{\psi}_n} = \langle \psi|\psi\rangle$$
because we can apply Proposition \ref{PINT} (since $J^{\psi}_n$ is also smooth,  bounded, and conserved according to Lemma \ref{LEMMALAST}) choosing a spacelike flat surface $S'$, obtaining 
\beq \int_S \omega^{J^{\psi}_n} =  \int_{S'} \omega^{J^{\psi}_n}  =\langle \psi|\psi\rangle\:.\label{CNORMN}\eeq
The last identity is due to the fact that the last integral is nothing but $\langle \psi|\sM^n_{S'}(S')\psi\rangle = \langle\psi| \psi\rangle$ due to Theorem \ref{MAIN2}.\\
We can finally apply Lemma \ref{LEMMAAGG} proving that there exists $\sM^n_S(\Delta) \in \gB(\cH)$ such that $0\leq M^n_S(\Delta)\leq I$
and \beq \langle \psi|M^n_S(\Delta) \psi\rangle = \int_\Delta \omega^{J^{\psi}_n} = \int_\Delta  T^{\psi}_{\mu\nu}(x)_n  n^\mu n^\nu_S d\nu_S(x)\:,\label{INTFI2}\eeq so that (\ref{SNF}) is satisfied for $\Delta$ bounded.
From this point on, the proof is identical to the one of Theorem \ref{MAIN}. In particular, the normalization condition (b) follows now from (\ref{CNORMN}).  The last statement of the thesis  immediately arises from Proposition \ref{PROPCOEES}.
\end{proof} 
\section{Moments of $\sT^g_\Sigma$ and  $\sM^n_\Sigma$,  Newton-Wigner operator, and Heisenberg inequality}\label{NWApp}
We addresss the reader to the  discussion and references in \cite{M2} about the Newton-Wigner observables for a massive  spinless particle. 

Referring to the representation $L^2(\bR^3, d^3p)$ (see Sect. \ref{OPS})  of the one-particle space $\cH$ through the Hilbert space  isomorphism (\ref{PHI}) $F:\cH \to L^2(\bR^3, d^3p)$, the joint PVM of the selfadjoint  Newton Wigner operators $N^1,N^2,N^3$ for a  massive spinless particle, associated  to the coordinates $\vec{x}=(x^1,x^2,x^3)$ on $\bR^3$ in a Minkowski frame $x^0,x^1,x^2,x^3$ at time $x^0=0$, is given by
\beq
(F\sQ(\Delta)F^{-1}\phi)(\vec{p}) := \int_\Delta \frac{e^{i(\vec{q}- \vec{p})\cdot \vec{x}}}{(2\pi)^3} \phi(\vec{q}) d^3q\:, \quad \phi \in\cS(\bR^3) \quad \mbox{or $C_c^\infty(\bR^3)$}\:, \quad \Delta \in \cB(\bR^3)\:.
\eeq 
It  extends by continuity to the whole Hilbert space. 
Notice that ${\cal S}(\cH) = F(\cS(\bR^3))$ and  $\cD(\cH) = F(C_c^\infty(\bR^3))$ are invariant spaces and  cores for each $N^a$ and thereon they are unitarily equivalent  to the respective differential operator $FN^a|_{\cS(\bR^3)}F^{-1} = i \frac{\partial}{\partial p_a}$.

Theorems \ref{1MNW1} and \ref{1MNW2} are subcases of the following result.

\begin{proposition}\label{PROPVAR}  Consider a normalized POVM on $\bR^3$ satisfying (with $\phi_\psi := F\psi$)
\beq \langle \psi|\sA(\Delta) \psi\rangle = \int_{\Delta}d^3x\int_{\bR^3}\int_{\bR^3}  d^3p d^3q \frac{e^{i(\vec{q}-\vec{p})\cdot \vec{x}}}{(2\pi)^3}\overline{\phi_\psi(\vec{p})}K_\sA(\vec{p}, \vec{q}) \phi_\psi(\vec{q})\:,  \quad \psi \in \cD(\cH)\:, \:\Delta \in \cB(\bR^3)\label{NPOS}\eeq
where $K_\sA$ is a positive definite kernel which satisfies $K_\sA(\vec{p}, \vec{p})=1$ for every $\vec{p} \in \bR^3$.
\begin{itemize} 
\item[(a)] If $K_\sA$ is  real and smooth, then \beq \int_{\bR^3} x^a \langle \psi| \sA(d^3x)\psi \rangle =  
\langle \psi |N^a \psi \rangle\quad \mbox{for $a=1,2,3$ and $\psi \in \cD(\cH)$, with $||\psi||=1$}\:,\label{IDA}\eeq
where $N^a$ is the unique selfadjoint operator in $\cH$ which satisfies the identity above.\\
More generally, if $\alpha:= (\alpha_1,\alpha_2,\alpha_3)$ is a multi index, so that $x^\alpha := (x^1)^{\alpha_1} (x^2)^{\alpha_2} (x^3)^{\alpha_3}$,
$$\int_{\mathbb{R}^3}(x^1)^{\alpha_1}(x^2)^{\alpha_2} (x^3)^{\alpha_3}  \langle\psi|A(d^3x)\psi\rangle=\langle \psi|(N^1)^{\alpha_1}(N^2)^{\alpha_2}(N^3)^{\alpha_3}\psi\rangle $$ \beq +\sum_{1<|\beta|, \beta\leq \alpha}\binom{\alpha}{\beta}\int_{\mathbb{R}^3}\left(i^{|\alpha|-|\beta|}\frac{\partial^{|\alpha|-|\beta|}\overline{\phi_\psi(\vec{p})}}{\partial p^{\alpha-\beta}}\right)\left.\left(i^{|\beta|}\frac{\partial^{|\beta|}K_\sA(\vec{p},\vec{q})}{\partial p^\beta}\right)\right|_{\vec{q}=\vec{p}}\phi_\psi(\vec{p})d^3p.\label{GENM}\eeq
\item[(b)]  If $K_\sA$  is real, smooth,  and has polynomial growth with all of its  derivatives of any order, then 
(\ref{NPOS}), (\ref{IDA}), (\ref{GENM}) are also valid for  $\psi \in  {\cal S}(\cH)$.
\item[(c)]  If $K_\sA$ is as in (b),  a formula for the second moment holds for $\psi \in {\cal S}(\cH)$ with $||\psi||=1$
\beq \int_{\bR^3} (x^a)^2 \langle \psi| \sA(d^3x)\psi \rangle  =  
\langle \psi |(N^a)^2 \psi \rangle  +  \langle \psi|\sK^{\sA}_{a} \psi\rangle\quad \mbox{for $a=1,2,3$}\:. \label{IDA33}\eeq
 $\sK^{\sA}_{a} \geq 0$  is the  (generally unbounded)  multiplicative  selfadjoint operator
\beq \left(\sK^{\sA}_{a} \psi\right)(p)   :=  \sqrt{p^0}
 \left( \left.\frac{\partial}{\partial q_a}\frac{\partial}{\partial p_a}K_\sA(\vec{q}, \vec{p})\right|_{\vec{p}=\vec{q}} \right)\frac{\psi(p)}{\sqrt{p^0}}\:, \quad \psi \in D(\sK^{\sA}_{a}) \label{kappa}\:,\eeq
with $\sK^{\sA}_{a}\in \gB(\cH)$ if  $\bR^3 \ni \vec{p}\mapsto \left.\frac{\partial}{\partial q_a}\frac{\partial}{\partial p_a}K_\sA(\vec{q}, \vec{p})\right|_{\vec{p}=\vec{q}} \in \bR$ is bounded.\\
\item[(d)]   If $K_\sA$ is as in (b),  a {\bf modified Heisenberg inequality} holds for $\psi \in {\cal S}(\cH)$ with $||\psi||=1$  (restoring Plank's constant),
\beq
\Delta_\psi  x_\sA^a \Delta_\psi P_a  \geq \frac{\hbar}{2} \sqrt{1+4 (\Delta_\psi P_a)^2 \langle \psi|\sK^{\sA}_{a} \psi\rangle}\:, \quad a=1,2,3\:.
\eeq
Above, $\Delta_\psi x_\sA^a$ is the standard deviation of the probability distribution $\langle \psi|\sA(\cdot)\psi\rangle$,  $\Delta_\psi P_a$ is the standard deviation of the probability distribution of th $a$-component of the momentum observable in the state $\psi$.
\item[(e)]  Suppose that  $K_\sA$ is not necessarily  real but is  {\bf finite} according to \cite{C}. In other words, it has the form, for $N< +\infty$  suitable measurable functions $u_j$
$$K_\sA(\vec{p}, \vec{q}) = \sum_{j=1}^N \overline{u_j(\vec{p})} u_j(\vec{q})\:.$$ In addition,  assume that the $N$  functions $u_j : \bR^3 \to \bC$ are smooth with polynomial growth with all of their derivatives of any order. Then the standard Heisenberg inequality holds in any cases  for  $\psi \in {\cal S}(\cH)$ with $||\psi||=1$
$$\Delta_\psi  x_\sA^a \Delta_\psi P_a  \geq  \frac{\hbar}{2}\:.$$
\end{itemize}
\end{proposition}
\begin{proof}   See Appendix \ref{APPENDIXA} \end{proof}

\section{Discussion}\label{DISC} 
The major achievement of this work (Theorem \ref{TONE1}) is the result that, if we use POVMs to describe the probability of spatial localization of quantum systems in a sufficiently general way,  then {\em spatial localization}, in the sense of  a {\em spacelike Cauchy localization observable} Definition \ref{GSL},  {\em implies causality}, in the terms of our {\em general causal condition} GCC in Definition \ref{DEFC}. The physical postulates  at the ground of this implication, encapsulated in the notion  of spacelike Cauchy localization observable,  are the following ones. \begin{itemize}\item[(1)] It is supposed that every {\em spacelike} Cauchy surface can be used to localize the system. In other words, if we fill a spacelike Cauchy surface with a net of detectors, we must  find somewhere the quantum system on that 3-space. \item[(2)]  If, for a given spacelike Cauchy localization observable, i.e., for a specific type of detectors,  a pair of Cauchy surface  coincide in a region, then  they share the same detectors therein.
\item[(3)] There is no chance to detect the quantum system  in a spatial region  with zero measure on a spacelike Cauchy surface.
\end{itemize}

An observation about the need for  condition (2), we named {\em coherence condition}, is important. As we have seen in  Remark (\ref{REMLASTC}),   condition (2) can be removed from the definition of  spacelike Cauchy localization observable and inserted as a further  hypothesis  of  Theorem \ref{TONE1}. Within  this scheme, the coherence condition would be recovered as a corollary of  Theorem \ref{TONE1}. In other words, if (1) and (3) are valid for a family of POVMs, then the {\em coherence condition} is {\em equivalent} to GCC.

We also proved, as a second achievement in double  form (Theorems  \ref{MAIN}  and \ref{MAIN3}):  The above general notions of localization exist at theoretical level. In fact, 
we presented two of them for massive Klein-Gordon particles. In the second case the localization observable was constructed in terms of physical quantities of the system (its energy). There is no evident obstruction to construct similar spacelike Cauchy localization observables for other types of particles like fermions considered in \cite{C0}.  The basic ingredient to construct these observables is a conserved causal current constructed out of the state of the system. It seems plausible that spacelike Cauchy localization observables can be built up also in a more general spacetime (referring to the one-particle structure of quantum field Gaussian states) provided the spacetime is  globally hyperbolic. This is because the central technology to produce the former achievement, some technical  results about Cauchy surfaces \cite{BS,BS3}  are at disposal in generic globally hyperbolic spacetimes. On the other hand the explicit  structure of spacelike Cauchy localization observable relies upon notions, like conserved currents,  which can be generalized
to every globally hyperbolic spacetime.  The only mathematically delicate issue which deserves attention is the fact that the region of influence $\Delta'$ of a Borel set $\Delta$ is again Borel or in a natural completion of that $\sigma$-algebra.

When coming back to flat Cauchy surfaces, i.e., rest frames of inertial observers, the resulting spatial localization observables show interesting features.   Since causality is not violated by the distribution of probability of a Klein-Gordon  massive particle, the no-go Hegerfeldt theorem is made  harmless. There is a price to pay however:   no strictly localized (in bounded spatial sets) probability distributions are permitted. Another  interesting fact, already evident in \cite{M2}, is that the Newton-Wigner operators insist to play some role in this much less naive picture, in spite of the fact that  the Hegerfeldt no go results seemed to have ruled out them long time ago.  Even if they no longer  represent observables, they account for the {\em timelike} spacetime evolution of the first moment of a massive particle.

A widely open issue is the relation between the constructed POVMs and their decompositions in terms of Kraus operators or quantum instruments. Related to this issue is the fact that the effects of the presented POVMs do not commute even when are localized in causally separated sets. This  is an urgent problem when analysing all the construction from  the perspective of the local operator algebras theory.  These outstanding problems will be investigated elsewhere.

\section*{Acknowledgments} V.M. is  very grateful to D.P.L.Castrigiano  for various remarks, suggestions, and discussions about several issues appearing  in this paper. 
We thank  N.Pinamonti and   M.S\'anchez for helpful discussions on some technical problems, and S. Lill and   M. Reddiger 
for pointing out some relevant literature.
This work has been written within the activities of INdAM-GNFM

\appendix

\section{Proof of some propositions}\label{APPENDIXA}
\noindent {\bf Proof of Proposition \ref{propS}}.  Every coordinate curve $\bR \ni x^0 \mapsto (x^0, \vec{x})\in\bR\times  \bR^3 \equiv \bM$ is smooth timelike and thus it intersects  $S$ exactly once for a corresponding value $x^0= t_S(\vec{x})$, therefore it defines a map $\bR^3 \ni \vec{x} \mapsto t_S(\vec{x}) \in \bR$. In the rest of  
 the proof write $t$ in place of $t_S$ for shortness.  By definition $S \equiv \{(t(\vec{x}), \vec{x})\:|\: \vec{x} \in \bR^3 \}$. Since $S$ is an embedded co-dimension $1$ submanifold, it must be locally described as the locus $f(x^0,x^1,x^2,x^3)=0$ of a smooth map with $\gor(df, df) \leq 0$ because the $g$-normal vectors to $df$ are spacelike or lightlike and thus $df$ is timelike or lightlike. It must be in particular $\partial_{x^0}f \neq 0$. As a consequence of the implicit function theorem, in a neighborhood of every $p\in S$, we can represent $S$ as a smooth map $x^0=t'(\vec{x})$. Therefore $t(\vec{x}) = t'(\vec{x})$
is locally smooth and thus smooth. Finally, as $S$ is the locus of $g(x^0,x^1,x^2,x^3)=0$ for $g(x^0,x^1,x^2,x^3) = x^0 - t(\vec{x})$, a normal co-vector to $S$ is $v=dx^0 - \sum_{k=1}^3 \frac{\partial t}{\partial x^k} dx^k$. If $S$ is spacelike, then its normal vector is timelike,  written as $\gor(v,v)<0$, is equivalent to $|\nabla t| <1$. 
In the other cases $v$ can also be lightlike $\gor(v,v) \leq 0$ which  is equivalent to $|\nabla t| \leq 1$.\\
The final statement is now obvious, since the wanted diffeomorphism is the identity map $\bR^3\ni p \mapsto p \in  \bR^3$ when adopting on $S$ the coordinates $\vec{x}$. \hfill $\Box$\\

\noindent {\bf Proof of Proposition \ref{ACM}}. 
We refer to  \cite{Cohn} for the definition of completion of a $\sigma$-algebra and a positive mesure on it.
We divide the main proof into some further lemmata.

\begin{lemma}\label{LEMMAM0} Let $\mu\colon \cB \rightarrow\overline{\mathbb{R}}_+$ be a positive measure on a $\sigma$-algebra $\cB$ on $X$, $\overline{\cB}^\mu$ be the completion of $\cB$ with respect to $\mu$, and let us denote by $\overline{\mu}$ the completion measure on  $\overline{\cB}^\mu$.\\ Then
$\Delta \in \overline{\cB}^\mu$ if and only if $\Delta = \Delta_0 \cup N$ with $\Delta_0 \in \cB$ and $N \subset N_0 \in \cB$ with $\mu(N_0)=0$. Furthermore, for every decomposition of $\Delta$ as above,  $\overline{\mu}(\Delta)= \mu(\Delta_0)$.
\end{lemma}

\begin{proof} Direct inspection. \end{proof}

\begin{lemma} \label{LEM} Let $\mu\colon \cB \rightarrow\overline{\mathbb{R}}_+$ be a positive measure on a $\sigma$-algebra $\cB$ on $X$ and $\cL :=
\overline{\cB}^\mu$ be the completion of $\cB$ with respect to $\mu$. Let $\nu\colon\cB\to \overline{\bR_+}$ be a positive   measure such that $\nu <\sp <\mu$. Then the completion $\overline{\cB}^\nu$ of $\cB$ respect to $\mu$ satisfies
$\cL \subset \overline{\cB}^\nu\:,$
 and there exists a unique positive  measure $\tilde{\nu}$ which extends $\nu$ on $\cL$ and $\tilde{\nu}<\sp<\overline{\mu}$, where $\overline{\mu}$ is the completion  of $\mu$ on $\cL$. Finally $\tilde{\nu}(X)= \nu(X)$.
\end{lemma}

\begin{proof}
(Uniqueness) By Lemma \ref{LEMMAM0},  if  $\Delta\in\cL$ we can decompose  $\Delta=\Delta'\cup N$ where $\Delta\in\cB$ and $N\subset \tilde{N}\in\cB$ is such that $\mu(\tilde{N})=0$. Suppose that  $\tilde{\nu}$ extends $\nu$ on $\cL$. Sub-additivity and monotony yield
$$ \tilde{\nu}(\Delta)\leq\tilde{\nu}(\Delta')+\tilde{\nu}(N)\leq\tilde{\nu}(\Delta')+\tilde{\nu}(\tilde{N})=\nu(\Delta')+\nu(\tilde{N})=\nu(\Delta')$$
where, in the last step, we used $\nu< \sp<\mu$. By monotony $\tilde{\nu}(\Delta)\geq\tilde{\nu}(\Delta')=\nu(\Delta')$, therefore $\tilde{\nu}(\Delta)=\nu(\Delta')$. In summary,  an extension has to be absolutely continuous with respect to $\overline{\mu}$. Indeed $\overline{\mu}(\Delta)=0$ implies $\mu(\Delta')=0$ and therefore $\tilde{\nu}(\Delta)=\nu(\Delta')=0$.\\
(Existence)
Let $\overline{\cB}^\nu$ be the completion of $\cB$ with respect to $\nu$ and $\overline{\nu}$ the completion  of $\nu$ on $\overline{\cB}^\nu$. Notice that $\cL\subset\overline{\cB}^\nu$, indeed $\mu(\tilde{N})=0$ implies $\nu(\tilde{N})=0$. Therefore $\tilde{\nu}:=\overline{\nu}\left|\right._\mathcal{L}$ is a well defined finite  measure which extends $\nu$ on $\mathcal{L}$.  
Since $X \in \cB \cap\overline{\cB}^\nu$, we finally have  $\tilde{\nu}(X)= \nu(X)$.
\end{proof}

\noindent In the following, if $\sE : \cB \to \gB(\cH)$ is a $\cH$-POVM and $\psi,\phi \in \cH$, $\mu^\sE_{\psi,\phi}(\Delta) := \langle \psi |\sE(\Delta) \phi \rangle$
for $\Delta \in \cB$. Furthermore  $\mu^\sE_{\psi}(\Delta) :=  \mu^\sE_{\psi,\psi}(\Delta)$.

\begin{lemma}\label{LEM2}
Let $\sA : \cB\to  \gB(\cH)$ be a normalized $\cH$-POVM and $\mu:\cB\to \overline{\mathbb{R}}_+$ be a positive measure on a $\sigma$-algebra $\cB$ such that, for every  $\psi\in\sH$, $\mu^\sA_\psi<\sp<\mu$. Then there exists a unique normalized $\cH$-POVM $\tilde{\sA}$ on $\cL :=
\overline{\cB}^\mu$  which extends $\sA$ and such that for every $\mu^{\tilde{A}}_\psi<\sp <\overline{\mu}$.
\end{lemma}

\begin{proof}
Taking advantage of Lemma \ref{LEM}, for every  $\psi\in\sH$, we can extend $\mu^\sA_\psi$ to $\widetilde{\mu^\sA_\psi}$ on $\cL$. Then, we define 
\begin{equation*}
	\widetilde{\mu^\sA_{\phi,\psi}}(\Delta):= 	\widetilde{\mu^\sA_{\phi+\psi}}(\Delta)-	\widetilde{\mu^\sA_{\phi-\psi}}(\Delta)-	\widetilde{\mu^\sA_{\phi-i\psi}}(\Delta)+	\widetilde{\mu^\sA_{\phi+i\psi}}(\Delta)\:.
\end{equation*}
For a fixed $\Delta \in \cL$ -- $\Delta = \Delta' \cup N$ as in the proof of  Lemma \ref{LEM} -- the map $\cH\times \cH \ni (\phi,\psi)\mapsto \widetilde{\mu^\sA_{\phi,\psi}}(\Delta)$ is  sesquilinear and continuous since $\widetilde{\mu^\sA_{\phi,\psi}}(\Delta)=	{\mu^\sA_{\phi,\psi}}(\Delta')$. As a consequence, for  $\Delta\in \cL$, the map 
$\cH \ni \psi\mapsto \widetilde{\mu^A_{\phi,\psi}}(\Delta) \in \mathbb{C}$
 is a linear and continuous functional with norm bounded $||\psi||$. A straightforward use of  the Riesz lemma gives that there is a unique operator $\tilde{A}(\Delta) : \cH \to \cH$
 such that $\langle \psi|\tilde{A}(\Delta) \phi \rangle = \widetilde{\mu^A_{\phi,\psi}}(\Delta)$. $\tilde{A}(\Delta)$is  positive because $\mu^{\tilde{A}}_\psi$ is positive. Bu construction $\cL \ni \Delta \mapsto \tilde{A}(\Delta)\in \cB(\cH)$
  defines a normalized POVM. Normalization arises from $\langle \phi |  \tilde{\sA}(S) \psi \rangle =   \widetilde{\mu^\sA_{\phi,\psi}}(S) =
 \mu^A_{\phi,\psi}(S) = \langle \phi|\psi \rangle$.
  Finally, $\mu^{\tilde{\sA}}_\psi=\widetilde{\mu^\sA_\psi}< \sp<\overline{\mu}$. 
\end{proof}

\noindent We can now conclude the proof of the main proposition.
 If  $S\in\cC^s_{\mathbb{M}}$, set $\cB=\mathscr{B}(S)$, $\sA:= \sA_S: \cB(S) \to \gB(\cH)$ and $\mu=\nu_S$ in Lemma \ref{LEM2}. The map $\mathscr{M}(S)\ni\Delta\mapsto\tilde{A}(\Delta)\in \gB(\cH)$ defines the wanted extension.  It is unique because every extension of the considered type is completely determined by the positive measure $\widetilde{\mu^\sA_{\psi}}$ which is uniquely determined by $\mu^\sA_{\psi}$, thus by $\sA$ itself.  \hfill $\Box$\\

\noindent{\bf Proof of Proposition \ref{PROPCOEES}}. We have only to prove that the extended POVMs satisfy the coherence condition (\ref{COHERENCE}). Take $\psi \in \cH$. If $\Delta \in \cM(S)$ then, according to the construction in the proof of Proposition \ref{ACM}, there exists a decomposition $\Delta = \Delta_0 \cup N$ with $\Delta_0\in \cB(S)$ and $N\subset N_0 \in \cB(S)$,  and finally
$\langle\psi|\sA_S(N_0)\psi \rangle=0$, $\langle \psi|\tilde{\sA}_S(\Delta)\psi \rangle = \langle \psi|\sA_S(\Delta_0)\psi \rangle$.  These last two identities do not depend on the decomposition of $\Delta$ as $\Delta_0\cup N$. Since $\Delta_0,N,N_0 \subset \Delta \subset S\cap S'$, we can repeat everything for $\langle \psi|\sA_{S'}(\Delta)\psi \rangle$ with the same decomposition of $\Delta$, proving that 
$$\langle \psi|\sA_{S'}(N_0)\psi \rangle = \langle \psi|\sA_{S}(N_0)\psi \rangle=0\:, \quad \langle \psi|\tilde{\sA}_{S'}(\Delta)\psi \rangle=\langle \psi|\sA_{S'}(\Delta_0)\psi \rangle=\langle \psi|\sA_{S}(\Delta_0)\psi \rangle=\langle \psi|\tilde{\sA}_{S}(\Delta)\psi \rangle$$ where we used the fact that $\sA$ satisfies the coherence condition on the Borel sets $\Delta_0, N_0\subset S\cap S'$ by hypothesis. In particular 
$\langle \psi|\tilde{\sA}_{S'}(\Delta)\psi \rangle=\langle \psi|\tilde{\sA}_{S}(\Delta)\psi \rangle$.
Arbitrariness of $\psi$, using polarization and the fact that $\cH$ is complex, proves the validity of the coherence condition for the extended POVMs ending the proof. \hfill $\Box$\\


\noindent {\bf Proof of Proposition \ref{PROPJ}}. Notice that (\ref{SECOND}) is nothing but (\ref{FIRST}) when exploiting (\ref{FJ}).
  Evidently, it is sufficient to establish the thesis when $S'$ is a Minkowski time slice and $S$ is a generic smooth  Cauchy surface.
From now on, we consider a Minkowski coordinate frame $x^0,x^1,x^3,x^4$ and describe $S$ in terms of  a smooth function $x^0=t(\vec{x})$, taking $S'$ as the time-slice $x^0=0$.   Define the auxiliary vector field
$$K := \frac{1}{(1+ r^2)^2} \partial_{x^0}\:,\quad \mbox{where $r:= \sqrt{\sum_{k=1}^3  (
x^k)^2}$}\:.$$
By construction $K$ is conserved, everywhere non-vanishing, timelike, and future-directed. Finally,
\beq  -K\cdot  n_{S'} =  \frac{1}{(1+ r^2)^2} > 0 \quad \mbox{and} \quad  -\int_{S'}  K\cdot  n_{S'} d\nu_{S'}  = \int_{S'} \omega^K < +\infty\:.\label{CONV}\eeq
Define $B'_R$ as the open ball in $S'\equiv \bR^3$ centered at the origin and with finite radius $R>0$. Correspondingly $B_R \subset S$ is the open set in $S$ defined by $x^0= t(\vec{x})$, $\vec{x} \in B'_R$.  Let $T_R\subset \bM$ be the closed (compact) cylinder with bases $\overline{B_R}$ and $\overline{B'_R}$ and lateral surface parallel to $K$ (i.e. to $\partial_{x^0}$). 
We are explicitly assuming that the closures of  $B'_R$ and $B_R$ have no intersection for now and we shall treat later the case where the two surfaces intersect.
As $K$ is conserved, $d\omega^K=0$ for (1) Proposition \ref{PINT}. Furthermore $\partial T_R$ is an orientable  smooth embedded submanifold up to zero-masure subsets, we can apply the  Poincar\'e  theorem, proving that $\int_{\partial T_R} \omega^K  =\int_{T_R} d\omega^K = 0$. Since the lateral surface gives no contribution as easily arises from (\ref{omegaJ}) using the fact that $J$ is parallel to that surface, the identity boils down to
$$\int_{B'_R}  \omega^K  = \int_{B_R}  \omega^K \:.$$ 
This identity is valid also when $B'_R$ and $B_R$ have intersection (just because $S$ and $S'$ do and $B'_R$ passes through that intersection). In that case, as $R$ is finite so that $t(\vec{x})$ is bounded for  $\vec{x} \in B'_R \subset S' \equiv \bR^3$, we can move $S'$ parallelly to $\partial_{x^0}$ till to another $x^0$-slice  $S''$ in the past of $S'$, in order that $B_R$ and $B''_R$ (the projection of $B'_R$ onto $S''$) do not meet. By construction, using the above argument once again for the relevant embedded submanifolds with orientable boundary, since now the bases do not touch each other,
$$\int_{B''_R}  \omega^K = \int_{B'_R}   \omega^K \:,$$ 
and also 
$$\int_{B''_R}  \omega^K= \int_{B_R}   \omega^K \:.$$ 
In summary 
$$\int_{B'_R}   \omega^K = \int_{B_R}   \omega^K \:.$$ 
is valid even if $B_R$ and $B'_R$ have intersection.\\
Taking $R\to +\infty$, a direct application of the monotone convergence theorem (as the integrals can be written as Lebesgue integrals in $\bR^3$ of non negative functions according to (\ref{VVN}) proves that
\beq 0\leq \int_{S'}  \omega^K   = \int_{S}  \omega^K  < +\infty\:,\label{IDINT}\eeq
where we also used (\ref{CONV}).\\
To go on, define $Y:= J+K$ noticing that this smooth conserved vector field is everywhere timelike and future directed. Since it is also bounded, its flow is global and defines a smooth one-parameter group of diffeomorphisms  $\{\phi^{Y}_{s}\}_{s\in \bR}$ of $\bM$.
Finally, the integral curves of $Y$ are future directed and inextendible. In fact, suppose that $\gamma(s)\to q\in \bM$ for $s\to +\infty$, i.e., the curve is not future inextendible (the same argument applies for $s\to -\infty)$. If $Y^k(q)\neq 0$ for some $k=0,1,2,3$, then it happens  in a connected neighborhood $U$ of $q$ where the sign of $Y^k= \frac{dx^k}{ds}$ is therefore fixed. There we can use $z=x^k$ to parametrize $\gamma$ obtaining $s(z) = \int_0^{z} \frac{dx^k}{Y^k(\gamma(x^k))}$. However, since $s$ varies till $+\infty$ when $\gamma$ approaches $q$ in $U$, then $Y^k(\gamma(x^k))$ must vanish as $x^k \to x^k(q)$ contrarily to the hypothesis. Hence all components of $Y$ must vanish at $q$ but this was exclued {\em a priori}.\\
As a consequences of these properties, the map $$\Phi : \bR \times S' \equiv \bR \times \bR^3  \ni (s,p) \mapsto \phi^Y_s(p) \in \bM$$ turns out to be a diffeomorphism as well.
In fact, $\Phi$ is injective because $\phi^Y_s(p)=\phi^Y_{s'}(p')$ implies $p = \phi^Y_{s'-s}(p')$ but this would mean that the integral curve of $Y$ through $p$ touches twice $S'$, and this is not possible unless $s=s'$ and thus $p=p'$; surjectivity of $\Phi$ follows from the fact that, if $q\in \bM$ there is an integral line of $Y$ through $q$ at $s=0$ and, it being timelike and inextendible, it must intersect $S'$ at some $q'$ for $s=T_q$ because $S'$ is a Cauchy surface. Therefore $\Phi(-T_q,q')=q$; finally $d\Phi \neq 0$ everywhere, because (a) $d\Phi(s,p) e_0 = Y(\phi^Y_s(p))$ and   ($e_0, e_1,e_2,e_3$ denoting the canonical  basis of $\bR\times \bR^3 \equiv \bR \times S'$)  (b) $d\Phi(s,p) e_k = d\phi^Y_s(p) e_k$ defines a basis of the tangent space at $\phi^Y_s(p)$ of the embedded 3D submanifold  $\phi^Y_s(S')$\footnote{Indeed,   $\phi^Y_s : \bM \to \bM$ is a diffeomorphism and thus it transforms the embedded submanifold $S'\equiv \bR^3$ transverse to $Y|_{S'}$ into the embedded submanifold $\phi^Y_s(S')$ transverse to $d\phi^Y|_{S'} Y= Y|_{\phi_s^Y(S')}$, in particular, $d\phi^Y_s$ bijectively transforms the tangent spaces of these manifolds accordingly.} that is transverse to $Y(\phi^Y_s(p)) = d\phi^Y_s Y(p)$ by construction. In summary, $d\Phi(s,p)$ sends the canonical  basis of $\bR^4 = \bR\times \bR^3 \equiv \bR \times S'$ to a basis of $T_{\phi_s(p)}\bM$ and it is therefore bijective.
The situation is identical to the  one of the vector $K$, with the only difference that  we have replaced  $K$ for $Y$ and  $x^0$ for the global parameter $s$.  Now  $s,x^1,x^2,x^3$ -- where the latter three coordinates are taken on $S'$ -- define a global chart on $\bM$.  In these coordinates, the action of $\phi^Y$ is trivial $\phi_t^Y :  \bR^4 \ni (s,x^1,x^2,x^3) \mapsto (s+t, x^1,x^2,x^3) \in \bR^4$. \\ Taking advantage of the diffeomorphism  $\Phi$, one easily proves the following facts whose details are left to the reader.  If $B'_R \subset S'\equiv \bR^3$ is as before an open ball of radius $R>0$, $\Gamma_R := \Phi(\bR \times \overline{B'_R})$ is a smooth 4-manifold with orientable boundary embedded in $\bM$. The boundary $\partial \Gamma_R$ is made of the integral linees of $Y$ exiting $\partial B'_R$.  We extract from $\Gamma_R$ the compact cylinder $T_R$ with bases given by the closures of $B'_R$ and  the closure of the smooth embedded submanifold
 $$B_R := S\cap \Gamma_R = S \cap \{ \phi_t(B'_R) \:|\: t \in \bR \}$$
(which is not a ball in general now!).
 As before we  assume that  the closures of $B_R$ and $B'_R$ have no intersection and the case of non empty intersection can be treated exactly as before in the representation $\bR \times \bR^3$ where $s\in \bR$ and $\bR^3 \equiv S'$.
Notice that $B_R \subset B_{R'} \subset S$ if $R< R'$ by construction and also 
\beq \cup_{R>0} B_R = S\:.\label{UBR} \eeq
This identity is valid because, for every $p\in S$, there is $R'>0$ such that $p \in B_{R'}$. (In fact, as before, there is an integral curve of $Y$ passing through $p$ and, since this curve is timelike and inextendible, there must be $p'\in S'$ such that $\phi^Y_s(p')=p$ for some $s$ because $S'$ is Cauchy. Since the union of the sets $B'_R$ covers $S'$, then $p'\in B'_{R'}$ for some $R'>0$ sufficiently large and so  $p\in B_{R'}$.)
  $\partial T_R$ is an orientable  smooth embedded submanifold up to zero-measure sets and we can apply   the  Poincar\'e theorem as before, obtaining
$$ \int_{B'_R}  \omega^Y  = \int_{B_R}  \omega^Y \:,$$
where we have disregarded the contribution of the lateral surface because it is made of integral lines of $Y$ itself and no contribution to the boundary integral  arises according to (\ref{omegaJ}).
Taking the limit for $R\to +\infty$ on both sides, the monotone convergence theorem (taking positivity of integrands into account
according to (2) Proposition \ref{PINT}
and  (\ref{UBR})) yields
$$\int_{S'}  \omega^{K+ J}  = \int_{S}  \omega^{K+ J}\:. $$
Namely,
$$\int_{S'}  \omega^{K} + \omega^J  = \int_{S}  \omega^{K} + \omega^J \:. $$
On the ground of the validity of  (\ref{IDINT}), this identity boils down to
$$0\leq \int_{S'}  \omega^J  = \int_{S} \omega^J   \leq  +\infty\:.$$
In particular both integrals converge or diverge simultaneously. \hfill $\Box$\\

\noindent {\bf Proof of Lemma \ref{LEMMASWAP}}.  Passing to a Minkowski chart adapted to $n_\Sigma$, the right-hand side of (\ref{MM}) is a finite linear combination of integrals of the form
$$ I:= \int_{\bR^3} d^3p \overline{f(\vec{p})}\int_{\Delta}  d^3x  \int_{\bR^3} d^3q\frac{e^{i(\vec{q}-\vec{p})\cdot \vec{x}}}{(2\pi)^3}  g(\vec{q}) 
=  \int_{\bR^3} d^3p \overline{f(\vec{p})} \int_{\bR^3} d^3x  \chi_\Delta(\vec{x}) \int_{\bR^3} d^3q\frac{e^{i(\vec{q}-\vec{p})\cdot \vec{x}}}{(2\pi)^3} g(\vec{q})$$
where $f,g \in {\cal \Sigma}(\bR^3)$ and the exponentials containing $x^0$ are embodied in these functions.  Therfore, proving the thesis for $I$ is enough for ending the proof of the theorem. We have,
$$I =  \int_{\bR^3} d^3p \overline{f(\vec{p})}\int_{\bR^3}  d^3x \chi_\Delta(\vec{x}) \frac{e^{-i\vec{p} \cdot \vec{x}}}{(2\pi)^{3/2}} \int_{\bR^3} d^3q \frac{e^{i\vec{q} \cdot \vec{x}}}{(2\pi)^{3/2}} g(\vec{q})$$
$$
= \int_{\bR^3} d^3p \overline{f(\vec{p})} \int_{\bR^3} d^3x  \frac{e^{-i\vec{p} \cdot \vec{x}}}{(2\pi)^{3/2}} \chi_\Delta(\vec{x}) \hat{g}(\vec{x}) 
= \int_{\bR^3} d^3x  \overline{\hat{f}(\vec{x})}  \chi_\Delta(\vec{x})\hat{g}(\vec{x})\:. $$
Above, $\chi_\Delta \cdot  \hat{g}\in L^1(\bR^3, d^3x) \cap L^2(\bR^3, d^3x)$ because $\chi_\Delta$ is bounded  and  $ \hat{g}$ Schwartz, and in the last  identity we exploited the Plancherel theorem in $L^2$. From the very definition of Fourier transform, we finally have
$$ I= \int_{\bR^3}   d^3x \chi_\Delta(\vec{x}) \int_{\bR^3} d^3p  \int_{\bR} d^3q \frac{e^{i(\vec{q}-\vec{p})\cdot \vec{x}}}{(2\pi)^3} \overline{f(\vec{p})} g(\vec{q})  = \int_{\Delta}   d^3x \int_{\bR^3} d^3p  \int_{\bR} d^3q \frac{e^{i(\vec{q}-\vec{p})\cdot \vec{x}}}{(2\pi)^3} \overline{f(\vec{p})} g(\vec{q}) \:. $$
Last identity concludes the proof. \hfill $\Box$\\

\noindent {\bf Proof of Proposition \ref{PROPVAR}}
(a) and (b). We adopt  throughout the easier notation  $\phi := \phi_\psi$. We start by observing that the function $\bR^6 \ni (\vec{p}, \vec{q})  \mapsto F(\vec{p}, \vec{q}) :=\overline{\phi(\vec{p})}K_\sA(\vec{p}, \vec{q}) \phi(\vec{q})$ is  $\cS(\bR^6)$ in both cases (a) and (b), and 
thus $$\bR^3 \ni \vec{x} \mapsto f(\vec{x}):= \int_{\bR^3}\int_{\bR^3}  d^3p d^3q \frac{e^{i(\vec{q}-\vec{p})\cdot \vec{x}}}{(2\pi)^3} F(\vec{p}, \vec{q}) $$ is in $\cS(\bR^3)$ as it can be proved by direct inspection.  Notice that, in the case (b), 
(\ref{NPOS}) must be valid also for $\phi \in \cS(\bR^3)$ since, for a fixed $\phi \in \cS(\bR^3)$, there is a sequence of functions $C_c^\infty(\bR^3) \ni \phi_n \to \phi$ pointwise and $|\phi_n|\leq \phi$. Hence we can take the limit on both sides of (\ref{NPOS})  by using continuity of $\sA(\Delta)$ on the left-hand side and the dominated convergence theorem on the right-hand side.\\
We pass to compute the various moments of the probability measure  $\langle \psi|\sA(\cdot) \psi \rangle$ taking advantage of the theory of Fourier transformation of distributions in $\cS'(\bR^n)$.
If $\alpha:= (\alpha_1,\alpha_2,\alpha_3)$ is a multi index, so that $x^\alpha := (x^1)^{\alpha_1} (x^2)^{\alpha_2} (x^3)^{\alpha_3}$, we can write
	\begin{equation*}
		\int_{\mathbb{R}^3}x^{\alpha} \langle\psi|A(d^3x)\psi\rangle=\int_{\mathbb{R}^3}x^{\alpha} f(x)\frac{d^3x}{(2\pi)^{\frac{3}{2}}}=\left\langle\frac{1}{(2\pi)^{\frac{3}{2}}},M^{\alpha} f\right\rangle
	\end{equation*}\\
where $M^\alpha$ is the multiplicative operator with $x^\alpha$ and $f \in \cS(\bR^3)$ is defined above.
Changing variables to $\vec{u}:=\vec{q}-\vec{p}$ and $\vec{v}:=\vec{q}$ we have,
\begin{equation*}
	f(\vec{x})=\frac{1}{(2\pi)^{\frac{3}{2}}}\int_{\mathbb{R}^3} e^{i \vec{u}\cdot \vec{x}}\left(\int_{\mathbb{R}^3}F(\vec{u}-\vec{v},\vec{v})d^3v\right) d^3u
\end{equation*}
with is the  inverse Fourier transform, always in $\cS(\bR^3)$, of 
$$G(\vec{u}):=\frac{1}{(2\pi)^{\frac{3}{2}}}\int_{\mathbb{R}^3} e^{-i \vec{u}\cdot \vec{x}}f(\vec{x}) d^3x
=\int_{\mathbb{R}^3}F(\vec{u}-\vec{v},\vec{v})d^3v
\:.$$
Therefore $$\int_{\mathbb{R}^3}x^{\alpha} \langle\psi|A(d^3x)\psi\rangle=\left\langle\mathcal{F}^{-1}\delta_0, i^{|\alpha|}\mathcal{F}^{-1}\partial^{\alpha} G\right\rangle=i^{|\alpha|}\left\langle\delta_0,\partial^{\alpha}G\right\rangle=i^{|\alpha|}\partial^\alpha G(0)\:,$$
$$ \partial^\alpha G (0)=\int_{\mathbb{R}^3}\frac{\partial^{|\alpha|}}{\partial p^\alpha}F(\vec{v}-\vec{u},\vec{v})|_{\vec{u}= 0}d^3 v=\sum_{\beta\leq\alpha}\binom{\alpha}{\beta}\int_{\mathbb{R}^3}\frac{\partial^{|\alpha|-|\beta|}\overline{\phi}}{\partial p^{\alpha-\beta}}(\vec{p})\frac{\partial^{|\beta|}K_\sA}{\partial p^\beta}|_{\vec{q}=\vec{p}} \phi(\vec{p} ) d^3p.$$\\
To go on, defining $l(\vec{p},\vec{q}):=(\vec{p},\vec{p})$, since $K_\sA(\vec{p}, \vec{p})=1$, we have that $0=\partial^{\alpha}(K_\sA\circ l)=(\partial^\alpha K_\sA)\circ l$ for $|\alpha|=1$. Therefore we obtain the general expression for the $\alpha$-th moment (\ref{GENM}),
$$\int_{\mathbb{R}^3}x^{\alpha} \langle\psi|A(d^3x)\psi\rangle=\langle \psi|N^\alpha\psi\rangle+\sum_{1<|\beta|, \beta\leq \alpha}\binom{\alpha}{\beta}\int_{\mathbb{R}^3}\left(i^{|\alpha|-|\beta|}\frac{\partial^{|\alpha|-|\beta|}\overline{\phi}}{\partial p^{\alpha-\beta}}\right)(\vec{p})\left(i^{|\beta|}\frac{\partial^{|\beta|}K_\sA}{\partial p^\beta}\right)|_{\vec{q}=\vec{p}}\phi(\vec{p})d^3p.$$ 
In particular, for $|\alpha|=1$, we find (\ref{IDA})
$$\int_{\mathbb{R}^3}x^a \langle\psi|A(d^3x)\psi\rangle=\langle \psi|N^a\psi\rangle \quad a=1,2,3\:.$$
Every other selfadjoint operator $B$ in $\cH$ which satisfies the identity above must also satisfy  (from polarization and density)
$$B|_{\cD(\cH)}= N^a|_{\cD(\cH)}$$
(also replacing $\cD(\cH)$ for ${\cal S}(\cH)$). Since $\cD(\cH)$ (resp. ${\cal S}(\cH)$) is a core for $N^a$:
$$N^a = \overline{N^a|_{\cD(\cH)}}=  \overline{B|_{\cD(\cH)}} \subset \overline{B}= B\:.$$
Since  a selfadjoint operator (here $N^a$) is maximally symmetric, it must hold
$N^a = B$.
This ends the proof of (a) and (b).\\
We can pass to prove (c) and (d).
For $|\alpha|=2$, the formula above for the $\alpha$-th moment yields
$$\int_{\mathbb{R}^3}x^\alpha\langle\psi|A(d^3x)\psi\rangle=\langle \psi|N^\alpha\psi\rangle+\int_{\mathbb{R}^3}\overline{\phi}(\vec{p})\left.\left(i^{2}\frac{\partial^{|\alpha|}K_\sA}{\partial p^\alpha}\right)\right|_{\vec{q}-\vec{p}}\phi(\vec{p})d^3p.$$ 
However, taking advantage of $\frac{\partial^{|\alpha|}K_\sA}{\partial p^\alpha}\circ l=-\frac{\partial^{|\alpha|}K_\sA}{\partial q^{\alpha_1}\partial p^{\alpha_2}}\circ l$ for $|\alpha|=2$ and when only a component of $\alpha$ does not vanish,  we can rephrase the found result  as 
$$\int_{\bR^3} (x^a)^2 \langle \psi| \sA(d^3x)\psi \rangle   =\langle \psi| (N^a)^2 \psi\rangle +  \int_{\bR^3} d^3p    \overline{\phi(\vec{p})} \phi(\vec{q})  \frac{\partial }{\partial q_a}  \frac{\partial }{\partial p_a} \left. K_\sA(\vec{q}, \vec{p})\right|_{\vec{p}= \vec{q}} \:. $$
 The multiplicative operator $\sK^{\sA}_a$  defined in (\ref{kappa}), here viewed in $L^2(\bR^3, d^3p)$, defined by the function $\left. \frac{\partial }{\partial p_a}\frac{\partial }{\partial q_a} K_\sA(\vec{p}, \vec{q})\right|_{\vec{p}= \vec{q}}$ is a selfadjoint operator which is a spectral function of the four  momentum operator $P$ as this function is real  (see e.g. \cite{M}). The  domain of the operator trivially includes $C_c^\infty(\bR^3)$, but  also $\cS(\bR^3)$ when $K_\sA(\vec{p}, \vec{q})$ has polynomial growth with all of its derivatives. Evidently $\sK^{\sA}_a$  is also bounded, and thus in $\gB(\cH)$, if the function $\left. \frac{\partial }{\partial p_a}\frac{\partial }{\partial q_a} K_\sA(\vec{p}, \vec{q})\right|_{\vec{p}= \vec{q}}$ is bounded. 
Finally the operator $\sK^{\sA}_a$ is positive just because
 $$  \left. \frac{\partial }{\partial p_a}\frac{\partial }{\partial q_a} K_\sA(\vec{p}, \vec{q})\right|_{\vec{p}= \vec{q}} \geq 0\:.$$
 This can be proved as follows. Consider a real smooth mollificator (with a common compact support) such that $\psi_{n}(\vec{x}-\vec{x_0}) \to \delta(\vec{x}-\vec{x}_0)$  weakly for $n\to +\infty$.  As $K_\sA$ is positive definite, we have $$0 \leq \int_{\bR^3}  d^3p\int_{\bR^3} d^3q \partial_{p^a}\psi_{n}(\vec{p}-\vec{k})  \partial_{q^a}\psi_{n}(\vec{q}-\vec{k})K_\sA(\vec{p}, \vec{q}) \to  \left. \frac{\partial }{\partial p_a}\frac{\partial }{\partial q_a} K_\sA(\vec{p}, \vec{q})\right|_{\vec{p}= \vec{q} =\vec{k}}$$ uniformly  as a function of $\vec{k}$ because $K_\sA$ is continuous  (Proposition 4.21 \cite{Brezis}) so that the resulting function of $\vec{k}$ is non-negative as well. In summary,
 $$(\Delta_\psi x^a)^2:= \int_{\bR^3} (x^a)^2 \langle \psi| \sA(d^3x)\psi \rangle  - \left(\int_{\bR^3} x^a \langle \psi| \sA(d^3x)\psi \rangle  \right)^2$$ $$=\langle \psi | (N^a)^2 \psi \rangle -  \langle \psi |N^a \psi \rangle^2 + \langle \psi|\sK^{\sA}_{a} \psi\rangle = (\Delta_\psi N^a)^2+   \langle \psi|\sK^{\sA}_{a} \psi\rangle\:.$$ 	
  Multiplying both sides of the found identity  with $(\Delta_\psi P_a)^2$ and taking the standard Heisenberg inequality into account we get the thesis:
 $$(\Delta_\psi x_\sA^a \Delta_\psi P_a)^2 \geq \frac{\hbar^2}{4} + \hbar^2 (\Delta_\psi P_a)^2 \langle \psi|\sK^{\sA}_{a} \psi\rangle$$
 where the second factor $\hbar^2$ arises when restoring the Planck constant in the exponentials $e^{i (\vec{q}-\vec{p})\cdot {\vec x}/\hbar}$ so that 
 the representation of $N^a$ on $\cS(\bR^3)$ in the $L^2(\bR^3,d^3p)$ space is $i\hbar \frac{\partial}{\partial p_a}$.\\
(e) $\cH_N := \oplus_{j=1}^N \cH$ is the complex Hilbert space of vectors of  measurable complex valued functions $\Psi:= (\psi_1,\ldots, \psi_N)$ equipped with the Hilbert space structure arising from the Hermitian scalar product
$$\langle \Psi|\Psi' \rangle_N := \sum_{j=1}^N \int_{\bR^3} \overline{\psi_j(\vec{p})} \psi'_j(\vec{p}) d^3p\:.$$
If $\psi\in {\cal S}(\cH)$, define $\Phi_\psi := N^{-1/2}(\phi_\psi u_1, \ldots, \phi_\psi u_N)$. Each component is in $\cS(\bR^3)$ in our hypotheses. 
Finally, using $K_\sA(\vec{p},\vec{p}) =1$
\beq
\langle \Phi_\psi| \Phi_\psi \rangle_N =\frac{1}{N} \sum_{j=1}^N \int_{\bR^3} \overline{\phi_\psi(\vec{p})} \overline{u_j(\vec{p})}u_j(\vec{p}) \phi_\psi(\vec{p})d^3p = \frac{1}{N} \sum_{j=1}^N \int_{\bR^3} \overline{\phi_\psi(\vec{p})} K_\sA(\vec{p},\vec{p}) \phi_\psi(\vec{p})d^3p = \langle \psi|\psi\rangle\:.\label{SAFE}
\eeq
With the same procedure as in the proof of (c) and (d), one sees that, for every constant $c\in \bR$
$$\int_{\bR^3} (x^a- c)^2 \langle \psi| \sA(d^3x)\psi \rangle   =   \frac{1}{N}\sum_{j=1}^N \int_{\bR^3} \overline{\phi_\psi(\vec{p}) u_j(\vec{p})} (i\partial_{p_a} - c)^2 \phi_\psi(\vec{p})u_j(\vec{p})d^3p $$ $$= \langle \Phi_\psi | (i\partial_{p_a} - c)^2 \Phi_\psi\rangle_N
= \langle  (i\partial_{p_a} - c)\Phi_\psi | (i\partial_{p_a} - c) \Phi_\psi\rangle_N=
|| (i\partial_{p_a} - c) \Phi_\psi||^2\:.$$
where we have used the fact that  $i\partial_{p_a}-c$, acting componentwise,  is symmetric on $\oplus_{j=1}^N \cS(\bR^3)$.
Similarly
$$\langle \psi | (P_a-bI)^2\psi\rangle =    \frac{1}{N}\sum_{j=1}^N \int_{\bR^3} \overline{\phi_\psi(\vec{p}) u_j(\vec{p})} (p_a - b)^2 \phi_\psi(\vec{p})u_j(\vec{p})d^3p= || (p_a- d) \Phi_\psi||^2\:.$$
Where  the multiplicative operator $p_a-d$, acting componentwise, is symmetric on $\oplus_{j=1}^N \cS(\bR^3)$.
Cauchy-Schwarz inequality and (\ref{SAFE}) yield
$$|| (i\partial_{p_a} - c) \Phi_\psi||\: || (p_a- d) \Phi_\psi|| \geq
|\langle  (i\partial_{p_a} - c)\Phi_\psi |(p_a- d)\Phi_\psi \rangle_N|    = 
|\langle \Phi_\psi | (i\partial_{p_a} - c)(p_a- d)\Phi_\psi \rangle_N|$$
$$\geq | Im   \langle \Phi_\psi | (i\partial_{p_a} - c)(p_a- d) \Phi_\psi \rangle_N|= \frac{1}{2}| \langle \Phi_\psi | [i\partial_{p_a}, p_a]\Phi_\psi \rangle_N| = \frac{1}{2} \langle \Phi_\psi|\Phi_\psi \rangle_N = \frac{1}{2}\langle \psi|\psi\rangle\:.$$ Restoring $\hbar$, we have found that, for $||\psi||=1$
$$\sqrt{\int_{\bR^3} (x^a- c)^2 \langle \psi| \sA(d^3x)\psi \rangle  }\:  \sqrt{\langle \psi | (P_a-bI)^2\psi\rangle}  \geq \frac{\hbar}{2}$$
This is the thesis when replacing $c$
for $\int_{\bR^3} x^a \langle \psi| \sA(d^3x)\psi \rangle $ and 
 and $b$ for  $\langle \psi |P_a \psi \rangle$.  \hfill $\Box$\\

\noindent {\bf Proof of Theorems \ref{1MNW1} and \ref{1MNW2}}.  Those theorems are easy consequences of Proposition \ref{PROPVAR}.
 Both $\sM^n_\Sigma(\Delta)$ and $\sT^g_\Sigma(\Delta)$, for $g$  real and  smooth,  have the expression (\ref{NPOS}) if $\Sigma$ coincides with the slice at $x^0=0$ of a Minkowski chart (see (\ref{POVMC2}) and (\ref{MM}) for $x^0=0$). In details,
$$K_{\sT^g_\Sigma}(\vec{q}, \vec{p}): = \frac{(q^0+p^0)g(\vec{q}\cdot \vec{p}-q^0p^0)}{2\sqrt{q^0p^0}}$$
and
$$K_{\sM^n_{\Sigma}}(\vec{q}, \vec{p}) :=  \frac{p\cdot n  \: q\cdot n_\Sigma+ q\cdot n\: p \cdot n_\Sigma - n\cdot n_\Sigma(p\cdot  q + m^2)}{2 q\cdot n\: p\cdot n}\:. $$
These kernels are positive definite, respectively, in view of Def. \ref{DEFCP} and Thm. \ref{MAIN2}. \\ 
The kernels $K_{\sM^n_{\Sigma}}(\vec{q}, \vec{p})$, and $K_{\sT^g_{\Sigma}}(\vec{q}, \vec{p})$  when $g$ has the form (\ref{gr}) and convex combinations of these functions, have  polynomial growth with all of their derivatives, so the case (b) of the proposition applies to these cases.\\
It is not difficult to prove in particular  that the functions of $\vec{p}$ given by  $\left.\frac{\partial}{\partial q_a}\frac{\partial}{\partial p_a}K_{\sM^n_{\Sigma}}(\vec{q}, \vec{p})\right|_{\vec{p}=\vec{q}}$ 
and $\left.\frac{\partial}{\partial q_a}\frac{\partial}{\partial p_a}K_{\sT^g_{\Sigma}}(\vec{q}, \vec{p})\right|_{\vec{p}=\vec{q}}$
are  bounded (in the second case for every smooth $g$, it is not necessary the form (\ref{gr}))  and therefore the corresponding operator in (\ref{kappa}) is  a positive  selfadjoint operator in $\gB(\cH)$.   \hfill $\Box$

\section{Conditions for positive-definite kernels on $\bR^n$}
We prove here a useful property of positive definite kernels.

\begin{proposition} \label{PropK}  Let $K : \bR^n\times \bR^n \to \bC$  be continuous, the following facts are equivalent. 
\begin{itemize}
\item[(a)] Every $N\times N$ matrix $[K(k_i,k_j)]_{i,j=1,\ldots, N}$ for every choice of $k_1,\ldots, k_N \in \bR^n$ and $N=1,2,\ldots$ are positive semidefinite  (in particular Hermitian); 
\item[(b)] $K$ is positive definite according to (\ref{POSK});  
$$ \sum_{i,j=1}^N \overline{c_i}c_j K(k_i, k_j) \geq 0\:, \quad \forall \{c_j\}_{j=1,\ldots, N}\subset \bC\:,  \forall \{k_j\}_{j=1,\ldots, N}\subset X \:, \forall N=1,2,\ldots\:; \quad \quad \:\: (\ref{POSK})$$
\item[(c)] the bilinear functional induced by $K$ on $C_c(\bR^n)$ is positive:
\beq
\int_{\bR^n\times \bR^n} \overline{f(p)} K(p,q) f(q) d^npd^nq \geq 0\:, \quad \forall f \in C_c(\bR^n) \label{POSKI}\:.
\eeq
\end{itemize}
\end{proposition}

\begin{proof}  Evidently (a) and (b)  are equivalent, so we only prove that (b) and (c) are equivalent. We start by establishing  that $K$ positive definite implies (\ref{POSKI}).  Define a compactly supported continuous function on $\bR^{2n}$ as $g(p,q)=\overline{f}(p)K(p,q)f(q)$ and let $Q\subset\mathbb{R}^{2 n}$ be an open, $2n$-square such that $supp (g)\subset Q$. Then,  for any $N\in\mathbb{N}$ there exist a family of $N^{2n}$ pairwise-disjoint open $2n$-squares  $Q_i^N\subset Q$, $i=1,\ldots, N^{2n}$,  with common (Lebesgue) measure  $\left|Q_i^N\right|=\frac{\left|Q\right|}{N^{2n}}$ and
$ \overline{Q}=\bigcup_{i=1}^{N^{2n}}\overline{Q_i^{N}}$.
Choosing $x_i^N\in Q_i^N$ with $i=1,\ldots,N^{2n}$ for every given  $N$, we can define the family of  simple functions
$ g_N(x)=\sum_{i=1}^{N^{2n}}g(x_i^N)\chi_{Q_i^N}(x)\:.$
Since $diam(Q_i^N)$=$\sqrt{2 d}\frac{\left|Q\right|^{\frac{1}{2 d}}}{N}$,  uniform continuity of $g$ on the compact $\overline{Q}$ implies that for  $\epsilon>0$ there is $N_{\epsilon}$ such that
$\left|g(x_i^N)-g(x)\right|<\epsilon \text{ if }x\in Q_i^N$  and   $N>N_{\epsilon}$ and for $i=1,\ldots, N^{2n}$. Therefore, if $N>N_{\epsilon}$ and $x\in Q$, 
$\left|g_N(x)-g(x)\right|\leq\sum_{i=1}^{N^{2n}}\left|g(x_i^N)-g(x)\right|\chi_{Q_i^N}(x)<\epsilon$, so that 
 $g_N$ converges to $g$ on $Q$ uniformly and also in  $L^1$ since $|Q|< +\infty$. The proof is over because
$$\int_{\bR^n\times \bR^n}\sp\sp  \sp\sp \overline{f(p)} K(p,q) f(q) d^npd^nq = \int_{Q}gd^npd^nq = \lim_{N\to +\infty}\int_{Q}g_Nd^npd^nq $$ $$=
 \lim_{N\to +\infty} \sum_{i=1}^{N^{2n}}\overline{f}(p_i^N)K(p_i^N,q_i^N)f(q_i^N)  \frac{|Q_N|}{N^{2n}}\geq 0$$
We eventually prove that (\ref{POSKI}) entails that the continuous kernel $K$ is positive definite.
To this end, we can choose $c_1,\ldots, c_N \in \bC$, $p_1,\ldots, p_N \in \bR^n$ and a $L\in (1,+\infty)$ parametrized family of functions
$$f_L := \sum_{k=1}^N c_k L h(L (p-p_k)) \in C_c(\bR^n)$$
where  $h: \bR^n \to [0, +\infty)$ is continuous, compactly supported,  $h(0)=1$, and $\int_{\bR^n} h d^nx =1$. 
In this case, it is easy to prove that $$ \int_{\bR^n\times \bR^n} \overline{f_L(p)}  g(p,q) f_L(q)  d^npd^nq  \to \sum_{k,h=1}^N \overline{c_k} c_h g(p_k,p_h)\quad \mbox{as $L \to +\infty$}$$ for every continuous function $g: \bR^n \times \bR^n \to \bC$. Replacing  the function $f$ for $f_L$ in  (\ref{POSKI}), the limit as $L \to +\infty$ implies (\ref{POSK}), ending the proof.
\end{proof}

\noindent{\bf Competing Interests Declaration}.
All authors certify that they have no affiliations with or involvement in any organization or entity with any financial interest or non-financial interest in the subject matter or materials discussed in this manuscript.\\

\noindent{\bf Data Availability Declaration}.  All data generated or analyzed during this study are contained in this document.\\

\end{document}